\PassOptionsToPackage{unicode}{hyperref}
\PassOptionsToPackage{hyphens}{url}
\PassOptionsToPackage{dvipsnames,svgnames,x11names}{xcolor}
\documentclass[12pt]{article}
\usepackage{amsmath,amssymb}
\usepackage{paralist}
\usepackage{graphicx,psfrag,epsf,float}
\usepackage{enumerate}
\usepackage{natbib}
\usepackage{etoolbox}    
\usepackage{url} 
\usepackage[english]{babel}
\usepackage{geometry,theorem,xr,enumitem,xcolor}
\usepackage{setspace}
\usepackage{subfig}
\usepackage{makecell}
\usepackage{placeins}

\usepackage[abs]{overpic}

\newtheorem{proposition}{Proposition}
{\theorembodyfont{\rmfamily}}
\newtheorem{theorem}{Theorem}
\newtheorem{assumption}{Assumption}

\newtheorem{lemma}{Lemma}[section]
\newtheorem{suppproposition}[lemma]{Proposition}
\newtheorem{corollary}[lemma]{Corollary}

\makeatletter
\newcommand*{\addFileDependency}[1]{
  \typeout{(#1)}
  \@addtofilelist{#1}
  \IfFileExists{#1}{}{\typeout{No file #1.}}
}
\makeatother

\newcommand*{\myexternaldocument}[1]{%
    \externaldocument{#1}%
    \addFileDependency{#1.tex}%
    \addFileDependency{#1.aux}%
}

\myexternaldocument{supp} 

\usepackage{iftex}
\ifPDFTeX
  \usepackage[T1]{fontenc}
  \usepackage[utf8]{inputenc}
  \usepackage{textcomp} 
\else 
  \usepackage{unicode-math}
  \defaultfontfeatures{Scale=MatchLowercase}
  \defaultfontfeatures[\rmfamily]{Ligatures=TeX,Scale=1}
\fi
\usepackage{lmodern}
\ifPDFTeX\else  
\fi
\IfFileExists{upquote.sty}{\usepackage{upquote}}{}
\IfFileExists{microtype.sty}{
  \usepackage[]{microtype}
  \UseMicrotypeSet[protrusion]{basicmath} 
}{}
\makeatletter
\@ifundefined{KOMAClassName}{
  \IfFileExists{parskip.sty}{%
    \usepackage{parskip}
  }{
    \setlength{\parindent}{0pt}
    \setlength{\parskip}{6pt plus 2pt minus 1pt}}
}{
  \KOMAoptions{parskip=half}}
\makeatother
\usepackage{xcolor}
\makeatletter
\ifx\paragraph\undefined\else
  \let\oldparagraph\paragraph
  \renewcommand{\paragraph}{
    \@ifstar
      \xxxParagraphStar
      \xxxParagraphNoStar
  }
  \newcommand{\xxxParagraphStar}[1]{\oldparagraph*{#1}\mbox{}}
  \newcommand{\xxxParagraphNoStar}[1]{\oldparagraph{#1}\mbox{}}
\fi
\ifx\subparagraph\undefined\else
  \let\oldsubparagraph\subparagraph
  \renewcommand{\subparagraph}{
    \@ifstar
      \xxxSubParagraphStar
      \xxxSubParagraphNoStar
  }
  \newcommand{\xxxSubParagraphStar}[1]{\oldsubparagraph*{#1}\mbox{}}
  \newcommand{\xxxSubParagraphNoStar}[1]{\oldsubparagraph{#1}\mbox{}}
\fi
\makeatother

\usepackage{longtable,booktabs,array}
\usepackage{calc} 
\usepackage{etoolbox}
\makeatletter
\patchcmd\longtable{\par}{\if@noskipsec\mbox{}\fi\par}{}{}
\makeatother
\IfFileExists{footnotehyper.sty}{\usepackage{footnotehyper}}{\usepackage{footnote}}
\makesavenoteenv{longtable}
\usepackage{graphicx}
\makeatletter
\def\maxwidth{\ifdim\Gin@nat@width>\linewidth\linewidth\else\Gin@nat@width\fi}
\def\maxheight{\ifdim\Gin@nat@height>\textheight\textheight\else\Gin@nat@height\fi}
\makeatother
\setkeys{Gin}{width=\maxwidth,height=\maxheight,keepaspectratio}
\makeatletter
\def\fps@figure{htbp}
\makeatother

\makeatletter
\@ifpackageloaded{caption}{}{\usepackage{caption}}
\AtBeginDocument{%
\ifdefined\contentsname
  \renewcommand*\contentsname{Table of contents}
\else
  \newcommand\contentsname{Table of contents}
\fi
\ifdefined\listfigurename
  \renewcommand*\listfigurename{List of Figures}
\else
  \newcommand\listfigurename{List of Figures}
\fi
\ifdefined\listtablename
  \renewcommand*\listtablename{List of Tables}
\else
  \newcommand\listtablename{List of Tables}
\fi
\ifdefined\figurename
  \renewcommand*\figurename{Figure}
\else
  \newcommand\figurename{Figure}
\fi
\ifdefined\tablename
  \renewcommand*\tablename{Table}
\else
  \newcommand\tablename{Table}
\fi
}
\@ifpackageloaded{float}{}{\usepackage{float}}
\floatstyle{ruled}
\@ifundefined{c@chapter}{\newfloat{codelisting}{h}{lop}}{\newfloat{codelisting}{h}{lop}[chapter]}
\floatname{codelisting}{Listing}

\makeatother
\makeatletter
\@ifpackageloaded{caption}{}{\usepackage{caption}}
\@ifpackageloaded{subcaption}{}{\usepackage{subcaption}}
\makeatother

\ifLuaTeX
  \usepackage{selnolig}  
\fi
\usepackage[]{natbib}
\usepackage{bookmark}

\IfFileExists{xurl.sty}{\usepackage{xurl}}{} 
\hypersetup{
  pdftitle={Title},
  pdfauthor={Author 1; Author 2},
  pdfkeywords={3 to 6 keywords, that do not appear in the title},
  colorlinks=true,
  linkcolor={blue},
  filecolor={Maroon},
  citecolor={Blue},
  urlcolor={Blue},
  pdfcreator={LaTeX via pandoc}}

\newcommand{\anon}{1}

\begin{document}

\def\spacingset#1{\renewcommand{\baselinestretch}%
{#1}\small\normalsize} \spacingset{1}


\if1\anon
{
  \title{\bf Consistent Spatial Clustering of Complex Objects via Bregman Geometry}
  \author{Srijato Bhattacharyya, Huiyan Sang\hspace{.2cm}\\
    Department of Statistics, Texas A\&M University\\
    }
  \date{\phantom{\today}}
  \maketitle
} \fi

\if0\anon
{
  \bigskip
  \bigskip
  \bigskip
  \begin{center}
    {\LARGE\bf Consistent Spatial Clustering of Complex Objects via Bregman Geometry}
\end{center}
  \medskip
} \fi

\bigskip
\begin{abstract}
Spatial clustering increasingly involves complex objects, including probability distributions, networks, and structured matrices, for which conventional parametric likelihoods may be difficult to specify. We construct a generalized likelihood by exponentiating a cluster-specific Bregman loss on fixed representations of the responses, avoiding the need to specify a parametric sampling model for each response class. Each cluster is associated with an unknown representative, and the generalized likelihood measures within-cluster homogeneity through the Bregman divergence between the represented observations and that representative. A prior family matched to the Bregman geometry yields conjugate generalized posterior updates and permits exact marginalization of the cluster-specific representatives. The resulting collapsed scores combine within-cluster Bregman dispersion with uncertainty in the representatives, while a spanning-tree partition prior restricts posterior partition support to spatially contiguous clusters. This collapsed representation also leads to a tractable MCMC algorithm for posterior computation. The framework accommodates multiple complex-object response classes within a common inferential construction. We establish posterior consistency for the spatial partition and cluster-specific representatives under infill-domain asymptotics. Simulation studies with different complex-object responses demonstrate accurate recovery of spatial partitions and cluster-specific representatives. We further analyze racial-composition distributions across spatial units in Houston, illustrating geographically coherent clusters with distinct demographic compositions.
\end{abstract}

\noindent%
{\it Keywords:} Bayesian partition models, Bregman divergence, complex-object responses, generalized Bayes, posterior contraction, spatial contiguity. 
\vfill

\newpage
\spacingset{1.8} 

\section{Introduction}\label{SEC:intro}

Clustering is a fundamental tool for identifying latent heterogeneity
in data. In many modern applications, however, the observations are
complex objects, including probability distributions, functions,
networks, and positive-definite matrices, whose structural
characteristics may not be adequately captured by conventional
Euclidean clustering methods. This has motivated the development of
clustering procedures tailored to different object types, including
Wasserstein-based clustering of distribution-valued data
\citep{irpino2014dynamic}, functional data clustering
\citep{peng2008distance, jacques2014functional}, clustering of network-valued
observations \citep{mukherjee2017clustering}, and clustering of
positive-definite matrices using information divergences
\citep{stanitsas2017clustering}. While these methods provide
cluster assignments based on the characteristics of the observed
objects, they do not directly provide posterior uncertainty
quantification for the unknown partition and its cluster-specific
representatives. 

Bayesian model-based clustering provides a probabilistic approach
to inference on the unknown partition and its cluster-specific
parameters. A common approach is to assign a parametric sampling
distribution to the observations within each cluster, with Gaussian
mixture models providing a well-established example for Euclidean
responses \citep{richardson1997bayesian}. For complex-object
responses, the sampling model must instead account for the
characteristics of the corresponding object space. Existing
Bayesian approaches include Dirichlet mixture models for
compositional data \citep{piancastelli2024bayesian}, mixture models
for populations of networks \citep{mantziou2024bayesian}, and
Wishart mixture models for positive-definite matrix responses
\citep{li2026mixture}. These approaches provide posterior inference
for their respective response types, but require the specification
of a sampling distribution and associated cluster-specific
parameters for each object class. This motivates the development
of a common clustering formulation that can accommodate different
complex-object responses without requiring a separate parametric
sampling model for each class.


An alternative to specifying a parametric sampling distribution
is to construct a generalized posterior by replacing the
likelihood with an exponentiated loss function
\citep{bissiri2016general}. In the context of clustering,
\citet{rigon2023generalized} developed a generalized Bayesian
approach based on cluster-specific cohesion functions. Their approach
provides posterior inference on the partition by evaluating
within-cluster losses at empirical cluster representatives, with the
number of clusters fixed in advance. Our goal is instead to retain the
cluster-specific representatives as unknown parameters and to learn the
number of clusters jointly with the partition by assigning a prior to
the partition space. We develop a Bregman-divergence-based generalized
Bayesian framework that enables joint inference on the partition and
its cluster-specific representatives across different classes of
complex-object responses via a tractable collapsed MCMC algorithm.

To illustrate the proposed generalized Bayesian construction, we
focus on spatial clustering, where the objective is to identify spatial
subregions exhibiting distinct patterns in the observed responses.
Spatial clustering plays an important role in applications such as
disease mapping and socioeconomic analysis, where identifying
geographically coherent regions can reveal localized patterns and 
abrupt changes across regional boundaries \citep{bhattacharyya2026bayesian}. Bayesian
approaches to spatial clustering include spatial product
partition models \citep{page2016spatial}
and random graph partition models that incorporate spatial
contiguity through spanning-tree constructions
\citep{teixeira2019bayesian,luo2021bayesian}. These methods have provided
flexible ways to identify spatial clusters while accounting
for the neighborhood structure of the observations.
In many applications, however, the inferential target is not
limited to assigning cluster labels to the observed spatial
units. It is also of interest to identify contiguous subregions
of the underlying spatial domain and determine the cluster
membership of locations at which responses have not been
observed. To address this objective, we combine our generalized
Bayesian construction with the spatial domain partition prior
of \citet{huang2025consistent}, which uses spanning-tree cuts
on a spatial mesh to induce contiguous partitions of the
underlying domain.

A major contribution of this work is to establish posterior clustering
consistency under the proposed generalized Bayesian formulation.
Unlike conventional Bayesian clustering models based on a specified
sampling distribution, our generalized posterior is constructed
directly from a Bregman loss and accommodates different classes of
complex-object responses within a common theoretical formulation.
We establish that, under suitable regularity conditions, the
generalized posterior consistently recovers the correct number of spatial
regions and concentrates around the true domain partition under
infill-domain asymptotics. We further derive posterior contraction
results for the induced clustering of the observed locations,
the cluster-specific representatives, and the associated spatial
representative field. These results provide theoretical guarantees
for simultaneous recovery of the clustering structure and its
representatives without requiring a correctly specified parametric
sampling model for the observed complex objects.

We evaluate the proposed method through simulation studies involving
Euclidean, distribution-valued, count, and positive-definite matrix
responses, comparing its clustering performance with competing methods
and examining the posterior errors as the sample size increases.
We further illustrate its practical utility through an analysis of
racial-composition distributions across Census Block Groups in Houston,
Texas. The remainder of the paper is organized as follows. Section~\ref{sec-methodology} presents the generalized Bayesian model and the spatial partition
prior. Section~\ref{sec-posterior-computation} describes the posterior computation algorithm, and
Section~\ref{sec-theory} establishes the posterior contraction results. Section~\ref{sec:tau-calibration}
discusses learning-rate calibration. Section~\ref{sec-simulation} presents the simulation
studies, followed by the Houston racial-composition analysis in
Section~\ref{sec:houston}. Section~\ref{sec-conc} concludes the paper.

\section{Methodology}\label{sec-methodology}
\noindent\textbf{Data structure and general setup.}
Let $\mathcal D\subset\mathbb R^2$ denote the spatial domain, and suppose that $(s_i,Y_i)$, $i=1,\ldots,n$, are observed, where $s_i\in\mathcal D$ is a spatial location and $Y_i$ takes values in an object space $\mathcal M$. We assume a fixed measurable representation $z:\mathcal M\rightarrow\mathcal V$, where $(\mathcal V,\langle\cdot,\cdot\rangle)$ is a finite-dimensional real inner-product space, and write $X_i=z(Y_i)$. Let $\mathcal U\subseteq\mathcal V$ denote the parameter space of admissible cluster-specific representatives.

The inferential target is a finite partition of $\mathcal D$ into spatially contiguous clusters within which the represented responses share a common representative $\mu_j\in\mathcal U$. Within-cluster homogeneity is quantified by the Bregman divergence $D_\Phi(X_i,\mu_j)$, which measures the discrepancy between each represented response and its corresponding cluster-specific representative. We first develop a Bregman-divergence-based generalized likelihood–prior formulation for the represented responses and their cluster-specific representatives, then introduce a spanning-tree-based prior on spatially contiguous partitions of $\mathcal D$, and finally specialize the framework to several common classes of complex-object responses.

\subsection{Generalized Bayesian model with Bregman loss}
\label{sec-generalized-bayes-model}

\noindent\textbf{Motivation.}
For the generalized likelihood construction developed here, we choose the Bregman divergence for two reasons. First, the empirical Bregman risk of any fixed cluster whose empirical mean lies in $\mathcal U$ is uniquely minimized at that mean \citep{banerjee2005clustering}, providing a natural measure of within-cluster heterogeneity around a well-defined empirical representative. Second, exponentiating the Bregman loss yields a generalized likelihood whose dependence on the cluster representative has an exponential-family form. This permits a geometry-matched conjugate prior for the cluster-specific representatives, yielding closed-form generalized posterior updates and analytic marginalization of the representatives.

The resulting marginalization is particularly advantageous for spatial partition inference because partition updates can be performed without jointly proposing parameters for the affected clusters, thereby avoiding trans-dimensional parameter proposals that can lead to poor MCMC mixing. Such collapsed strategies play an important role in the computational efficiency of Bayesian spatial clustering methods such as \citet{luo2021bayesian, huang2025consistent, bhattacharyya2026bayesian}.

\noindent\textbf{Bregman geometry.} 
Let $\Phi:\mathcal V\rightarrow(-\infty,\infty]$ be a strictly convex
function chosen according to the object type to reflect the geometry of the represented response space, and define its effective domain by $\operatorname{dom}(\Phi)=\{x\in\mathcal V:\Phi(x)<\infty\}$.
We assume $X_i\in\operatorname{dom}(\Phi)$ for $i=1,\ldots,n$. Let
$\mathcal U$ be a nonempty open convex subset of
$\operatorname{int}\{\operatorname{dom}(\Phi)\}$ on which $\Phi$ is
continuously differentiable. For
$x\in\operatorname{dom}(\Phi)$ and $\mu\in\mathcal U$, define the Bregman divergence \citep{bregman1967relaxation} by $D_\Phi(x,\mu)
=
\Phi(x)-\Phi(\mu)
-\left\langle\nabla\Phi(\mu),x-\mu\right\rangle$, with the represented response in the first argument and the cluster-specific representative in the second. For a nonempty cluster $C\subseteq\{1,\ldots,n\}$, let
$R_C(\mu)=\sum_{i\in C}D_\Phi(X_i,\mu)$,
$n_C=|C|$, and
$\bar X_C=n_C^{-1}\sum_{i\in C}X_i$.
The Bregman decomposition of \citet[Proposition~1]{banerjee2005clustering} implies, whenever $\bar X_C\in\mathcal U$, $R_C(\mu)
=
R_C(\bar X_C)
+
n_C D_\Phi(\bar X_C,\mu),$ $
 \mu\in\mathcal U$. Hence $\bar X_C$ is the unique minimizer of $R_C$, and $W(C)
=
\min_{\mu\in\mathcal U}R_C(\mu)
=
R_C(\bar X_C)$. For a partition $\mathcal P=\{C_1,\ldots,C_k\}$ with
$\bar X_{C_j}\in\mathcal U$ for all $j$, define
$W(\mathcal P)=\sum_{j=1}^kW(C_j)$. Then $\min_{(\mu_1,\ldots,\mu_k)\in\mathcal U^k}
\sum_{j=1}^k R_{C_j}(\mu_j)
=
W(\mathcal P)$, with unique minimizer
$(\bar X_{C_1},\ldots,\bar X_{C_k})$. 



For a fixed number \(k\) of clusters, minimizing $W(\mathcal P)$ over partitions yields the hard Bregman clustering criterion of \citet{banerjee2005clustering}, which substitutes the empirical centroid $\bar X_C$ for each cluster representative. Our objective instead requires posterior inference for both the spatial partition and its cluster-specific representatives. We therefore treat the number of clusters, the partition, and the representatives as unknown, and use the cluster risks $R_C(\mu)$ to construct the generalized likelihood. The resulting generalized posterior quantifies uncertainty in both the partition and the representatives.

\noindent\textbf{Generalized likelihood and conjugate prior.}
For a nonempty cluster $C$, define the generalized likelihood
\begin{equation}
L_\tau(X_C\mid\mu_C)
=
\exp\{-\tau R_C(\mu_C)\},
\qquad \mu_C\in\mathcal U,
\label{eq-bregman-generalized-likelihood}
\end{equation}
where $X_C=\{X_i:i\in C\}$, $\mu_C$ is the representative for cluster $C$, and $\tau>0$ is the learning rate. Equation~\eqref{eq-bregman-generalized-likelihood} exponentiates the cluster-specific loss rather than specifying a sampling density for $X_C$ \citep{bissiri2016general}, and therefore need not integrate to one over the observations.

Let $S_C=\sum_{i\in C}X_i$, $H_C=\sum_{i\in C}\Phi(X_i)$, and
$A_\Phi(\mu)=\langle\nabla\Phi(\mu),\mu\rangle-\Phi(\mu)$. Since
$R_C(\mu)=H_C+n_CA_\Phi(\mu)-\langle S_C,\nabla\Phi(\mu)\rangle$,
\begin{equation}
L_\tau(X_C\mid\mu_C)
=
\exp(-\tau H_C)
\exp\left\{
\left\langle\tau S_C,\nabla\Phi(\mu_C)\right\rangle
-\tau n_C A_\Phi(\mu_C)
\right\}.
\label{eq-bregman-generalized-likelihood-dual}
\end{equation}
Hence the representative-dependent kernel depends on $\mu_C$ through
$\{\nabla\Phi(\mu_C),A_\Phi(\mu_C)\}$ and on the observations through $(S_C,n_C)$.
For an exponential-family--Bregman pair,
$p(x\mid\theta)=h(x)\exp\{\langle x,\theta\rangle-\psi(\theta)\}$,
the identities $\theta=\nabla\Phi(\mu_C)$ and $\psi(\theta)=A_\Phi(\mu_C)$
\citep{banerjee2005clustering} transform this kernel into
$\exp\{\tau\langle S_C,\theta\rangle-\tau n_C\psi(\theta)\}$.
This algebraic correspondence motivates the prior specification over the cluster-specific representatives while leaving the sampling distribution of the represented responses unspecified.

Let $\lambda_{\mathcal V}$ denote Lebesgue measure on $\mathcal V$, and define
$Z_\Phi(\chi,\kappa)=\int_{\mathcal U}\exp\{\langle\chi,\nabla\Phi(\mu)\rangle-\kappa A_\Phi(\mu)\}\lambda_{\mathcal V}(d\mu)$ and
$\mathcal H_\Phi=\{(\chi,\kappa):0<Z_\Phi(\chi,\kappa)<\infty\}$.
For $(\chi,\kappa)\in\mathcal H_\Phi$, define
\begin{equation}
\frac{d\Pi_{\chi,\kappa}}{d\lambda_{\mathcal V}}(\mu)
=
\frac{
\exp\left\{
\left\langle\chi,\nabla\Phi(\mu)\right\rangle
-\kappa A_\Phi(\mu)
\right\}
}{
Z_\Phi(\chi,\kappa)
},
\qquad \mu\in\mathcal U.
\label{eq-bregman-prior-family}
\end{equation}
This is an exponential family on $\mathcal U$ with sufficient statistic
$\{\nabla\Phi(\mu),-A_\Phi(\mu)\}$. Under the identification above, its kernel becomes
$\exp\{\langle\chi,\theta\rangle-\kappa\psi(\theta)\}$, which is the usual natural-conjugate kernel
\citep{diaconis1979conjugate}. However, \eqref{eq-bregman-prior-family} is defined directly with respect to
$\lambda_{\mathcal V}(d\mu)$. A Diaconis--Ylvisaker prior defined with respect to Lebesgue measure on the natural-parameter space and transformed through $\theta=\nabla\Phi(\mu)$ would instead induce the additional factor
$\left|\det\nabla^2\Phi(\mu)\right|$. Therefore, the two priors are not generally the same probability measure on $\mathcal U$.

Conditional on a candidate partition, we assign independent common priors $\Pi_{\chi_0,\kappa_0}$ to the cluster-specific representatives, with $(\chi_0,\kappa_0)\in\mathcal H_\Phi$. For a cluster $C$, the generalized posterior is
$\Pi_\tau(d\mu_C\mid X_C)\propto \exp\{-\tau R_C(\mu_C)\}\Pi_{\chi_0,\kappa_0}(d\mu_C)$,
with integrated score
$M_\tau(C)=\int_{\mathcal U}\exp\{-\tau R_C(\mu_C)\}\Pi_{\chi_0,\kappa_0}(d\mu_C)$. For completeness, we define $R_{\varnothing}(\mu)=0$ and hence $M_\tau(\varnothing)=1$. We call $M_\tau(C)$ an integrated generalized-Bayes score rather than a marginal likelihood because \eqref{eq-bregman-generalized-likelihood} is loss-based and need not define a normalized sampling density \citep{jewson2022general}.

\begin{theorem}[Conjugate update and collapsed Bregman score]
\label{thm-bregman-conjugacy-collapse}
Let $C$ be nonempty and suppose
$(\chi_0+\tau S_C,\kappa_0+\tau n_C)\in\mathcal H_\Phi$. Then
$\Pi_\tau(d\mu_C\mid X_C)
=
\Pi_{\chi_0+\tau S_C,\;\kappa_0+\tau n_C}(d\mu_C)$ and
$M_\tau(C)
=
\exp(-\tau H_C)
Z_\Phi(\chi_0+\tau S_C,\kappa_0+\tau n_C)/
Z_\Phi(\chi_0,\kappa_0)$.
If, in addition, $\bar X_C\in\mathcal U$, define
$J_m(u)=
\int_{\mathcal U}
\exp\{-\tau mD_\Phi(u,\mu)\}
\Pi_{\chi_0,\kappa_0}(d\mu)$ for $m\geq1$ and $u\in\mathcal U$. Then
\begin{equation}
M_\tau(C)
=
\exp\{-\tau W(C)\}
J_{n_C}(\bar X_C).
\label{eq-collapsed-risk-decomposition}
\end{equation}
\end{theorem}

\noindent\textbf{Remark.}
If $Z_\Phi(\chi,\kappa)$ admits an explicitly evaluable expression
at both the prior hyperparameters and all admissible updated
hyperparameters, then Theorem~\ref{thm-bregman-conjugacy-collapse}
yields closed-form collapsed scores $M_\tau(C)$. This permits exact
evaluation of the response-dependent factors in the
Metropolis--Hastings acceptance ratios, enabling collapsed MCMC
updates of the spatial partition without sampling the
cluster-specific representatives.

The proof of theorem~\ref{thm-bregman-conjugacy-collapse} is given in the Supplementary Material. The theorem gives the additive update $(\chi_0,\kappa_0)\mapsto(\chi_0+\tau S_C,\kappa_0+\tau n_C)$, so the conditional generalized posterior of each representative remains in the same prior family and can be sampled directly conditional on the sampled partition whenever that family admits direct simulation. It also gives the collapsed score required for partition inference. Since $M_\tau(C)$ depends on the observations only through the additive quantities $(n_C,S_C,H_C)$, a partition move requires recomputation only for the clusters it modifies.

The Bregman cohesion of \citet{rigon2023generalized} uses
$\exp\{-\lambda W(C)\}$, obtained by substituting the empirical centroid
$\bar X_C$ directly for the cluster representative. In contrast, we treat
$\mu_C$ as an unknown representative with generalized posterior
$\Pi_\tau(d\mu_C\mid X_C)$ and integrate it when computing $M_\tau(C)$.
By \eqref{eq-collapsed-risk-decomposition}, this gives
$M_\tau(C)=\exp\{-\tau W(C)\}J_{n_C}(\bar X_C)$, where the first factor is
the exponentiated minimum empirical Bregman risk and the second arises
from marginalizing the representative. This distinction is important because the representatives are inferential targets, and integrating them allows their uncertainty to enter posterior comparison across partitions. In addition, whereas \citet{rigon2023generalized} condition on a prespecified number of clusters, our spanning-tree partition prior in Section~\ref{sec-spatial-partition-prior} assigns positive probability to partitions with different numbers of spatially contiguous clusters, allowing the number of clusters to be inferred jointly with the partition.


The factor $J_{n_C}(\bar X_C)$ integrates the Bregman kernel centered at the empirical representative $\bar X_C$ with respect to the prior distribution of $\mu$. The integrand attains its maximum value one at $\mu=\bar X_C$ and decreases as $D_\Phi(\bar X_C,\mu)$ increases. Consequently, $J_{n_C}(\bar X_C)$ is larger when the prior assigns greater probability to representative values that are close to $\bar X_C$ in Bregman divergence, and smaller when the prior places its mass farther from $\bar X_C$. The factor $n_C\tau$ controls the concentration of this weighting around $\bar X_C$, so larger clusters place progressively greater weight on representatives lying close to their empirical representative. Thus the collapsed score combines the minimum within-cluster Bregman dispersion $W(C)$ with the prior-weighted support for representative values near $\bar X_C$.

For a candidate partition
$\mathcal P=\{C_1,\ldots,C_k\}$, conditional independence of the
representatives gives
$M_\tau(\mathcal P)=\prod_{j=1}^kM_\tau(C_j)$. If
$\bar X_{C_j}\in\mathcal U$ for all $j$, then
\begin{equation}
\log M_\tau(\mathcal P)
=
-\tau W(\mathcal P)
+
\sum_{j=1}^k
\log J_{n_{C_j}}(\bar X_{C_j}).
\label{eq-partition-score-decomposition}
\end{equation}
The order of the integration term follows from Laplace approximation.
Let $d_{\mathcal V}=\dim(\mathcal V)$ and
$q_0=d\Pi_{\chi_0,\kappa_0}/d\lambda_{\mathcal V}$. Fix $u\in\mathcal U$
and suppose that $\Phi$ is twice continuously differentiable near $u$,
$\nabla^2\Phi(u)$ is positive definite, $q_0$ is continuous and positive
at $u$, and there exist $\varepsilon,\delta>0$ such that the closed
$\varepsilon$-ball around $u$ is contained in $\mathcal U$ and
$\inf_{\mu\in\mathcal U:\|\mu-u\|\geq\varepsilon}
D_\Phi(u,\mu)\geq\delta$. Then, as $m\to\infty$,
$$
\begin{aligned}
J_m(u)
&=
q_0(u)
\left(\frac{2\pi}{\tau m}\right)^{d_{\mathcal V}/2}
\left|\nabla^2\Phi(u)\right|^{-1/2}
\{1+o(1)\},\\
\log J_m(u)
&=
-\frac{d_{\mathcal V}}{2}\log m
+\log q_0(u)
-\frac12\log|\nabla^2\Phi(u)|
+\frac{d_{\mathcal V}}{2}\log\left(\frac{2\pi}{\tau}\right)
+o(1).
\end{aligned}
$$
Therefore integrating a $d_{\mathcal V}$-dimensional representative
contributes the leading term
$-(d_{\mathcal V}/2)\log m$ to the cluster log score. Equivalently,
$J_m(u)$ has leading order $m^{-d_{\mathcal V}/2}$, matching the usual
dimension penalty obtained by integrating a regular
$d_{\mathcal V}$-dimensional Bayesian parameter
\citep{schwarz1978estimating,tierney1986accurate}, although here it arises
from integration of the exponentiated Bregman loss rather than an ordinary
sampling likelihood.

\subsection{Spatial partition prior}
\label{sec-spatial-partition-prior}

We impose spatial contiguity through a spanning-tree partition prior.
Spanning-tree cuts have been used to construct Bayesian priors on
contiguous graph partitions
\citep{teixeira2019bayesian,luo2021bayesian,come2024bayesian}.
We use the block-based domain construction of
\citet{huang2025consistent}. Without loss of generality, take $\mathcal D=[0,1]^2$ and partition it
into $K^2$ square blocks
$\mathcal B_K=\{B_{K,1},\ldots,B_{K,K^2}\}$ of side length $K^{-1}$. As noted by \citet{huang2025consistent}, this construction extends to more general domains that are homeomorphic to $[0,1]^2$ through a bi-Lipschitz homeomorphism.
Let $G_K=(\mathcal B_K,E_K)$ be the mesh graph in which two blocks are
adjacent when they share a boundary segment, and let
$\mathfrak T_K$ denote its set of spanning trees. We assign the discrete
uniform spanning-tree prior
$p_K(T)=|\mathfrak T_K|^{-1},$ $ T\in\mathfrak T_K$, as in \citet{teixeira2019bayesian} and
\citet{bhattacharyya2026bayesian}.


Conditional on $T\in\mathfrak T_K$, assign the number of clusters the truncated Poisson prior
$\Pr(k\mid T)
=
\frac{\lambda_n^k/k!}
{\displaystyle\sum_{h=1}^{k_{\max}}\lambda_n^h/h!},$ $
k\in\{1,\ldots,k_{\max}\}$ where \(\lambda_n>0\) and \(1\le k_{\max}\le K^2\).
Given $(T,k)$, let $F\subset E(T)$ denote the set of deleted edges. Since every spanning tree on the $K^2$ mesh vertices has $K^2-1$ edges, we take
$\Pr(F\mid T,k)
=
\binom{K^2-1}{k-1}^{-1}
\mathbf 1\!\left\{
F\subset E(T),\ |F|=k-1
\right\}.$
Deleting $F$ from $T$ produces exactly $k$ connected components
$\mathcal V_1(T,F),\ldots,\mathcal V_k(T,F)$. These components induce the domain partition
$\pi(T,F)
=
\{D_1(T,F),\ldots,D_k(T,F)\},$ $
D_j(T,F)
=
\bigcup_{B\in\mathcal V_j(T,F)}B.$
Because each $\mathcal V_j(T,F)$ is connected in $T$ and $E(T)\subset E_K$, every $D_j(T,F)$ is contiguous under the mesh adjacency.

Let $\mathfrak T_K(\pi)$ denote the set of spanning trees for which deletion of the $|\pi|-1$ edges separating the cells of $\pi$ produces $\pi$. Marginalizing over the tree-and-cut representation gives the induced prior on supported domain partitions,
$\Pi_K(\pi)
=
\frac{\lambda_n^{|\pi|}/|\pi|!}
{\displaystyle\sum_{h=1}^{k_{\max}}\lambda_n^h/h!}
\binom{K^2-1}{|\pi|-1}^{-1}
\sum_{T\in\mathfrak T_K(\pi)}p_K(T),$ $
\pi\in\mathfrak P_K.$

A domain partition $\pi=\{D_1,\ldots,D_k\}$ induces observation clusters
$C_j(\pi)=\{i:s_i\in D_j\}$, so observations falling in the same mesh
block necessarily receive the same cluster label. Let
$\mathfrak P_K$ denote the support of $\Pi_K$. Since the
cluster-specific representatives have already been integrated out, the
generalized posterior distribution of the spatial partition is
\begin{equation}
\Pi_{\tau,n}(\pi\mid X_{1:n},s_{1:n})
=
\frac{
\Pi_K(\pi)
\displaystyle\prod_{j=1}^{|\pi|}
M_\tau\{C_j(\pi)\}
}{
\displaystyle
\sum_{\widetilde\pi\in\mathfrak P_K}
\Pi_K(\widetilde\pi)
\prod_{j=1}^{|\widetilde\pi|}
M_\tau\{C_j(\widetilde\pi)\}
},
\qquad
\pi\in\mathfrak P_K.
\label{eq-spatial-partition-posterior}
\end{equation}

\subsection{Canonical response geometries.}\label{sec-examples}

The specification $(\mathcal M,z,\mathcal V,\mathcal U,\Phi)$ determines the Bregman loss, the geometry-matched prior in \eqref{eq-bregman-prior-family}, its generalized posterior update, and the collapsed cluster score. For Euclidean responses, $\mathcal M=\mathcal V=\mathcal U=\mathbb R^d$, $z$ is the identity map, and $\Phi(x)=\|x\|^2/2$, so $D_\Phi(x,\mu)=\|x-\mu\|^2/2$ and $\Pi_{\chi,\kappa}=N_d(\chi/\kappa,\kappa^{-1}I_d)$. The cluster-specific representative therefore has a Gaussian generalized posterior, and its contribution can be integrated analytically in the collapsed cluster score.

\medskip
\noindent\emph{Distribution-valued responses.}
Let $\mathcal M=\Delta_{d-1}$, where $\Delta_{d-1}=\{p\in[0,1]^d:\sum_{h=1}^d p_h=1\}$. We use the intrinsic representation $z(p)=(p_1,\ldots,p_{d-1})^\top$, with $\mathcal V=\mathbb R^{d-1}$ and $\mathcal U=\{u\in(0,1)^{d-1}:\sum_{h=1}^{d-1}u_h<1\}$. Take $\Phi$ to be the negative-entropy generator on $\mathcal V$, with effective domain $\operatorname{dom}(\Phi)=\{x\in[0,1]^{d-1}:\sum_{h=1}^{d-1}x_h\leq 1\}$. For $x\in\operatorname{dom}(\Phi)$, write $x_d=1-\sum_{h=1}^{d-1}x_h$ and set $\Phi(x)=\sum_{h=1}^d x_h\log x_h$, with $0\log0=0$. For $\mu\in\mathcal U$, write $\mu_d=1-\sum_{h=1}^{d-1}\mu_h$. Then $\nabla_h\Phi(\mu)=\log(\mu_h/\mu_d)$, $h=1,\ldots,d-1$, and $A_\Phi(\mu)=-\log\mu_d$. Hence $D_\Phi\{z(p),\mu\}=\sum_{h=1}^d p_h\log\{p_h/\mu_h\}=\operatorname{KL}(p\,\|\,\mu)$. For $(\chi,\kappa)\in\mathcal H_\Phi$, define $\alpha_h=\chi_h+1$, $h=1,\ldots,d-1$, and $\alpha_d=\kappa-\sum_{h=1}^{d-1}\chi_h+1$. Identifying $\mu\in\mathcal U$ with the full simplex vector $(\mu_1,\ldots,\mu_d)^\top$, whenever $\alpha_h>0$ for all $h$, $\Pi_{\chi,\kappa}=\operatorname{Dirichlet}(\alpha)$ and $Z_\Phi(\chi,\kappa)=B(\alpha)$. Writing $
\boldsymbol{\alpha}_0
$ for the prior Dirichlet parameter, the generalized posterior is $\mu\mid X_C\sim\operatorname{Dirichlet}\left(
\boldsymbol{\alpha}_0
+\tau\sum_{i\in C}p_i\right)$, and the collapsed score is $M_\tau(C)=\exp\left\{-\tau\sum_{i\in C}\sum_{h=1}^d p_{ih}\log p_{ih}\right\}B\left(
\boldsymbol{\alpha}_0
+\tau\sum_{i\in C}p_i\right)/B(
\boldsymbol{\alpha}_0
)$.

\medskip
\noindent\emph{Positive-definite matrix responses.}
Let $\mathcal M=\mathcal U=\mathbb S_{++}^q$, let $\mathcal V=\mathbb S^q$ be the vector space of real symmetric $q\times q$ matrices equipped with the trace inner product $\langle A,B\rangle=\operatorname{tr}(AB)$, and take $z$ to be the identity map. Take $\Phi$ to be the negative log-determinant generator on $\mathcal V$, with effective domain $\operatorname{dom}(\Phi)=\mathbb S_{++}^q$ and $\Phi(X)=-\log|X|$ for $X\in\mathbb S_{++}^q$. We then have $\nabla\Phi(M)=-M^{-1}$ and $A_\Phi(M)=\log|M|-q$. The resulting Bregman divergence is $D_\Phi(X,M)=\operatorname{tr}(M^{-1}X)-\log|M^{-1}X|-q$, which is the LogDet Bregman divergence generated by $-\log|\cdot|$ \citep{kulis2009low}. Under the parameterization $\chi=\Psi/2$ and $\kappa=(\nu+q+1)/2$, the prior family in \eqref{eq-bregman-prior-family} is inverse-Wishart. Specifically, if $\nu>q-1$ and $\Psi\in\mathbb S_{++}^q$, then $\Pi_{\chi,\kappa}=\operatorname{IW}_q(\nu,\Psi)$. Let $c_q(\nu,\Psi)=|\Psi|^{\nu/2}\{2^{\nu q/2}\Gamma_q(\nu/2)\}^{-1}$ denote the inverse-Wishart normalizing constant. The corresponding Bregman normalizer is $Z_\Phi(\chi,\kappa)=\exp(\kappa q)/c_q(\nu,\Psi)$. Therefore, under the prior $M\sim\operatorname{IW}_q(\nu_0,\Psi_0)$, the generalized posterior is $M\mid X_C\sim\operatorname{IW}_q\left(\nu_0+2\tau n_C,\Psi_0+2\tau S_C\right)$. The collapsed cluster score is $M_\tau(C)=\exp\left\{\tau\sum_{i\in C}\log|X_i|+\tau n_Cq\right\}c_q(\nu_0,\Psi_0)/c_q\left(\nu_0+2\tau n_C,\Psi_0+2\tau S_C\right)$.


Corresponding derivations for nonnegative count vectors and network-valued responses are provided in the Supplementary Material. Together, the Euclidean, distribution-valued, count, network, and positive-definite matrix constructions show that the same generalized-Bayes formulation yields response-specific Bregman losses, geometry-matched representative priors, closed-form generalized posterior updates, and analytically evaluable collapsed cluster scores.

\section{Posterior Computation}
\label{sec-posterior-computation}

We sample the spatial partition through the spanning-tree representation
used to define the prior in Section~\ref{sec-spatial-partition-prior}.
Let $T\in\mathfrak T_K$ denote the current spanning tree and let
$F\subset E(T)$ contain its deleted edges. Then
$k=|F|+1$, and $\pi(T,F)$ denotes the induced domain partition. We retain
$(T,F)$ in the Markov state because direct evaluation of the induced
partition prior $\Pi_K(\pi)$ requires summation over all spanning trees
that can generate $\pi$. The cluster-specific representatives are instead
integrated out using
Theorem~\ref{thm-bregman-conjugacy-collapse}. The resulting target
distribution on $(T,F)$ has unnormalized mass
\begin{equation}
\widetilde\Pi_{\tau,n}(T,F\mid X_{1:n},s_{1:n})
\propto
p_K(T)
\frac{\lambda_n^k}{k!}
\binom{K^2-1}{k-1}^{-1}
\prod_{j=1}^{k} M_\tau\{C_j(\pi(T,F))\},
\qquad k=|F|+1 .
\label{eq-collapsed-tree-cut-posterior}
\end{equation}
Summing \eqref{eq-collapsed-tree-cut-posterior} over all tree--cut states
that induce the same partition $\pi$ recovers the generalized posterior
in \eqref{eq-spatial-partition-posterior}. 

Posterior sampling follows the spanning-tree partition samplers of
\citet{luo2021bayesian} and \citet{huang2025consistent}. Birth moves
delete one currently uncut tree edge and increase the number of clusters
by one. Death moves restore one deleted edge and decrease the number of
clusters by one. Change moves restore one deleted edge and then delete one edge that is uncut in the resulting intermediate tree, leaving the number of clusters unchanged. We additionally
refresh the spanning tree conditional on the current partition while
preserving that partition using an exact sampler derived in \citet{bhattacharyya2026bayesian}. Every Metropolis--Hastings ratio for the
birth, death, and change moves is calculated from
\eqref{eq-collapsed-tree-cut-posterior}. Since the response factor in the MCMC acceptance ratio is a
product of the cluster-specific scores $M_\tau(C)$, factors associated
with clusters whose memberships are unchanged cancel from these ratios.
A proposed partition modification therefore requires evaluation of
$M_\tau(C)$ only for the clusters created, removed, or modified by that
move. The complete proposal distributions, acceptance probabilities, and
tree-hyper construction are provided
in the Supplementary Material.

After a partition $\pi=\{D_1,\ldots,D_k\}$ has been obtained, let
$C_j(\pi)=\{i:s_i\in D_j\}$. Theorem~\ref{thm-bregman-conjugacy-collapse}
gives independent conditional draws
$\mu_j\mid\pi,X_{1:n}\sim
\Pi_{\chi_0+\tau S_{C_j(\pi)},\,
\kappa_0+\tau n_{C_j(\pi)}}$ for $j=1,\ldots,k$.
Posterior summaries of the spatial partition are therefore computed from
the sampled tree--cut states, while posterior summaries of the
cluster-specific representatives are obtained directly from their exact
conditional posterior distributions.

\section{Posterior Contraction Theory}
\label{sec-theory}


Quasi-posteriors formed from general non-likelihood statistical
criterion functions were studied by \citet{chernozhukov2003mcmc}, who
established large-sample concentration and asymptotic normality under
regularity conditions. More recently, \citet{miller2021asymptotic}
provided general sufficient conditions for concentration and asymptotic
normality of generalized posteriors. These general results, however, do
not directly address posterior contraction over a growing
mesh-based spatial partition. In the spatial clustering setting,
\citet{huang2025consistent} developed posterior contraction theory for a
spanning-tree domain partition model under infill asymptotics,
establishing recovery of the number and configuration of spatial
regions together with contraction of the associated regional
parameters and prediction error. Their theoretical analysis is
developed under a Gaussian spatially clustered linear regression
model. Our objective is to extend this contraction theory to the
substantially broader Bregman generalized-Bayes setting considered
here. This extension is
nontrivial, since the Gaussian integrated-likelihood arguments used in
the existing spatial contraction theory are no longer available and
must be replaced by uniform control of the collapsed Bregman
generalized-Bayes score. We show that, under suitable regularity
conditions on the Bregman geometry, the
resulting generalized posterior retains the same mesh-driven
contraction rate \citep{huang2025consistent} for the spatial partition, while also yielding
contraction of the corresponding cluster-specific Bregman
representatives.

\noindent\textbf{Theoretical framework and assumptions.}
We study posterior contraction under infill asymptotics on
$\mathcal D=[0,1]^2$. Let
$\pi_0=\{D_{1,0},\ldots,D_{k_0,0}\}$ denote the true spatial partition, where
$k_0<\infty$ is fixed, and let
$\mathcal B_0=\bigcup_{\ell\neq h}(D_{\ell,0}\cap D_{h,0})$ denote its
boundary set. For $s\in D_{\ell,0}\setminus\mathcal B_0$, write
$\ell_0(s)=\ell$, using any fixed deterministic convention on
$\mathcal B_0$. The transformed responses are $X_i=z(Y_i)\in\mathcal V$. Let $P_0$ denote the true data-generating law and let $E_0$ denote expectation under $P_0$. Using the deterministic boundary convention above, define $\mu_0(s)=\mu_{\ell_0(s),0}$ for all $s\in\mathcal D$. Also define $C_{\ell,0}=\{i:\ell_0(s_i)=\ell\}$ and $\mathcal P_0=\{C_{1,0},\ldots,C_{k_0,0}\}$.

\begin{assumption}[Spatial design and partition prior]
\label{assump-spatial}
\begin{compactenum}
\item The locations $s_1,\ldots,s_n$ are independently sampled from
$\mathcal D$ with density $p_D$ satisfying
$0<p_{\min}\leq p_D(s)\leq p_{\max}<\infty$ for every $s\in\mathcal D$.

\item There exists $C_{\mathcal B}<\infty$ such that
$N(\mathcal B_0,\nu,\|\cdot\|)\leq C_{\mathcal B}\nu^{-1}$
for all sufficiently small $\nu>0$. Moreover, for each $\ell=1,\ldots,k_0$, there exists $c_\ell>0$
such that any $s,s'\in D_{\ell,0}$ can be joined by a path
$P_\ell(s,s')\subseteq D_{\ell,0}$ satisfying
$\operatorname{dist}\{P_\ell(s,s'),\mathcal B_0\}
\geq
\min\{\operatorname{dist}(s,\mathcal B_0),
       \operatorname{dist}(s',\mathcal B_0),c_\ell\}.$

\item The true regions have nonvanishing area,
$\min_{1\leq\ell\leq k_0}|D_{\ell,0}|\geq a_{\min}>0$. For fixed
$\alpha_b>0$ and $0<\alpha_p<\alpha_b$, let
$L_n=\log^{1+\alpha_b}(n)$ and assume
$K\asymp(n/L_n)^{1/2},$ $
\lambda_n=o(1),$ $
\log(\lambda_n^{-1})\asymp n\log^{-\alpha_p}(n),$
so that $K^2L_n\asymp n$.

\item The partition prior is supported on at most a fixed
$k_{\max}$ clusters, where $k_0\leq k_{\max}<\infty$. The spanning-tree
prior is the discrete uniform distribution on $\mathfrak T_K$, as specified
in Section~\ref{sec-methodology}.
\end{compactenum}
\end{assumption}

Assumption~\ref{assump-spatial}(1)--(3) uses the spatial sampling,
boundary-regularity, connectivity, and resolution regime underlying the
spanning-tree contraction theory of \citet{huang2025consistent}. The nonvanishing-area condition rules out asymptotically negligible true
regions and ensures that each region is represented by a nonvanishing
proportion of mesh blocks as the spatial resolution increases. The uniform tree prior in
Assumption~\ref{assump-spatial}(iv) is a specialization of the spanning-tree
prior condition in \citet{huang2025consistent}. In particular, under our prior model,
$p_K(T_1)/p_K(T_2)=1$, so their required tree-prior ratio bound holds. Thus the spatial approximation, block-occupancy, counting,
and good-partition geometric arguments are inherited from their framework,
while the response-side assumption below is different.

\begin{assumption}[Bregman representation and response regularity]
\label{assump-bregman-response}
\begin{compactenum}
\item The inner-product space $\mathcal V$ has fixed intrinsic dimension
$r=\dim(\mathcal V)<\infty$, independent of $n$, and
$X_i\in\operatorname{dom}(\Phi)$ almost surely. The generator
$\Phi:\mathcal V\to(-\infty,\infty]$ is proper, closed, and strictly convex,
and there exists a nonempty open convex set
$\mathcal U\subseteq\operatorname{int}_{\mathcal V}\{\operatorname{dom}(\Phi)\}$
such that $\Phi\in C^2(\mathcal U)$.

\item There exist fixed vectors
$\mu_{1,0},\ldots,\mu_{k_0,0}\in\mathcal U$ such that, conditional on
$s_{1:n}$,
$E_0(X_i\mid s_{1:n})=\mu_{\ell_0(s_i),0}$ for
$i=1,\ldots,n$. Writing
$\mu_0(s)=\mu_{\ell,0}$ for
$s\in D_{\ell,0}\setminus\mathcal B_0$, the distinct regional centers satisfy
$\min_{\ell<h}\|\mu_{\ell,0}-\mu_{h,0}\|\geq\delta_0$ for some
$\delta_0>0$.

\item Let
$\mathcal H_0=\operatorname{conv}\{\mu_{1,0},\ldots,\mu_{k_0,0}\}$.
There exist $\eta,\rho>0$ such that
$\mathcal K=\{u:\operatorname{dist}(u,\mathcal H_0)\leq\eta\}$ and
$\mathcal K_\rho=\{u:\operatorname{dist}(u,\mathcal K)\leq\rho\}$ satisfy
$\mathcal K_\rho\subset\mathcal U$. For constants
$0<h_\Phi\leq\overline h_\Phi<\infty$,
\begin{align*}
h_\Phi I_r\preceq\nabla^2\Phi(u)\preceq\overline h_\Phi I_r \quad
(u\in\mathcal K_\rho),\quad
\inf_{\substack{u\in\mathcal K,\ \mu\in\mathcal U\\
\|\mu-u\|\geq\rho}}
D_\Phi(u,\mu)>0.
\label{eq-bregman-geometry-theory}
\end{align*}

\item Conditional on $s_{1:n}$, the transformed responses
$X_1,\ldots,X_n$ are independent. With
$\epsilon_i=X_i-\mu_0(s_i)$, there exist constants
$\sigma_X^2<\infty$ and $0<b_X<\infty$ such that, for every unit
$v\in\mathcal V$, every $i=1,\ldots,n$, and every
$t\in\mathbb R$ satisfying $|t|\leq b_X^{-1}$,
$E_0\!\left[
\exp\{t\langle v,\epsilon_i\rangle\}
\,\middle|\,s_{1:n}
\right]
\leq
\exp(\sigma_X^2t^2/2).$
\end{compactenum}
\end{assumption}

The local exponential-moment condition in
Assumption~\ref{assump-bregman-response} is weaker than uniform
sub-Gaussianity. In particular, if the centered represented responses
$\epsilon_i=X_i-\mu_0(s_i)$ are uniformly sub-Gaussian conditional on
$s_{1:n}$, so that for some $\sigma_X^2<\infty$ the same quadratic
moment-generating-function bound holds for every $t\in\mathbb R$, then
the stated condition holds automatically on $|t|\leq b_X^{-1}$ for any
fixed $b_X>0$. Thus uniform sub-Gaussianity is a special case of the
response-tail condition used in the contraction theory. Assumption~\ref{assump-bregman-response} replaces the Gaussian response
conditions of \citet{huang2025consistent} with conditions on the
represented responses appropriate for the Bregman formulation. The
parameters $\mu_{\ell,0}$ define the piecewise-constant conditional mean
structure in the representation space. The separation condition requires distinct true regions to have distinct representation-space means. Regions having the same mean are not identifiable under this target and would therefore need to be regarded as a single true region. When
$E_0\{|\Phi(X_i)|\mid s_{1:n}\}<\infty$, the conditional Bregman risk is
also well defined, and
$\mu_{\ell_0(s_i),0}$ is the unique minimizer of
$\mu\mapsto E_0\{D_\Phi(X_i,\mu)\mid s_{1:n}\}$.

The geometric conditions are immediate for the Euclidean quadratic
generator and, for a fixed functional basis, follow from positive
definiteness of its Gram matrix. For probability-vector responses they
hold on compact subsets of the simplex interior, and boundedness of the
simplex gives the exponential-moment condition in
Assumption~\ref{assump-bregman-response}. For positive-definite matrix
responses, the LogDet geometry satisfies the curvature conditions on
compact subsets of $\mathbb S_{++}^q$. If the matrix responses are conditionally Wishart with fixed matrix dimension $q$, fixed degrees of freedom $\nu>q-1$, and regional mean matrices lying in the compact set specified in Assumption~\ref{assump-bregman-response}, then the responses are almost surely positive definite and their centered directional moment-generating functions satisfy the required quadratic bound on a common neighborhood of zero. The generalized Kullback--Leibler geometry
for nonnegative counts likewise satisfies the geometric conditions on
compact subsets of $(0,\infty)^d$. In particular, conditionally Poisson
responses with regional mean vectors in this compact set satisfy the same
local exponential-moment condition. These distributional specifications
are sufficient conditions for the contraction theory and are not imposed
by the generalized-Bayes model itself.

\begin{assumption}[Matching prior and learning rate]
\label{assump-representative-prior}
\begin{compactenum}
\item The generalized-Bayes learning rate is fixed and satisfies
$0<\tau<\infty$.

\item The hyperparameters $\chi_0\in\mathcal V$ and
$\kappa_0\in\mathbb R$ are fixed independently of $n$ and satisfy
$0<Z_\Phi(\chi_0,\kappa_0)<\infty$. Conditional on any supported partition,
the cluster-specific representatives are independent with common prior
$\Pi_{\chi_0,\kappa_0}$ defined in \eqref{eq-bregman-prior-family}.
\end{compactenum}
\end{assumption}

Assumption~\ref{assump-representative-prior} concerns only the new
generalized-Bayes response component. Under
Assumptions~\ref{assump-bregman-response} and
\ref{assump-representative-prior}, the matching-prior density is continuous,
strictly positive, and bounded above and away from zero on the compact set
$\mathcal K_\rho$. No additional local prior-mass assumption is therefore
needed. In contrast, \citet{huang2025consistent} use their Gaussian
response specification to obtain the response-side likelihood-ratio bounds;
here that role is played by the collapsed Bregman score, while their spatial
prior and combinatorial arguments are retained.

For any fixed $\alpha_0>0$, define
\begin{align*}
q_n=\lfloor K\log^{\alpha_0}(n)\rfloor,\qquad
\rho_n=\frac{\log^{\alpha_0}(n)}{K}
\asymp n^{-1/2}\log^{\alpha_0+(1+\alpha_b)/2}(n),\qquad
\zeta_n=q_nL_n\asymp n\rho_n.
\end{align*}
In particular,
$n\rho_n^2\asymp\log^{2\alpha_0+1+\alpha_b}(n)\to\infty$.

\noindent\textbf{Collapsed-score control.}
Recall from \eqref{eq-collapsed-risk-decomposition} that
$M_\tau(C)=\exp\{-\tau W(C)\}J_{|C|}(\bar X_C)$ whenever
$\bar X_C\in\mathcal U$. The following uniform bound is the response-side
reduction used throughout the contraction proof.

\begin{proposition}[Uniform integrated-center control]
\label{prop-uniform-integrated-center}
Under Assumptions~\ref{assump-bregman-response} and
\ref{assump-representative-prior}, there exist constants
$0<c_J\leq C_J<\infty$ such that, uniformly over $m\geq1$ and
$u\in\mathcal K$,
$c_Jm^{-r/2}\leq J_m(u)\leq C_Jm^{-r/2}$.
Consequently, if
$\mathcal P=\{C_1,\ldots,C_k\}$ and
$\widetilde{\mathcal P}=\{A_1,\ldots,A_h\}$ have nonempty cells,
$k,h\leq k_{\max}$, and all their empirical cell centroids belong to
$\mathcal K$, then, uniformly over such pairs,
\begin{equation}
\log\frac{M_\tau(\mathcal P)}
{M_\tau(\widetilde{\mathcal P})}
=
-\tau\{W(\mathcal P)-W(\widetilde{\mathcal P})\}
+O(\log n).
\label{eq-uniform-score-ratio}
\end{equation}
\end{proposition}


The factor $J_m(u)$ reflects the uncertainty associated with the
cluster-specific representative after the minimum empirical Bregman risk
has been attained at $u$. Its magnitude depends on the local Bregman
geometry around $u$, the prior distribution of the representative, and
the dimension $r$ of the representation space. The first conclusion of
Proposition~\ref{prop-uniform-integrated-center} shows that this
integration contributes a term of order $-(r/2)\log m$ to the
cluster-specific log score. The second conclusion gives its consequence
for partition comparison. Thus Proposition~\ref{prop-uniform-integrated-center} separates the collapsed log-score comparison into the empirical Bregman-risk difference and an integration term that is uniformly $O(\log n)$ over the stated class of partitions. The contraction argument below combines this decomposition with class-specific bounds on the empirical Bregman-risk differences and the spatial partition prior.

\noindent\textbf{Partition contraction.}
For domain partitions
$\pi=\{D_1,\ldots,D_k\}$ and
$\widetilde\pi=\{\widetilde D_1,\ldots,\widetilde D_{\widetilde k}\}$,
and sampled-location partitions
$\mathcal P=\{C_1,\ldots,C_k\}$ and
$\widetilde{\mathcal P}
=\{\widetilde C_1,\ldots,\widetilde C_{\widetilde k}\}$, define
\begin{equation}
\begin{aligned}
\epsilon_1(\pi,\widetilde\pi)
&=1-\sum_{j=1}^k\max_{\ell\leq\widetilde k}
|D_j\cap\widetilde D_\ell|,
&
\epsilon_{n,1}(\mathcal P,\widetilde{\mathcal P})
&=1-\frac1n\sum_{j=1}^k\max_{\ell\leq\widetilde k}
|C_j\cap\widetilde C_\ell|,
\end{aligned}
\label{eq-partition-directed-distances}
\end{equation}
and let
$\epsilon(\pi,\widetilde\pi)
=\epsilon_1(\pi,\widetilde\pi)+\epsilon_1(\widetilde\pi,\pi)$ and
$\epsilon_n(\mathcal P,\widetilde{\mathcal P})
=\epsilon_{n,1}(\mathcal P,\widetilde{\mathcal P})
+\epsilon_{n,1}(\widetilde{\mathcal P},\mathcal P)$.
These label-invariant partition metrics are those used by
\citet{huang2025consistent}.

Let
$\pi_{0,K}^\star
=\{D_{1,0,K}^\star,\ldots,D_{k_0,0,K}^\star\}$
denote the supported comparator constructed in the Supplement. For a
supported candidate
$\pi^\star=\{D_1^\star,\ldots,D_k^\star\}$, let
$\mathcal V_j$ and $\mathcal V_{\ell,0,K}^\star$ be its candidate and
comparator block sets, respectively, set
$b_{j\ell}=|\mathcal V_j\cap\mathcal V_{\ell,0,K}^\star|$, and choose
$m_{\pi^\star}(j)\in\arg\max_{\ell\leq k_0}b_{j\ell}$. Define
$q_K(\pi^\star)=K^2-\sum_{j=1}^k b_{j,m_{\pi^\star}(j)}$ and
\begin{equation}
\Pi_{5,K}^\star=
\left\{
\pi^\star\in\mathfrak P_K:
|\pi^\star|=k_0,\ 
\{m_{\pi^\star}(j):1\leq j\leq k_0\}=\{1,\ldots,k_0\},\
q_K(\pi^\star)<q_n
\right\}.
\label{eq-good-partition-class}
\end{equation}


\begin{theorem}[Posterior partition contraction]
\label{thm-partition-contraction}
Under Assumptions~\ref{assump-spatial}--\ref{assump-representative-prior},
for every fixed $\alpha_0>0$ there exist constants $C,c,M>0$, independent of
$n$, such that, with $P_0$-probability tending to one,
$\Pi_{\tau,n}(\Pi_{5,K}^\star\mid X_{1:n},s_{1:n})
\geq 1-Ce^{-c\zeta_n}$ and
$\Pi_{\tau,n}(|\pi^\star|=k_0\mid X_{1:n},s_{1:n})
\geq 1-Ce^{-c\zeta_n}$. Moreover,
\begin{align}
\Pi_{\tau,n}\!\left\{
\epsilon(\pi^\star,\pi_0)\leq M\rho_n
\,\middle|\,X_{1:n},s_{1:n}
\right\}
&\geq 1-Ce^{-c\zeta_n},\\
\Pi_{\tau,n}\!\left\{
\epsilon_n\{\mathcal P(\pi^\star),\mathcal P_0\}\leq M\rho_n
\,\middle|\,X_{1:n},s_{1:n}
\right\}
&\geq 1-Ce^{-c\zeta_n}.
\label{eq-main-partition-contraction}
\end{align}
\end{theorem}
The detailed proof of Theorem \ref{thm-partition-contraction} is given in the Supplementary Material. The partition-contraction rate in
Theorem~\ref{thm-partition-contraction} coincides with that of
\citet{huang2025consistent}. In both settings, the rate is determined by
the spatial resolution of the mesh and the complexity of distinguishing
nearby contiguous partitions. The good partition class permits at most
$q_n=K\log^{\alpha_0}(n)$ mismatched blocks, and each block has area of
order $K^{-2}$. The corresponding domain discrepancy is therefore of
order
$\rho_n=q_n/K^2=\log^{\alpha_0}(n)/K$. With
$K\asymp\{n/\log^{1+\alpha_b}(n)\}^{1/2}$, this gives
$\rho_n\asymp
n^{-1/2}\log^{\alpha_0+(1+\alpha_b)/2}(n)$, which is exactly the
partition-contraction rate obtained by \citet{huang2025consistent}.
Replacing their Gaussian response model by the present Bregman
generalized-Bayes construction does not alter this spatial scale. Under
the curvature and separation conditions, differences between distinct
regional centers remain detectable through the Bregman criterion, while
Proposition~\ref{prop-uniform-integrated-center} shows that integrating
the representatives changes partition log-score comparisons only by
$O(\log n)$. Hence the response formulation changes the criterion used
to distinguish candidate partitions, but not the mesh-driven balance
between spatial approximation and partition complexity that determines
the contraction rate.

\noindent\textbf{Representative and field contraction.}
For a joint posterior draw
$\{\pi^\star,\mu_{1:|\pi^\star|}\}$, define
$\mu_{\pi^\star}(s)=\sum_{j=1}^{|\pi^\star|}
\mu_j\mathbf 1\{s\in D_j^\star\}$. General concentration-rate theory for Gibbs posteriors based directly
on loss functions is developed by \citet{syring2023gibbs} under
conditions controlling the excess loss and its stochastic fluctuations.
The result below has an additional spatial component. The
cluster-specific representatives are indexed by an unknown partition,
so their contraction is established jointly with concentration of the
partition on $\Pi_{5,K}^\star$, while contraction of the center field
also accounts for the spatial region on which the sampled and true
partitions disagree.


\begin{theorem}[Representative and center-field contraction]
\label{thm-representative-field-contraction}
Under Assumptions~\ref{assump-spatial}--\ref{assump-representative-prior},
fix $\alpha_0>0$. For any deterministic sequences
$M_n,M_n',M_n''\to\infty$, with $P_0$-probability tending to one,
\begin{align}
\Pi_{\tau,n}\!\left(
\pi^\star\in\Pi_{5,K}^\star,\,
\max_{1\le j\le k_0}
\|\mu_j-\mu_{m_{\pi^\star}(j),0}\|>M_n\rho_n
\,\middle|\,X_{1:n},s_{1:n}
\right)
&\longrightarrow 0,
\label{eq-representative-contraction}
\\
\Pi_{\tau,n}\!\left(
\int_{\mathcal D}
\|\mu_{\pi^\star}(s)-\mu_0(s)\|^2p_D(s)\,ds
>M_n'\rho_n
\,\middle|\,X_{1:n},s_{1:n}
\right)
&\longrightarrow 0,
\label{eq-field-l2-contraction}
\\
\Pi_{\tau,n}\!\left(
\int_{\mathcal D}
D_\Phi\{\mu_0(s),\mu_{\pi^\star}(s)\}p_D(s)\,ds
>M_n''\rho_n
\,\middle|\,X_{1:n},s_{1:n}
\right)
&\longrightarrow 0.
\label{eq-field-bregman-contraction}
\end{align}
\end{theorem}
The detailed proof of Theorem \ref{thm-representative-field-contraction} is given in the Supplementary Material. On $\Pi_{5,K}^\star$, the matching map $m_{\pi^\star}$ is a bijection, so
$\{\mu_{m_{\pi^\star}(j),0}:1\le j\le k_0\}
=\{\mu_{\ell,0}:1\le\ell\le k_0\}$ deterministically. The substantive representative result is
\eqref{eq-representative-contraction}.
On the good partition class, candidate-region empirical centroids are within $O(\rho_n)$ of their matched true centers, while the conditional generalized posterior concentrates around each empirical centroid under the curvature, prior, and local exponential-moment conditions stated above. For the center-field losses, the proof combines concentration of the cluster-specific representatives with the fact that the spatial mismatch set has area $O(\rho_n)$. The contributions from both the correctly matched portion of the domain and the spatial mismatch region are therefore controlled at order $O(\rho_n)$, yielding the integrated squared and Bregman field-contraction results in \eqref{eq-field-l2-contraction} and \eqref{eq-field-bregman-contraction}. This adds the representative and field-contraction layer to the partition-recovery machinery of \citet{huang2025consistent}.

\section{Learning-rate calibration}
\label{sec:tau-calibration}

The learning rate $\tau$ determines the scale of the generalized likelihood relative to the partition prior. Since this scale is not supplied by a parametric response likelihood, we calibrate $\tau$ by matching the sampling dispersion of the empirical Bregman representative with the dispersion implied by the local quadratic generalized likelihood.

Consider first $m$ independent represented responses with $\mathbb{E}_0(X_i)=\mu_0$ and $\operatorname{Var}_0(X_i)=\Omega$, and write $H_0=\nabla^2\Phi(\mu_0)$ and $\bar X=m^{-1}\sum_{i=1}^mX_i$. By the Bregman risk decomposition, the representative-dependent part of the generalized likelihood is proportional to $\exp\{-\tau mD_\Phi(\bar X,\mu)\}$. Since $D_\Phi(\bar X,\mu)=\frac{1}{2}(\mu-\bar X)^\top\nabla^2\Phi(\bar X)(\mu-\bar X)+o(\|\mu-\bar X\|^2)$ and $\bar X\to\mu_0$ in probability, its local covariance is asymptotically $(\tau m)^{-1}H_0^{-1}$. In contrast, the sampling covariance of $\bar X$ is $\Omega/m$. A scalar $\tau$ cannot generally equate these covariance matrices in every direction, so we match their total dispersion in the local Bregman metric $H_0$. Specifically, $\mathbb{E}_0[m(\bar X-\mu_0)^\top H_0(\bar X-\mu_0)]=\operatorname{tr}(H_0\Omega)$, whereas under the local Gaussian approximation the corresponding expected quadratic dispersion is $r/\tau$, where $r=\dim(\mathcal V)$. Equating the two gives the local population calibration $\tau^\star=r/\operatorname{tr}(H_0\Omega)$.

This criterion is the Bregman specialization of the scalar magnitude adjustment used for composite and criterion-based likelihoods \citep{pauli2011bayesian,ribatet2012bayesian}. For a general loss $\ell(\theta,X)$, let $A=\nabla_\theta^2\mathbb{E}_0\{\ell(\theta,X)\}|_{\theta=\theta_0}$ and $B=\operatorname{Var}_0\{\nabla_\theta\ell(\theta_0,X)\}$. The empirical risk minimizer has sandwich covariance $n^{-1}A^{-1}BA^{-1}$, whereas the local generalized-posterior covariance is $(\tau n)^{-1}A^{-1}$. Matching their expected quadratic dispersions in the $A$ metric yields the scalar magnitude adjustment $r/\operatorname{tr}(A^{-1}B)$. For the Bregman loss $\ell(\mu,X)=D_\Phi(X,\mu)$, $A=H_0$ and $B=H_0\Omega H_0$, so $r/\operatorname{tr}(A^{-1}B)=r/\operatorname{tr}(H_0\Omega)$. Thus the proposed calibration is exactly the corresponding scalar magnitude adjustment for the Bregman criterion. Alternative generalized-Bayes calibrations use different scalar targets, including information matching and frequentist coverage \citep{lyddon2019general,syring2019calibrating}.

For the spatial model, let $\Omega_0(s)=\operatorname{Var}_0(X\mid S=s)$ and $H_0(s)=\nabla^2\Phi\{\mu_0(s)\}$. Define the population Bregman noise level by $a(P_0)=\mathbb{E}_S[\operatorname{tr}\{H_0(S)\Omega_0(S)\}]=\mathbb{E}_0[\{X-\mu_0(S)\}^{\top}H_0(S)\{X-\mu_0(S)\}]$. When $0<a(P_0)<\infty$, define the associated population learning-rate functional by $\tau(P_0)=r/a(P_0)$. Thus $\tau(P_0)$ is not a parameter of the response-generating distribution; it is the unique scalar learning rate induced by $P_0$ under the preceding Bregman-dispersion matching criterion.


To estimate $a(P_0)$ without knowing the true regions, for each mesh block $B$ let $I_B=\{i:s_i\in B\}$ and $m_B=|I_B|$, and, whenever $m_B\geq1$, let $\bar X_B=m_B^{-1}\sum_{i\in I_B}X_i$. Fix any $u_{\mathrm{ref}}\in\mathcal U$ and define $\bar X_B^{\mathrm{int}}=\bar X_B$ if $\bar X_B\in\mathcal U$ and $\bar X_B^{\mathrm{int}}=u_{\mathrm{ref}}$ otherwise. We define $\widehat a_n=\{\sum_{B:m_B\ge2}\sum_{i\in I_B}(X_i-\bar X_B)^\top\nabla^2\Phi(\bar X_B^{\mathrm{int}})(X_i-\bar X_B)\}/\{\sum_{B:m_B\ge2}(m_B-1)\}$ when $\sum_{B:m_B\ge2}(m_B-1)>0$, set $\widehat a_n=0$ otherwise, and define $\widehat\tau_n=r/\widehat a_n$ when $\widehat a_n>0$, with $\widehat\tau_n=1$ otherwise. For a block $B$ contained in the true region $D_{\ell,0}$, let $H_\ell=\nabla^2\Phi(\mu_{\ell,0})$ and $\Omega_i=\operatorname{Var}_0(X_i\mid s_{1:n})$. Conditional independence gives $\mathbb{E}_0[(m_B-1)^{-1}\sum_{i\in I_B}(X_i-\bar X_B)^\top H_\ell(X_i-\bar X_B)\mid s_{1:n}]=m_B^{-1}\sum_{i\in I_B}\operatorname{tr}(H_\ell\Omega_i)$. Thus within-block centering removes the unknown regional mean, while the factor $m_B-1$ gives the degrees-of-freedom correction for the desired Bregman-scaled conditional variation.


\begin{theorem}[Consistency of the learning-rate calibration]
\label{thm:tau-consistency}
Suppose Assumptions~1--2 hold. Assume additionally that there exists a measurable covariance function $\Omega_0(\cdot)$ such that $\operatorname{Var}_0(X_i\mid s_{1:n})=\Omega_0(s_i)$ almost surely for every $i=1,\ldots,n$, and that $0<a(P_0)<\infty$. Define $a_n=n^{-1}\sum_{i=1}^n\operatorname{tr}[\nabla^2\Phi\{\mu_0(s_i)\}\Omega_0(s_i)]$. Then $\widehat a_n-a_n\xrightarrow{P_0}0$ and $a_n\xrightarrow{P_0}a(P_0)$. Consequently, $\widehat a_n\xrightarrow{P_0}a(P_0)$ and $\widehat\tau_n\xrightarrow{P_0}\tau(P_0)$.
\end{theorem}

The proof is given in the Supplementary Material. 
Hence the implemented calibration consistently estimates the population scalar magnitude adjustment associated with the true response distribution and the chosen Bregman geometry, without specifying a parametric conditional response likelihood.

\section{Simulation Study}
\label{sec-simulation}

We evaluate the recovery of spatial partitions and compare our method (\textbf{Bregman-Bayes}) with competing clustering procedures, and examine the empirical behavior of the posterior errors as the number of sampled locations increases.

\noindent\textbf{Data generation and model settings.}
Throughout the simulation study, we sample locations uniformly from the fixed three-cluster U-shaped domain shown in Figure~1(a), and use the same spatial domain and true partition for all response classes. We consider four response-generating mechanisms. The Euclidean, count-vector, and positive-definite matrix responses have representation dimension 15, while the distribution-valued responses lie on a 16-category simplex with intrinsic dimension 15. For Euclidean responses, $X_i\mid s_i\in D_{\ell,0}
\sim N_{15}(\mu_{\ell,0},1.10^2I_{15}),
\qquad \ell=1,2,3,$ where $\mu_{\ell,0}
=
2\sqrt{\frac{3}{2}}
\begin{pmatrix}
e_\ell-\mathbf 1_3/3\\
\mathbf 0_{12}
\end{pmatrix}$, with $e_\ell$ denoting the $\ell$th standard basis vector in
$\mathbb R^3$. The three regional means are equidistant, with
$\|\mu_{\ell,0}-\mu_{h,0}\|^2=12$ for $\ell\neq h$.

For distribution-valued responses, $P_i\mid s_i\in D_{\ell,0}\sim\operatorname{Dirichlet}(3000p_\ell)$ on 16 categories, where
$p_1=(0.003,0.003,0.007,0.007,q,\ldots,q)$,
$p_2=(0.005,0.005,0.005,0.005,q,\ldots,q)$, and
$p_3=(0.007,0.007,0.003,0.003,q,\ldots,q)$, with
$q=0.98/12$ repeated twelve times. Thus all between-region differences occur in the first four low-probability categories, while the remaining twelve categories have the same probabilities in all three regions.

For nonnegative count-vector responses, the 15 coordinates are conditionally independent with $X_{ih}\mid s_i\in D_{\ell,0}\sim\operatorname{Poisson}(\mu_{\ell h})$, where
$\mu_1=(0.8\,\mathbf 1_8,60\,\mathbf 1_7)$,
$\mu_2=(2.0\,\mathbf 1_8,60\,\mathbf 1_7)$, and
$\mu_3=(4.2\,\mathbf 1_8,60\,\mathbf 1_7)$.
The first eight coordinates contain the regional signal, while the final seven coordinates have intensity 60 in every region and therefore contain no partition information.

For positive-definite matrix responses, each observation is a $5\times5$ matrix and hence has a 15-dimensional symmetric-matrix representation. We generate
$X_i\mid s_i\in D_{\ell,0}\sim W_5(8,M_\ell/8)$, so that $E(X_i\mid s_i\in D_{\ell,0})=M_\ell$, where
$M_1=\operatorname{diag}(0.010,0.010,0.010,100,100)$,
$M_2=\operatorname{diag}(0.026,0.026,0.026,100,100)$, and
$M_3=\operatorname{diag}(0.068,0.068,0.068,100,100)$.
The regional signal consists of successive multiplicative changes in the three small eigenvalues of the regional mean matrices, with successive ratios approximately $2.6$, while the two common eigenvalues equal to 100 generate substantial sampling variability that dominates ordinary Frobenius distance without contributing to between-region mean separation.

The representative priors are $N_{15}(0,I_{15})$ for Euclidean responses, $\operatorname{Dirichlet}(\mathbf 1_{16})$ for distribution-valued responses, independent $\operatorname{Gamma}(1,1)$ distributions for count coordinates, and $\operatorname{IW}_5(6,I_5)$ for positive-definite matrices. We parameterize the mesh resolution and cluster-number penalty as
$K_n=\left\lfloor c_b\{n/\log^{1+\alpha_b}(n)\}^{1/2}\right\rfloor$ and
$\log(\lambda_n^{-1})=c_p n/\log^{\alpha_p}(n)$, where $c_b>0$ controls the spatial mesh resolution and $c_p>0$ controls the strength of the prior penalty on the number of regions. We set $c_b=3.5$, $\alpha_b=1$, $c_p=0.5$, $\alpha_p=0.5$, $\alpha_0=0.1$, and $k_{\max}=8$.
For each simulated dataset, the learning rate $\tau$ is selected once using the calibration procedure in Section~\ref{sec:tau-calibration} and is then held fixed during posterior sampling. Unless stated otherwise, we run $80{,}000$ MCMC iterations, discard the first $60{,}000$, and retain every twentieth subsequent draw.

For a finite-sample illustration, we generate one distribution-valued dataset with $n=3000$ using the specification above. For this illustration, we run $120{,}000$ MCMC iterations with a burn-in of $80{,}000$ and thinning interval 20. The displayed posterior partition is the retained draw having the largest collapsed posterior mass, and the corresponding cluster-specific representatives are taken from the same posterior draw.

\begin{figure}[t]
    \centering
    \includegraphics[width=0.7\textwidth]{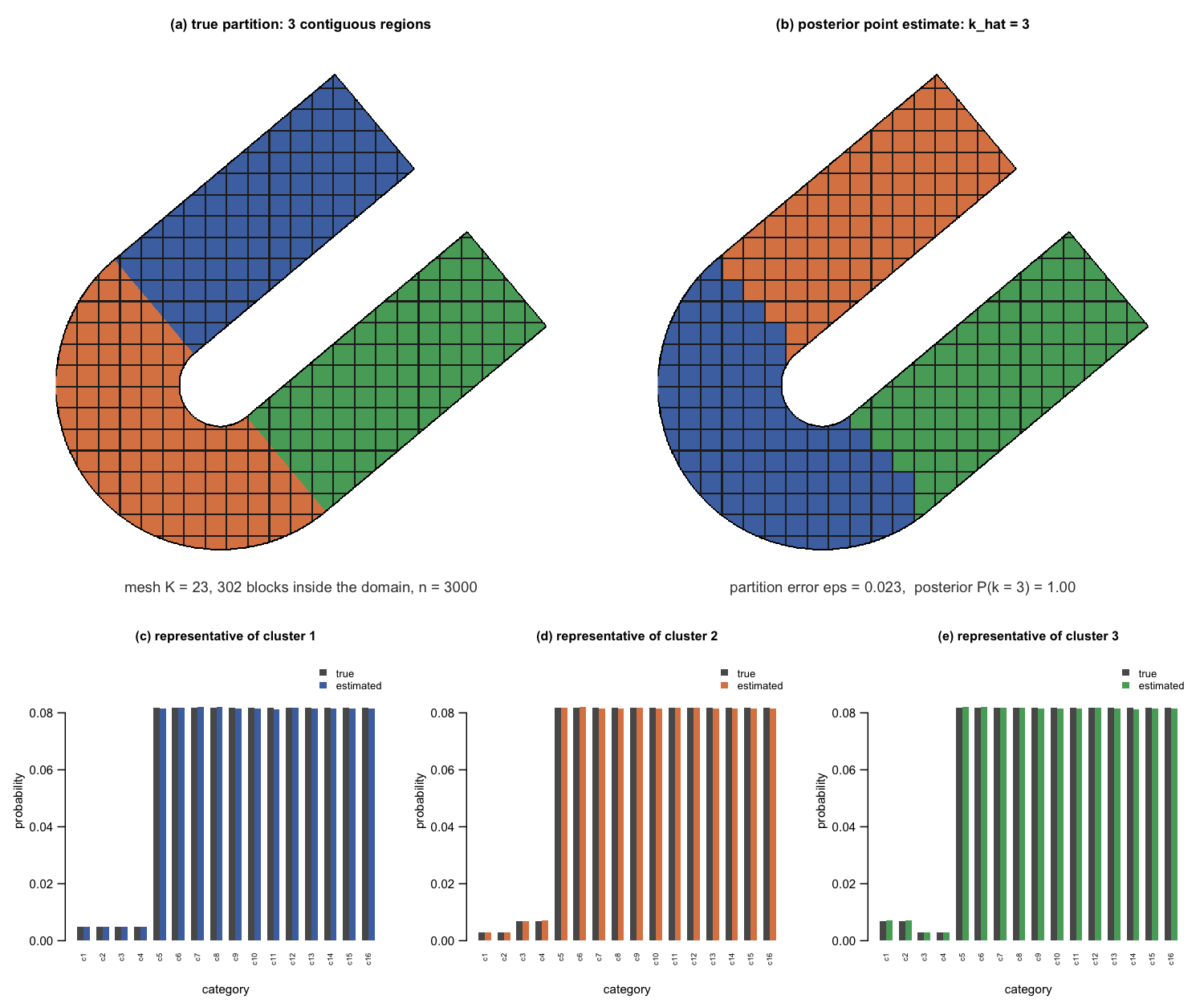}
    \caption{Illustration for one distribution-valued dataset with $n=3000$. Panel (a) shows the true three-region partition and panel (b) shows the posterior partition point estimate. Panels (c)--(e) compare the true and estimated probability-vector representatives for the three matched regions.}
    \label{fig-simulation-illustration}
\end{figure}

Figure~\ref{fig-simulation-illustration} shows that the estimated partition reproduces the two internal boundaries of the U-shaped truth up to the resolution of the square mesh. The remaining discrepancies are concentrated in blocks intersecting the true boundaries. The posterior assigns probability one to $k=3$, the selected partition contains three regions, and its domain partition error is $0.023$. Panels (c)--(e) show close agreement between the true and estimated representatives across all 16 categories. In particular, the method recovers the regional differences among the four low-probability categories while correctly estimating the twelve larger probabilities that are common across regions. The figure therefore illustrates simultaneous recovery of the spatial partition and its associated cluster-specific representatives.

\subsection{Comparison}
\label{sec-simulation-comparison}

We compare the methods at $n=1000$ using 30 independently generated datasets from each of the four response-generating mechanisms above. We consider SKATER \citep{assunccao2006efficient}, which constructs a minimum spanning tree using Euclidean attribute dissimilarities and recursively removes edges according to reductions in within-region Euclidean sums of squares, hard Bregman clustering \citep{banerjee2005clustering}, which minimizes the response-specific Bregman within-cluster criterion without imposing spatial contiguity, and a Gaussian SPAT-CLUS \citep{huang2025consistent} comparator using the same meshed partition prior as we use in our model. Conditional on a partition $\mathcal P=\{C_1,\ldots,C_k\}$, the Gaussian comparator uses
$X_i\mid i\in C_j,\mu_j,\sigma^2\sim N_r(\mu_j,\sigma^2I_r)$ and
$\mu_j\mid\sigma^2,\lambda\sim N_r(0,\sigma^2\lambda^{-1}I_r)$,
with $\sigma^2\sim\operatorname{IG}(1/2,1/2)$ and
$\lambda\sim\operatorname{Gamma}(5\times10^{-7},5\times10^{-7})$ under the shape-rate parameterization. The cluster means are integrated analytically, while $\sigma^2$ and $\lambda$ are updated within the MCMC sampler.

SKATER and Gaussian SPAT-CLUS are applied to the same spatial mesh used by our method. 
Hard Bregman clustering is applied directly to the individual observations. Bregman-Bayes and Gaussian SPAT-CLUS use exactly the same spanning-tree partition prior and both estimate the number of clusters with support restricted by $k_{\max}=8$. SKATER and hard Bregman clustering are instead supplied the true number of clusters $k_0=3$. We evaluate the estimated partition at the observed locations using the adjusted Rand index (ARI) and the symmetric set-partition discrepancy $\epsilon_n$ defined in Section~\ref{sec-theory}. For both Bayesian methods, the reported partition is the variation-of-information (VI) optimal posterior clustering obtained using SALSO \citep{dahl2022search}.

\begin{table}[!t]
\centering
\small
\setlength{\tabcolsep}{8.5pt}
\caption{Comparison of spatial partition recovery at $n=1000$ over 30 independently generated datasets. Entries for ARI and $\epsilon_n$ are Monte Carlo means with standard deviations in parentheses. Larger ARI and smaller $\epsilon_n$ indicate better recovery. SKATER and Bregman $k$-means are supplied the true number of clusters $k_0=3$, while Bregman-Bayes and Gaussian SPAT-CLUS estimate $k$ under the same partition prior. The best mean ARI and $\epsilon_n$ within each response class are shown in bold, with both values bolded when the displayed means are tied.}
\label{tab-simulation-comparison}
\begin{tabular}{llccc}
\toprule
Response & Method & ARI & $\epsilon_n$ & $\widehat{k}$ \\
\midrule
Euclidean
& Bregman-Bayes & $\mathbf{0.960\ (0.008)}$ & $\mathbf{0.026\ (0.006)}$ & $3.00\ (0.00)$ \\
& Gaussian SPAT-CLUS & $\mathbf{0.960\ (0.009)}$ & $0.027\ (0.006)$ & $3.00\ (0.00)$ \\
& SKATER & $0.954\ (0.014)$ & $0.030\ (0.009)$ & $3$ (given) \\
& Bregman $k$-means & $0.709\ (0.022)$ & $0.211\ (0.018)$ & $3$ (given) \\
\addlinespace[2pt]

Distribution
& Bregman-Bayes & $\mathbf{0.956\ (0.012)}$ & $\mathbf{0.029\ (0.008)}$ & $3.00\ (0.00)$ \\
& Gaussian SPAT-CLUS & $0.000\ (0.000)$ & $0.643\ (0.009)$ & $1.00\ (0.00)$ \\
& SKATER & $0.524\ (0.028)$ & $0.356\ (0.027)$ & $3$ (given) \\
& Bregman $k$-means & $0.787\ (0.019)$ & $0.150\ (0.015)$ & $3$ (given) \\
\addlinespace[2pt]

Count
& Bregman-Bayes & $\mathbf{0.901\ (0.140)}$ & $\mathbf{0.068\ (0.100)}$ & $2.87\ (0.35)$ \\
& Gaussian SPAT-CLUS & $0.000\ (0.000)$ & $0.639\ (0.015)$ & $1.00\ (0.00)$ \\
& SKATER & $0.511\ (0.034)$ & $0.364\ (0.034)$ & $3$ (given) \\
& Bregman $k$-means & $0.792\ (0.020)$ & $0.147\ (0.015)$ & $3$ (given) \\
\addlinespace[2pt]

SPD matrix
& Bregman-Bayes & $\mathbf{0.952\ (0.013)}$ & $\mathbf{0.032\ (0.009)}$ & $3.00\ (0.00)$ \\
& Gaussian SPAT-CLUS & $0.000\ (0.000)$ & $0.642\ (0.010)$ & $1.00\ (0.00)$ \\
& SKATER & $0.000\ (0.001)$ & $0.643\ (0.011)$ & $3$ (given) \\
& Bregman $k$-means & $0.802\ (0.029)$ & $0.139\ (0.022)$ & $3$ (given) \\
\bottomrule
\end{tabular}
\end{table}

For Euclidean responses, Bregman-Bayes and Gaussian SPAT-CLUS both attain mean ARI $0.960$, with mean $\epsilon_n$ values $0.026$ and $0.027$, respectively. SKATER gives comparable recovery with mean ARI $0.954$. This agreement is expected because the Gaussian SPAT-CLUS sampling model is correctly specified and the Euclidean Bregman divergence reduces to squared-Euclidean loss. In contrast, nonspatial Bregman $k$-means has mean ARI $0.709$ despite being supplied $k_0=3$, showing the gain from spatial aggregation and contiguity even when the response geometry itself is Euclidean.

The three non-Euclidean mechanisms separate the response-specific Bregman geometry from the Euclidean geometry. For distribution-valued responses, the regional differences occur only among probabilities between $0.003$ and $0.007$, while the twelve larger coordinates are identical across regions and dominate ordinary Euclidean variation. Kullback--Leibler geometry weights changes in these rare categories according to their local probability scale. Bregman-Bayes attains mean ARI $0.956$, compared with $0.787$ for nonspatial Bregman $k$-means and $0.524$ for SKATER. Gaussian SPAT-CLUS selects a single region in every replication, producing mean ARI zero. For counts, the seven Poisson coordinates with intensity 60 contain no clustering signal but contribute substantially to raw Euclidean variation, whereas the generalized Kullback--Leibler curvature scales coordinate $h$ by $1/\mu_h$. Bregman-Bayes consequently attains mean ARI $0.901$, compared with $0.792$ for Bregman $k$-means and $0.511$ for SKATER. Gaussian SPAT-CLUS again selects one region in every replication. The larger standard deviation for Bregman-Bayes in this setting and the average $\widehat{k}=2.87$ reflect occasional finite-sample merging of two true regions.

The positive-definite matrix experiment gives the strongest separation from Euclidean and Frobenius criteria. The regional signal consists of multiplicative changes in the three small eigenvalues, with ratios approximately $2.6$, while the two common eigenvalues equal to 100 generate large Frobenius-scale sampling variation without contributing regional information. The LogDet divergence responds to the relative matrix change and is invariant to the common scale in the manner required by this construction. Bregman-Bayes attains mean ARI $0.952$, while Bregman $k$-means attains $0.802$. SKATER and Gaussian SPAT-CLUS both fail to recover the partition under their Euclidean or Gaussian criteria. The systematic underclustering of Gaussian SPAT-CLUS across the non-Euclidean settings is consistent with its use of an inappropriate Euclidean geometry, under which variation in response components carrying little or no regional information can obscure the region-specific signal. Across the distribution, count, and matrix experiments, Bregman $k$-means consistently improves on methods using an inappropriate Euclidean geometry, while Bregman-Bayes further improves recovery by combining the response-specific divergence with spatial contiguity and regional pooling. Moreover, unlike SKATER and Bregman $k$-means, Bregman-Bayes is not supplied with the true number of regions.

\FloatBarrier

\subsection{Asymptotic Study}
\label{sec-simulation-asymptotic}

We next examine the posterior errors as the number of sampled locations increases for the distribution-valued, count, and positive-definite matrix responses. The spatial domain, true three-cluster partition, regional representatives, and response-generating distributions are held fixed at the specifications used in the comparison study. For distribution-valued and positive-definite matrix responses, we use
$n\in\{200,250,435,550,1000,1610,2392,3365\}$.
For count responses, we use
$n\in\{1000,1610,2392,3365,4538,5435,6960,8100\}$ because the finite-sample count posterior begins to recover the three-region partition reliably near $n=1000$ under the present signal and prior specification. 

For each response class and sample size, we generate 30 independent datasets. For distribution-valued and positive-definite matrix responses, we run $120{,}000$ MCMC iterations, discard the first $100{,}000$, and retain every twentieth subsequent draw. For count responses, we run $220{,}000$ iterations, discard the first $180{,}000$, and retain every fortieth draw. Thus $M=1000$ posterior draws are retained for every fitted dataset.

For each simulated dataset, let $s=1,\ldots,M$ index the retained posterior draws. For draw $s$, let $k_s$ and $\pi_s$ denote its number of regions and domain partition, let $\mu_{s,j}$ denote the representative of its $j$th region, and let $M_s(j)$ be the true region having the largest intersection area with that estimated region. We use the normalization $N_n=n^{1/2}/\log^{\alpha_0+(1+\alpha_b)/2}(n)=n^{1/2}/\log^{1.1}(n)$ corresponding to the spatial contraction rate in Section~\ref{sec-theory}. For each simulated dataset we compute four posterior summaries. The first is the posterior probability of the correct region number, $\widehat p_k=M^{-1}\sum_{s=1}^M\mathbf 1(k_s=k_0)$. The second is the normalized posterior mean domain-partition error, $e_1=N_nM^{-1}\sum_{s=1}^M\epsilon(\pi_s,\pi_0)$. The third is the normalized maximum representative error, $e_2=N_nM^{-1}\sum_{s=1}^M\max_{1\leq j\leq k_s}\|\mu_{s,j}-\mu_{M_s(j),0}\|$. The fourth is the normalized integrated Bregman field error, $e_4=N_nM^{-1}\sum_{s=1}^M\int_{\mathcal D}D_\Phi\{\mu_0(u),\mu_s(u)\}\,du$. The norm in $e_2$ is the norm of the representation space used in the theory, while $e_4$ uses the response-specific Bregman divergence with the true field in its first argument. Figure~\ref{fig-simulation-asymptotic} gives boxplots of these four dataset-level posterior summaries across the 30 replications at every sample size. As in the finite-sample study, $\tau$ is estimated separately for each dataset using the calibration procedure of Section~\ref{sec:tau-calibration} and then held fixed during that posterior run.

\begin{figure}[!t]
    \centering
    \includegraphics[width=\textwidth]{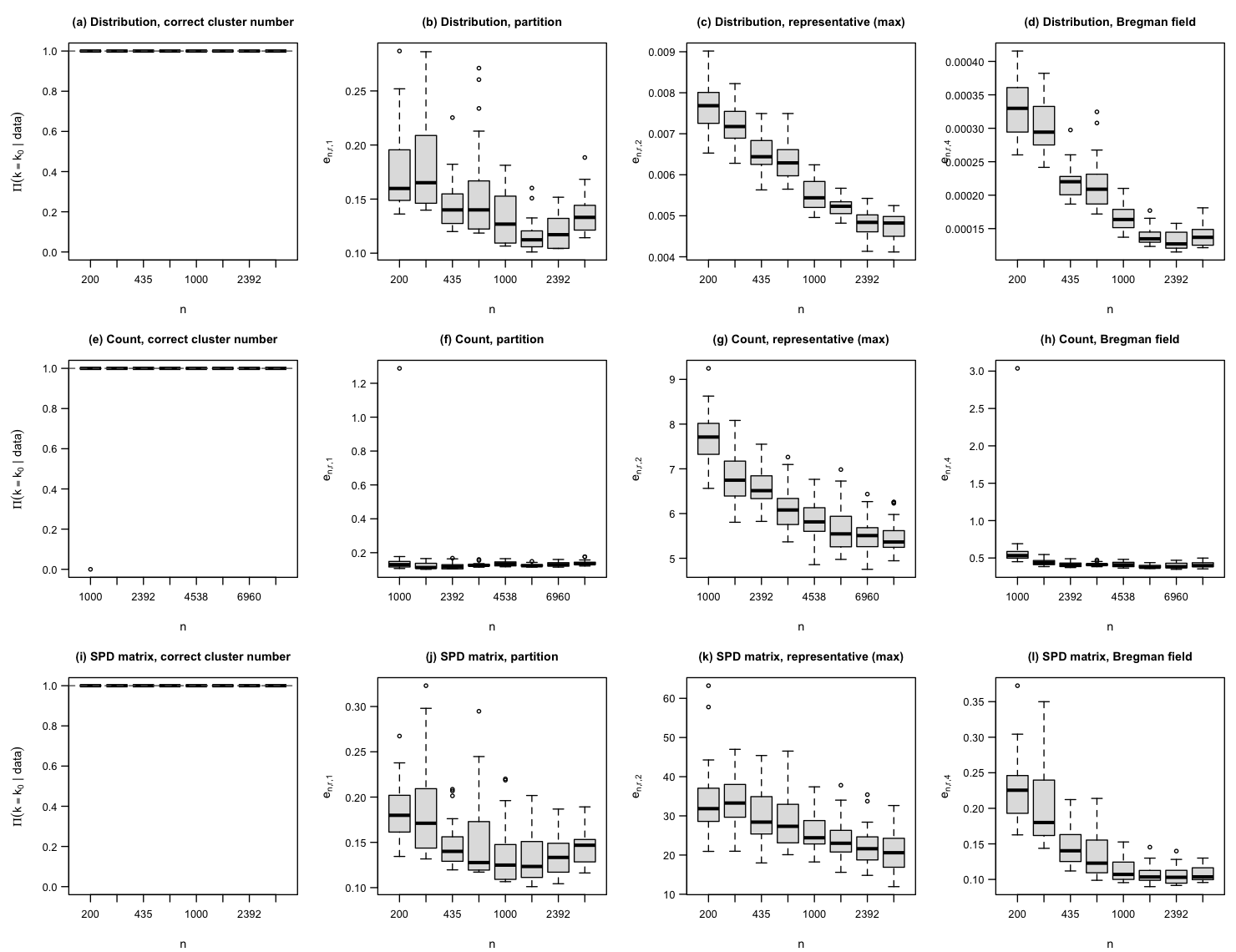}
    \caption{Asymptotic simulation results for distribution-valued, count, and positive-definite matrix responses. Each row corresponds to one response class and each boxplot summarizes 30 independently generated datasets. The four columns show the posterior probability of the true number of regions, the posterior mean domain-partition error, the maximum representative error, and the integrated Bregman center-field error. The three error measures are multiplied by $N_n=n^{1/2}/\log^{1.1}(n)$.}
    \label{fig-simulation-asymptotic}
\end{figure}

Figure~\ref{fig-simulation-asymptotic} shows stable recovery of the true number of regions over the displayed sample-size ranges. For distribution-valued and positive-definite matrix responses, the posterior probability of $k=3$ is essentially one throughout the experiment. For counts, the posterior is also concentrated on $k=3$ over the displayed range apart from an isolated failure at the smallest sample size, after which the posterior probability is essentially one.

The normalized representative and Bregman field errors decrease substantially as $n$ increases for all three response classes. The normalized partition error decreases over the initial sample sizes and then fluctuates within a bounded range. This behavior is expected because the blockwise partition cannot improve beyond the approximation error induced by the spatial mesh, and multiplication by $N_n$ makes the finite-sample behavior of this mesh contribution visible. The effect is particularly apparent for the count response, whose displayed range begins near the finite-sample threshold at which the posterior consistently separates all three regions. For distribution-valued and positive-definite matrix responses, the representative and Bregman field errors show a clear reduction while the scaled partition error approaches a stable level determined primarily by boundary blocks.


\FloatBarrier

\section{Houston Racial-Composition Analysis}
\label{sec:houston}

\noindent\textbf{Data and analysis.}
We apply the proposed method to racial-composition data for 333 Census Block Groups (CBGs) in central Harris County, Texas. We restrict the analysis to CBGs whose centroids lie within longitude $[-95.5,-95.3]$ and latitude $[29.65,29.75]$ in the urban core of Houston, Texas, including the Heights, Montrose, River Oaks, Museum District, Midtown, Third Ward, East End, and Near Northside areas. For CBG $i$, the response is the four-category composition $P_i=(P_{i,\mathrm{White}},P_{i,\mathrm{African\ American}},P_{i,\mathrm{Asian}},P_{i,\mathrm{Hispanic}})$, obtained by normalizing the four recorded components to sum to one, and its location $s_i$ is taken to be the CBG centroid.

We use the negative-entropy generator, so that the Bregman divergence between two compositions is the Kullback--Leibler divergence, and assign the cluster-specific representatives the matched $\mathrm{Dirichlet}(1,1,1,1)$ prior. All remaining settings are those of the simulation study, namely $c_p=0.5$, $\alpha_p=0.5$, $c_b=3.5$, $\alpha_b=1$ and $k_{\max}=8$ . The learning rate is calibrated as per Section~\ref{sec:tau-calibration}. We run $80{,}000$ MCMC iterations, discard the first $60{,}000$, and retain every twentieth subsequent draw, giving $1{,}000$ posterior draws. The reported partition is the variation-of-information optimal posterior clustering of the 333 CBGs obtained using SALSO \citep{dahl2022search}.

For each CBG, we average its assigned representative over the retained posterior draws to obtain its posterior mean representative. The representative displayed for an estimated cluster is the average of these CBG-specific posterior means over the CBGs assigned to that cluster by the reported partition.

\begin{figure}[t]
    \centering
    \includegraphics[width=\textwidth]{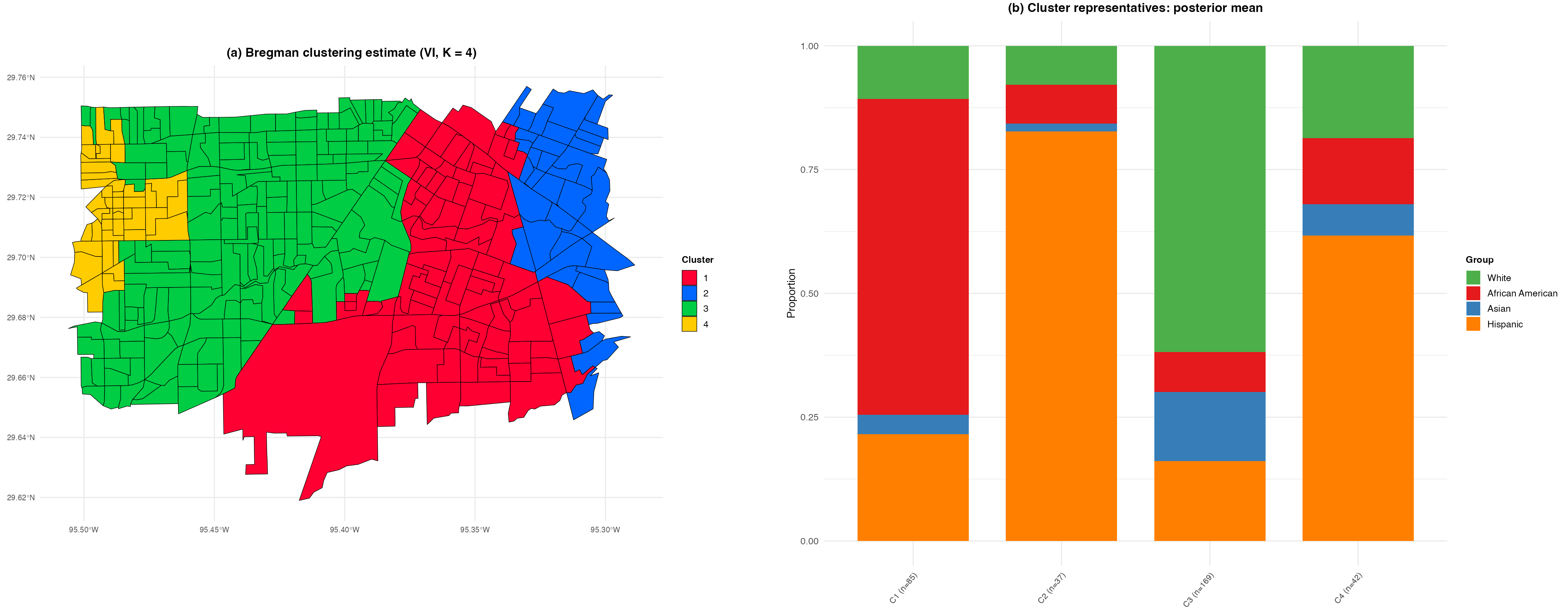}
    \caption{Spatial clustering of racial-composition distributions in central Harris County.
    Panel (a) shows the variation-of-information optimal posterior partition, which contains
    four clusters. Panel (b) shows the cluster-wise
    averages of the CBG-specific posterior mean representatives for White, African American,
    Asian, and Hispanic composition.}
    \label{fig:houston}
\end{figure}

\noindent\textbf{Results.}
The posterior assigns probability one to $k=4$ regions, and Figure~\ref{fig:houston} shows the corresponding four racial-composition regimes. Cluster~1, containing 85 CBGs, occupies the southern and east-central part of the study region and has posterior mean composition $(0.107,0.638,0.040,0.215)$, so it is predominantly African American. Cluster~2, containing 37 CBGs, occupies the northeastern portion of the domain and extends southward along its eastern edge. Its composition is $(0.079,0.079,0.015,0.827)$, with approximately $83\%$ Hispanic representation. Cluster~3, containing 169 CBGs, covers most of the western and central-northern portion of the domain and has composition $(0.619,0.080,0.140,0.161)$; it is predominantly White and also has the largest Asian proportion among the four clusters. Cluster~4, containing 42 CBGs, lies along the northwestern boundary and has composition $(0.186,0.134,0.063,0.617)$. It forms a second Hispanic-dominant region, but with a substantially smaller Hispanic proportion and larger White and African American proportions than cluster~2. The inferred partition therefore separates geographically contiguous areas according to differences in their complete racial-composition vectors.

\section{Conclusion}\label{sec-conc}

We have developed a generalized Bayesian framework for spatial clustering of complex-object responses in which the generalized likelihood is constructed directly from a response-specific Bregman divergence. Combined with a spatially contiguous partition prior, this yields joint inference on the number of regions, their spatial configuration, and their cluster-specific representatives without requiring a fully specified sampling distribution for each response class. The framework covers several important classes of complex objects, admits tractable collapsed computation, and is supported by posterior contraction results.

Several extensions are natural. One direction is to incorporate covariates directly into the generalized Bayesian model, allowing spatial clustering to account for observed predictors or to estimate region-specific relationships between covariates and complex-object responses jointly with the partition. A second direction is to move beyond Bregman geometry and develop analogous generalized Bayesian constructions for object spaces whose natural discrepancy is not induced by a Bregman potential, while retaining tractable inference for spatial partitions and representative objects. Finally, the learning rate $\tau$ determines both the scale of the generalized likelihood and the concentration of the resulting generalized posterior. The calibration considered here targets the scale of the Bregman loss, and an important direction is to develop calibration procedures that additionally preserve uncertainty quantification, for example by selecting $\tau$ so that generalized posterior credible regions attain an appropriate frequentist coverage interpretation. \\













\vspace{-5mm}  
\baselineskip=14pt 
\singlespacing 
\begingroup
    \setstretch{0.5}
    \bibliographystyle{apalike}
    {\small
    \bibliography{bibliography}}
\endgroup 


\clearpage
\appendix


\makeatletter

\newcommand{\appcontents}{%
    \section*{Contents of the Supplementary Material}
    \@starttoc{apc}
}

\newcommand{\appsection}[1]{%
    \section{#1}
    \addcontentsline{apc}{section}
    {\protect\numberline{\thesection}#1}
}

\newcommand{\appsubsection}[1]{%
    \subsection{#1}
    \addcontentsline{apc}{subsection}
    {\protect\numberline{\thesubsection}#1}
}

\newcommand{\appsubsubsection}[1]{%
    \subsubsection{#1}
    \addcontentsline{apc}{subsubsection}
    {\protect\numberline{\thesubsubsection}#1}
}

\makeatother

\begin{center}
{\Large\bf Supplementary Material}
\end{center}

\bigskip

\appcontents

\clearpage

\appsection{Proof of Theorem~\ref{thm-bregman-conjugacy-collapse}}
\label{app-proof-theorem1}

\noindent\textit{Proof.}
Let $C$ be a nonempty cluster. Recall that
\[
R_C(\mu)
=
H_C+n_CA_\Phi(\mu)
-\langle S_C,\nabla\Phi(\mu)\rangle,
\]
where
\[
S_C=\sum_{i\in C}X_i,
\qquad
H_C=\sum_{i\in C}\Phi(X_i),
\qquad
A_\Phi(\mu)
=
\langle\nabla\Phi(\mu),\mu\rangle-\Phi(\mu).
\]
It follows that $\exp\{-\tau R_C(\mu)\}
=
\exp(-\tau H_C)
\exp\left\{
\langle\tau S_C,\nabla\Phi(\mu)\rangle
-\tau n_CA_\Phi(\mu)
\right\}.$

Multiplying by the prior density in
\eqref{eq-bregman-prior-family} gives a generalized posterior
kernel proportional to
$\exp\left\{
\left\langle
\chi_0+\tau S_C,\nabla\Phi(\mu)
\right\rangle
-
(\kappa_0+\tau n_C)A_\Phi(\mu)
\right\}.$
Since
$(\chi_0+\tau S_C,\kappa_0+\tau n_C)
\in\mathcal H_\Phi,$
this kernel is integrable, and therefore
\[
\Pi_\tau(d\mu_C\mid X_C)
=
\Pi_{\chi_0+\tau S_C,\,
      \kappa_0+\tau n_C}(d\mu_C).
\]

Integrating the generalized likelihood with respect to
$\Pi_{\chi_0,\kappa_0}$ gives
\[
\begin{aligned}
M_\tau(C)
&=
\int_{\mathcal U}
\exp\{-\tau R_C(\mu)\}
\Pi_{\chi_0,\kappa_0}(d\mu)
\\
&=
\frac{\exp(-\tau H_C)}
     {Z_\Phi(\chi_0,\kappa_0)}
\int_{\mathcal U}
\exp\left\{
\left\langle
\chi_0+\tau S_C,\nabla\Phi(\mu)
\right\rangle
-
(\kappa_0+\tau n_C)A_\Phi(\mu)
\right\}
\lambda_{\mathcal V}(d\mu)
\\
&=
\exp(-\tau H_C)
\frac{
Z_\Phi(\chi_0+\tau S_C,\kappa_0+\tau n_C)
}{
Z_\Phi(\chi_0,\kappa_0)
}.
\end{aligned}
\]

Now suppose additionally that $\bar X_C\in\mathcal U$.
By the Bregman decomposition,
$R_C(\mu)
=
R_C(\bar X_C)
+
n_CD_\Phi(\bar X_C,\mu).$
Since
$W(C)=R_C(\bar X_C),$
we obtain
\[
\begin{aligned}
M_\tau(C)
&=
\int_{\mathcal U}
\exp\left\{
-\tau W(C)
-\tau n_CD_\Phi(\bar X_C,\mu)
\right\}
\Pi_{\chi_0,\kappa_0}(d\mu)
\\
&=
\exp\{-\tau W(C)\}
\int_{\mathcal U}
\exp\{-\tau n_CD_\Phi(\bar X_C,\mu)\}
\Pi_{\chi_0,\kappa_0}(d\mu)
\\
&=
\exp\{-\tau W(C)\}
J_{n_C}(\bar X_C).
\end{aligned}
\]
This proves the result.
\hfill$\square$


\appsection{Uniform integrated-center control}
\label{app-integrated-center}

\appsubsection{Prior-density and local quadratic bounds}
\label{app-integrated-center-preliminaries}

Write
\[
q_0(\mu)
=
\frac{d\Pi_{\chi_0,\kappa_0}}
     {d\lambda_{\mathcal V}}(\mu)
=
\frac{
\exp\left\{
\langle\chi_0,\nabla\Phi(\mu)\rangle
-\kappa_0A_\Phi(\mu)
\right\}
}{
Z_\Phi(\chi_0,\kappa_0)
},
\qquad \mu\in\mathcal U.
\]

\begin{lemma}[Uniform prior-density bounds]
\label{lem-app-prior-density}
Under Assumptions~\ref{assump-bregman-response} and
\ref{assump-representative-prior}, there exist constants
$0<\underline q\leq\overline q<\infty$ such that
\[
\underline q
\leq q_0(\mu)
\leq\overline q,
\qquad
\mu\in\mathcal K_\rho.
\]
\end{lemma}

\noindent\textit{Proof.}
By Assumption~\ref{assump-bregman-response},
$\Phi\in C^2(\mathcal U)$. Hence
$\nabla\Phi$ and
$A_\Phi(\mu)
=
\langle\nabla\Phi(\mu),\mu\rangle-\Phi(\mu)$
are continuous on $\mathcal U$. Therefore
$\mu\longmapsto
\exp\left\{
\langle\chi_0,\nabla\Phi(\mu)\rangle
-\kappa_0A_\Phi(\mu)
\right\}$
is continuous and strictly positive on $\mathcal U$.

Since $\mathcal V$ is finite dimensional,
$\mathcal H_0=\operatorname{conv}
\{\mu_{1,0},\ldots,\mu_{k_0,0}\}$ is compact, and therefore
$\mathcal K_\rho$ is compact. By
Assumption~\ref{assump-bregman-response},
$\mathcal K_\rho\subset\mathcal U$. The numerator of $q_0$
thus attains a strictly positive minimum and a finite maximum on
$\mathcal K_\rho$. Since
$0<Z_\Phi(\chi_0,\kappa_0)<\infty$
by Assumption~\ref{assump-representative-prior}, the stated bounds
follow.
\hfill$\square$

\begin{lemma}[Local quadratic comparison]
\label{lem-app-local-bregman}
Under Assumption~\ref{assump-bregman-response}, if
$u\in\mathcal K$ and $\mu\in\mathcal U$ satisfy
$\|\mu-u\|\leq\rho$, then
\[
\frac{h_\Phi}{2}\|u-\mu\|^2
\leq
D_\Phi(u,\mu)
\leq
\frac{\overline h_\Phi}{2}\|u-\mu\|^2.
\]
\end{lemma}

\noindent\textit{Proof.}
For $t\in[0,1]$, let
$\gamma(t)=\mu+t(u-\mu).$
Since $u\in\mathcal K$ and $\|\mu-u\|\leq\rho$,
$\operatorname{dist}\{\gamma(t),\mathcal K\}
\leq
\|\gamma(t)-u\|
=
(1-t)\|\mu-u\|
\leq\rho.$
Hence $\gamma(t)\in\mathcal K_\rho\subset\mathcal U$ for all
$t\in[0,1]$.
Taylor's formula with integral remainder gives
\[
D_\Phi(u,\mu)
=
\int_0^1
(1-t)
\left\langle
u-\mu,
\nabla^2\Phi\{\mu+t(u-\mu)\}(u-\mu)
\right\rangle dt.
\]
The Hessian bounds in
Assumption~\ref{assump-bregman-response} therefore yield
\[
h_\Phi\|u-\mu\|^2
\leq
\left\langle
u-\mu,
\nabla^2\Phi\{\mu+t(u-\mu)\}(u-\mu)
\right\rangle
\leq
\overline h_\Phi\|u-\mu\|^2.
\]
Integrating and using
$\int_0^1(1-t)\,dt=1/2$ proves the result.
\hfill$\square$

\appsubsection{Proof of Proposition~\ref{prop-uniform-integrated-center}}
\label{app-proof-integrated-center}

\noindent\textit{Proof.}
Recall that
$J_m(u)
=
\int_{\mathcal U}
\exp\{-\tau mD_\Phi(u,\mu)\}
\Pi_{\chi_0,\kappa_0}(d\mu),
\qquad
m\geq1,\quad u\in\mathcal K.$

We first prove the lower bound. For $m\geq1$,
\[
B\left(u,\frac{\rho}{\sqrt m}\right)
\subseteq\mathcal K_\rho
\subseteq\mathcal U.
\]
By Lemmas~\ref{lem-app-prior-density} and
\ref{lem-app-local-bregman},
\[
\begin{aligned}
J_m(u)
&\geq
\underline q
\int_{\|\mu-u\|\leq\rho/\sqrt m}
\exp\left\{
-\frac{\tau m\overline h_\Phi}{2}
\|\mu-u\|^2
\right\}
\lambda_{\mathcal V}(d\mu).
\end{aligned}
\]
With the change of variables
$v=\sqrt m(\mu-u)$,
\[
J_m(u)
\geq
\underline q\,m^{-r/2}
\int_{\|v\|\leq\rho}
\exp\left\{
-\frac{\tau\overline h_\Phi}{2}\|v\|^2
\right\}
\lambda_{\mathcal V}(dv).
\]
The last integral is a fixed strictly positive finite constant.
Hence there exists $c_J>0$, independent of $m$ and $u$, such that
$J_m(u)\geq c_Jm^{-r/2}.$

For the upper bound, define
\[
\delta_{\Phi,\rho}
=
\inf_{\substack{
u\in\mathcal K,\ \mu\in\mathcal U\\
\|\mu-u\|\geq\rho}}
D_\Phi(u,\mu).
\]
By Assumption~\ref{assump-bregman-response},
$\delta_{\Phi,\rho}>0$. Split
\[
J_m(u)
=
J_m^{\mathrm{loc}}(u)
+
J_m^{\mathrm{far}}(u),
\]
where the two terms integrate respectively over
$\{\|\mu-u\|<\rho\}$ and
$\{\|\mu-u\|\geq\rho\}$.

On the local region,
Lemmas~\ref{lem-app-prior-density} and
\ref{lem-app-local-bregman} give
\[
\begin{aligned}
J_m^{\mathrm{loc}}(u)
&\leq
\overline q
\int_{\mathcal V}
\exp\left\{
-\frac{\tau m h_\Phi}{2}\|\mu-u\|^2
\right\}
\lambda_{\mathcal V}(d\mu)
\\
&=
\overline q
\left(\frac{2\pi}{\tau h_\Phi}\right)^{r/2}
m^{-r/2}.
\end{aligned}
\]
On the nonlocal region,
$J_m^{\mathrm{far}}(u)
\leq
e^{-\tau m\delta_{\Phi,\rho}},$
because $\Pi_{\chi_0,\kappa_0}$ is a probability measure.
Since
$\sup_{x\geq1}
x^{r/2}e^{-\tau\delta_{\Phi,\rho}x}
<\infty,$
there exists $C_{\mathrm{far}}<\infty$ such that
\[
e^{-\tau m\delta_{\Phi,\rho}}
\leq
C_{\mathrm{far}}m^{-r/2},
\qquad m\geq1.
\]
Consequently,
\[
J_m(u)
\leq
\left\{
\overline q
\left(\frac{2\pi}{\tau h_\Phi}\right)^{r/2}
+
C_{\mathrm{far}}
\right\}
m^{-r/2}.
\]
Thus, for some constants
$0<c_J\leq C_J<\infty$ independent of $m$ and $u$,
\[
c_Jm^{-r/2}
\leq
J_m(u)
\leq
C_Jm^{-r/2},
\qquad
m\geq1,\quad u\in\mathcal K.
\]

It remains to prove the partition-score statement.
The preceding bounds imply that
\[
\log J_m(u)
=
-\frac r2\log m+\xi_m(u),
\qquad
|\xi_m(u)|\leq C_\Gamma,
\]
uniformly over $m\geq1$ and $u\in\mathcal K$, for some
$C_\Gamma<\infty$.

By Theorem~\ref{thm-bregman-conjugacy-collapse},
$\log M_\tau(\mathcal P)
=
-\tau W(\mathcal P)
+
\sum_{j=1}^k
\log J_{n_{C_j}}(\bar X_{C_j}),$
and similarly
$\log M_\tau(\widetilde{\mathcal P})
=
-\tau W(\widetilde{\mathcal P})
+
\sum_{\ell=1}^h
\log J_{n_{A_\ell}}(\bar X_{A_\ell}).$
Therefore,
\[
\begin{aligned}
\log\frac{M_\tau(\mathcal P)}
{M_\tau(\widetilde{\mathcal P})}
&=
-\tau\{W(\mathcal P)-W(\widetilde{\mathcal P})\}
-\frac r2
\left\{
\sum_{j=1}^k\log n_{C_j}
-
\sum_{\ell=1}^h\log n_{A_\ell}
\right\}
+
\sum_{j=1}^k
\xi_{n_{C_j}}(\bar X_{C_j})
-
\sum_{\ell=1}^h
\xi_{n_{A_\ell}}(\bar X_{A_\ell}).
\end{aligned}
\]
Because all cells are nonempty,
$1\leq n_{C_j},n_{A_\ell}\leq n$, while
$k,h\leq k_{\max}$ and $k_{\max}$ is fixed. Hence the final two
lines are $O(\log n)$ uniformly over all such pairs of partitions.
Thus
\[
\log\frac{M_\tau(\mathcal P)}
{M_\tau(\widetilde{\mathcal P})}
=
-\tau\{W(\mathcal P)-W(\widetilde{\mathcal P})\}
+
O(\log n),
\]
uniformly over the stated class of partition pairs.
\hfill$\square$

\appsection{Preliminaries for the proof of
Theorem~\ref{thm-partition-contraction}}
\label{app-theorem2-preliminaries}

\appsubsection{Spatial comparator and occupancy bounds}
\label{app-spatial-comparator}

Recall the mesh partition
$\mathcal B_K=\{B_{K,1},\ldots,B_{K,K^2}\}$.
The comparator construction of \citet{huang2025consistent}, under
Assumption~\ref{assump-spatial}, yields a supported partition
$\pi_{0,K}^\star
=
\{D_{1,0,K}^\star,\ldots,D_{k_0,0,K}^\star\}
\in\mathfrak P_K$
with exactly $k_0$ regions. Let
$\mathcal V_{\ell,0,K}^\star
=
\{B\in\mathcal B_K:B\subseteq D_{\ell,0,K}^\star\},
\qquad
\ell=1,\ldots,k_0,$

and define the comparator-error block set
\[
W_{0,K}
=
\bigcup_{\ell=1}^{k_0}
\left\{
B\in\mathcal V_{\ell,0,K}^\star:
B\not\subseteq D_{\ell,0}
\right\}.
\]
The comparator approximation result of
\citet{huang2025consistent} gives, for a constant
$C_{\mathrm{app}}<\infty$,
\begin{equation}
|W_{0,K}|
\leq
C_{\mathrm{app}}K.
\label{eq-app-comparator-error}
\end{equation}

For each mesh block write
$N_{K,m}
=
\#\{i:s_i\in B_{K,m}\},$
and define the occupancy event
$\mathcal A_n
=
\left\{
c_BL_n\leq N_{K,m}\leq C_BL_n
\text{ for every }m=1,\ldots,K^2
\right\},$
where $0<c_B<C_B<\infty$ are fixed constants. The uniform
block-occupancy result of \citet{huang2025consistent} gives
\begin{equation}
P_0(\mathcal A_n)\longrightarrow1.
\label{eq-app-occupancy-probability}
\end{equation}
Consequently, if $\mathcal S\subseteq\mathcal B_K$ is any collection
of blocks and
\[
I(\mathcal S)
=
\{i:s_i\in B\text{ for some }B\in\mathcal S\},
\qquad
N(\mathcal S)=|I(\mathcal S)|,
\]
then, on $\mathcal A_n$,
\begin{equation}
c_B|\mathcal S|L_n
\leq
N(\mathcal S)
\leq
C_B|\mathcal S|L_n.
\label{eq-app-block-union-occupancy}
\end{equation}

\begin{lemma}[Macroscopic comparator regions]
\label{lem-app-macroscopic-comparator}
There exists $v_0>0$ such that, for all sufficiently large $n$,
\[
\min_{1\leq\ell\leq k_0}
|\mathcal V_{\ell,0,K}^\star|
\geq
v_0K^2.
\]
Moreover, on $\mathcal A_n$, there exists $c_{\mathrm{obs}}>0$ such
that, for all sufficiently large $n$,
\[
\min_{1\leq\ell\leq k_0}
|C_{\ell,0,K}^\star|
\geq
c_{\mathrm{obs}}n,
\qquad
C_{\ell,0,K}^\star
=
\{i:s_i\in D_{\ell,0,K}^\star\}.
\]
\end{lemma}

\noindent\textit{Proof.}
For each $\ell$, define
\[
R_{\ell,K}^{+}
=
D_{\ell,0,K}^\star\setminus D_{\ell,0},
\qquad
R_{\ell,K}^{-}
=
D_{\ell,0}\setminus D_{\ell,0,K}^\star.
\]
Up to block boundaries of Lebesgue measure zero,
$\cup_{\ell=1}^{k_0}R_{\ell,K}^{+}$ is contained in the union of
the blocks in $W_{0,K}$. Since each block has area $K^{-2}$,
\eqref{eq-app-comparator-error} gives
\[
\sum_{\ell=1}^{k_0}|R_{\ell,K}^{+}|
\leq
\frac{C_{\mathrm{app}}}{K}.
\]
Both
$\{D_{1,0},\ldots,D_{k_0,0}\}$ and
$\{D_{1,0,K}^\star,\ldots,D_{k_0,0,K}^\star\}$ partition
$\mathcal D$ up to null boundaries, so
\[
\sum_{\ell=1}^{k_0}|R_{\ell,K}^{-}|
=
\sum_{\ell=1}^{k_0}|R_{\ell,K}^{+}|
\leq
\frac{C_{\mathrm{app}}}{K}.
\]
Hence, for every $\ell$,
$|D_{\ell,0,K}^\star|
\geq
|D_{\ell,0}|-|R_{\ell,K}^{-}|
\geq
a_{\min}-\frac{C_{\mathrm{app}}}{K}.$
Since $K\to\infty$, the right-hand side is at least
$a_{\min}/2$ for all sufficiently large $n$. Because
$D_{\ell,0,K}^\star$ is a union of blocks of area $K^{-2}$,
\[
|\mathcal V_{\ell,0,K}^\star|
=
K^2|D_{\ell,0,K}^\star|
\geq
\frac{a_{\min}}{2}K^2.
\]
Thus the first claim holds, for example, with
$v_0=a_{\min}/2$.

On $\mathcal A_n$, \eqref{eq-app-block-union-occupancy} gives
\[
|C_{\ell,0,K}^\star|
\geq
c_B|\mathcal V_{\ell,0,K}^\star|L_n
\geq
c_Bv_0K^2L_n.
\]
Since $K^2L_n\asymp n$, the second claim follows.
\hfill$\square$

An immediate consequence of
\eqref{eq-app-block-union-occupancy} and
\eqref{eq-app-comparator-error} is
\begin{equation}
N(W_{0,K})
\leq
C_BC_{\mathrm{app}}KL_n
\qquad\text{on }\mathcal A_n.
\label{eq-app-comparator-error-count}
\end{equation}

\appsubsection{Bregman refinement bounds}
\label{app-bregman-refinement}

We first record the quadratic comparison needed when both arguments
belong to $\mathcal K$.

\begin{lemma}[Quadratic comparison on $\mathcal K$]
\label{lem-app-bregman-K}
Under Assumption~\ref{assump-bregman-response}, for every
$u,v\in\mathcal K$,
\[
\frac{h_\Phi}{2}\|u-v\|^2
\leq
D_\Phi(u,v)
\leq
\frac{\overline h_\Phi}{2}\|u-v\|^2.
\]
\end{lemma}

\noindent\textit{Proof.}
The set
\[
\mathcal K
=
\{u:\operatorname{dist}(u,\mathcal H_0)\leq\eta\}
\]
is convex because $\mathcal H_0$ is convex. Hence the line segment
$v+t(u-v)$, $0\leq t\leq1$, lies in
$\mathcal K\subseteq\mathcal K_\rho$. Taylor's formula gives
\[
D_\Phi(u,v)
=
\int_0^1(1-t)
\left\langle
u-v,
\nabla^2\Phi\{v+t(u-v)\}(u-v)
\right\rangle dt.
\]
The Hessian bounds in
Assumption~\ref{assump-bregman-response} and
$\int_0^1(1-t)\,dt=1/2$ yield the result.
\hfill$\square$

For sampled-location partitions $\mathcal Q$ and $\mathcal P$, write
$\mathcal Q\preceq\mathcal P$ when $\mathcal Q$ is a refinement of
$\mathcal P$. For $C\in\mathcal P$, let
$\mathcal Q(C)=\{A\in\mathcal Q:A\subseteq C\}.$

\begin{lemma}[Bregman refinement identity and bounds]
\label{lem-app-refinement}
Suppose $\mathcal Q\preceq\mathcal P$ and the empirical centroids of
all cells of $\mathcal P$ and $\mathcal Q$ belong to $\mathcal K$.
Then
\begin{equation}
W(\mathcal P)-W(\mathcal Q)
=
\sum_{C\in\mathcal P}
\sum_{A\in\mathcal Q(C)}
|A|D_\Phi(\bar X_A,\bar X_C)
\geq0.
\label{eq-app-exact-refinement}
\end{equation}
Moreover,
\begin{align}
W(\mathcal P)-W(\mathcal Q)
&\geq
\frac{h_\Phi}{2}
\sum_{C\in\mathcal P}
\frac{1}{|C|}
\sum_{\substack{A,B\in\mathcal Q(C)\\A<B}}
|A||B|
\|\bar X_A-\bar X_B\|^2,
\label{eq-app-refinement-lower}
\\
W(\mathcal P)-W(\mathcal Q)
&\leq
\frac{\overline h_\Phi}{2}
\sum_{C\in\mathcal P}
\frac{1}{|C|}
\sum_{\substack{A,B\in\mathcal Q(C)\\A<B}}
|A||B|
\|\bar X_A-\bar X_B\|^2.
\label{eq-app-refinement-upper}
\end{align}
Here $A<B$ denotes any fixed ordering of the cells in
$\mathcal Q(C)$, so each unordered pair is counted once.
\end{lemma}

\noindent\textit{Proof.}
For $A\in\mathcal Q(C)$, the Bregman centroid decomposition gives
$\sum_{i\in A}D_\Phi(X_i,\bar X_C)
=
W(A)+|A|D_\Phi(\bar X_A,\bar X_C).$
Summing first over $A\in\mathcal Q(C)$ and then over
$C\in\mathcal P$ yields \eqref{eq-app-exact-refinement}.

By Lemma~\ref{lem-app-bregman-K},
\[
\frac{h_\Phi}{2}
\|\bar X_A-\bar X_C\|^2
\leq
D_\Phi(\bar X_A,\bar X_C)
\leq
\frac{\overline h_\Phi}{2}
\|\bar X_A-\bar X_C\|^2.
\]
Finally, since
$\bar X_C
=
\frac{1}{|C|}
\sum_{A\in\mathcal Q(C)}|A|\bar X_A,$
the weighted variance identity gives
\[
\sum_{A\in\mathcal Q(C)}
|A|\|\bar X_A-\bar X_C\|^2
=
\frac{1}{|C|}
\sum_{\substack{A,B\in\mathcal Q(C)\\A<B}}
|A||B|
\|\bar X_A-\bar X_B\|^2.
\]
Combining these displays proves both bounds.
\hfill$\square$

\appsubsection{A uniform high-probability response event}
\label{app-uniform-response-event}

For each observation write
\[
\epsilon_i
=
X_i-\mu_{\ell_0(s_i),0},
\]
and for each block define
\[
R_{K,m}
=
\sum_{i:s_i\in B_{K,m}}\epsilon_i.
\]
Fix a constant $M_F>0$ and let
\[
\mathcal F_n
=
\bigcap_{m=1}^{K^2}
\left\{
\|R_{K,m}\|^2
\leq
M_FN_{K,m}\log n
\right\}.
\]

\begin{lemma}[Uniform block-residual concentration]
\label{lem-app-block-residual}
Under Assumption~\ref{assump-bregman-response}, $M_F$ can be chosen
so that
\[
P_0(\mathcal F_n)\longrightarrow1.
\]
More precisely, if $\sigma_X^2>0$, then for all sufficiently large
$n$ and every realization $s_{1:n}\in\mathcal A_n$,
\[
P_0(\mathcal F_n^c\mid s_{1:n})
\leq
2rK^2
n^{-M_F/(2r\sigma_X^2)}.
\]
Thus any fixed choice
\[
M_F>2r\sigma_X^2
\]
is sufficient.
\end{lemma}

\noindent\textit{Proof.}
Assume first that $\sigma_X^2>0$ and condition on
$s_{1:n}\in\mathcal A_n$. Let
$e_1,\ldots,e_r$ be an orthonormal basis of $\mathcal V$ and define
\[
Z_{m,a}
=
\langle e_a,R_{K,m}\rangle.
\]
Conditional independence and
Assumption~\ref{assump-bregman-response} imply, for
$|t|\leq b_X^{-1}$,
\[
E_0\{\exp(tZ_{m,a})\mid s_{1:n}\}
\leq
\exp\left(
\frac{\sigma_X^2N_{K,m}t^2}{2}
\right).
\]
Set
$x_{m,n}
=
\left(
\frac{M_FN_{K,m}\log n}{r}
\right)^{1/2}.$
The optimizing Chernoff value
$t_{m,n}
=
\frac{x_{m,n}}{\sigma_X^2N_{K,m}}$
satisfies $t_{m,n}\to0$ uniformly in $m$, because on
$\mathcal A_n$,
$N_{K,m}\geq c_BL_n
=
c_B\log^{1+\alpha_b}(n).$
Thus, for all sufficiently large $n$,
$|t_{m,n}|\leq b_X^{-1}$. Chernoff's inequality applied to
$Z_{m,a}$ and $-Z_{m,a}$ yields
\[
P_0\{|Z_{m,a}|>x_{m,n}\mid s_{1:n}\}
\leq
2n^{-M_F/(2r\sigma_X^2)}.
\]
If
$\|R_{K,m}\|^2>M_FN_{K,m}\log n$, at least one of its $r$
coordinates has squared magnitude exceeding
$M_FN_{K,m}\log n/r$. Hence
\[
P_0\left\{
\|R_{K,m}\|^2>M_FN_{K,m}\log n
\,\middle|\,s_{1:n}
\right\}
\leq
2r n^{-M_F/(2r\sigma_X^2)}.
\]
A union bound over the $K^2$ blocks gives the displayed bound for
$\mathcal F_n^c$.

Since $K^2\asymp n/L_n$,
$M_F>2r\sigma_X^2$ implies
$
K^2n^{-M_F/(2r\sigma_X^2)}
\longrightarrow0.
$
Together with
$P_0(\mathcal A_n)\to1$, this proves
$P_0(\mathcal F_n)\to1$.

If $\sigma_X^2=0$, the exponential-moment bound and the conditional
mean-zero property imply, by equality in Jensen's inequality, that
$\langle v,\epsilon_i\rangle=0$ almost surely for every
$v\in\mathcal V$. Hence $\epsilon_i=0$ almost surely and
$\mathcal F_n$ holds with conditional probability one.
\hfill$\square$

Fix henceforth such an $M_F$ and define
\begin{equation}
\mathcal G_n
=
\mathcal A_n\cap\mathcal F_n.
\label{eq-app-Gn}
\end{equation}
By \eqref{eq-app-occupancy-probability} and
Lemma~\ref{lem-app-block-residual},
\begin{equation}
P_0(\mathcal G_n)\longrightarrow1.
\label{eq-app-Gn-probability}
\end{equation}

\appsubsection{Uniform centroid control}
\label{app-centroid-control}

For a nonempty block collection $\mathcal S\subseteq\mathcal B_K$,
define
\[
\bar X_{\mathcal S}
=
\frac{1}{N(\mathcal S)}
\sum_{i\in I(\mathcal S)}X_i,
\qquad
\mu_{\mathcal S}^{\circ}
=
\frac{1}{N(\mathcal S)}
\sum_{i\in I(\mathcal S)}
\mu_{\ell_0(s_i),0}.
\]
Also define
\begin{equation}
a_n
=
\left(\frac{M_F}{c_B}\right)^{1/2}
\log^{-\alpha_b/2}(n).
\label{eq-app-an}
\end{equation}

\begin{lemma}[Uniform block-union centroid control]
\label{lem-app-union-centroid}
On $\mathcal G_n$, every nonempty
$\mathcal S\subseteq\mathcal B_K$ satisfies
\[
\|\bar X_{\mathcal S}-\mu_{\mathcal S}^{\circ}\|
\leq
a_n.
\]
Consequently, for all sufficiently large $n$, on $\mathcal G_n$,
\[
\bar X_{\mathcal S}\in\mathcal K
\]
simultaneously for every nonempty block collection
$\mathcal S\subseteq\mathcal B_K$.
\end{lemma}

\noindent\textit{Proof.}
On $\mathcal G_n$,
\[
\bar X_{\mathcal S}-\mu_{\mathcal S}^{\circ}
=
\frac{1}{N(\mathcal S)}
\sum_{B_{K,m}\in\mathcal S}R_{K,m}
=
\sum_{B_{K,m}\in\mathcal S}
\frac{N_{K,m}}{N(\mathcal S)}
\frac{R_{K,m}}{N_{K,m}}.
\]
The coefficients in the final expression are nonnegative and sum to
one. Therefore
\[
\|\bar X_{\mathcal S}-\mu_{\mathcal S}^{\circ}\|
\leq
\max_{B_{K,m}\in\mathcal S}
\frac{\|R_{K,m}\|}{N_{K,m}}.
\]
On $\mathcal F_n$ and $\mathcal A_n$,
\[
\frac{\|R_{K,m}\|}{N_{K,m}}
\leq
\left(
\frac{M_F\log n}{N_{K,m}}
\right)^{1/2}
\leq
\left(
\frac{M_F\log n}{c_BL_n}
\right)^{1/2}
=
a_n.
\]
This proves the first claim.

The vector $\mu_{\mathcal S}^{\circ}$ is a convex combination of the
true regional centers, so
$\mu_{\mathcal S}^{\circ}\in\mathcal H_0$. Hence
\[
\operatorname{dist}
(\bar X_{\mathcal S},\mathcal H_0)
\leq a_n.
\]
Since $a_n\to0$, we have $a_n\leq\eta$ for all sufficiently large
$n$, and therefore
$\bar X_{\mathcal S}\in\mathcal K$.
\hfill$\square$

In particular, on $\mathcal G_n$, for all sufficiently large $n$,
the empirical centroid of every candidate region, every comparator
region, and every nonempty intersection of such regions belongs to
$\mathcal K$. Hence Proposition~\ref{prop-uniform-integrated-center}
and Lemma~\ref{lem-app-refinement} apply uniformly over the supported
partition space.

\appsubsection{Approximate purity of comparator-overlap cells}
\label{app-overlap-purity}

Fix
\[
\pi^\star
=
\{D_1^\star,\ldots,D_k^\star\}
\in\mathfrak P_K,
\]
and let $\mathcal V_j$ be the block set of $D_j^\star$.
For $b_{j\ell}>0$, define
\[
\mathcal V_{j\ell}
=
\mathcal V_j\cap
\mathcal V_{\ell,0,K}^\star,
\qquad
C_{j\ell}
=
I(\mathcal V_{j\ell}),
\qquad
N_{j\ell}
=
|C_{j\ell}|,
\]
and
\[
\bar X_{j\ell}
=
\frac{1}{N_{j\ell}}
\sum_{i\in C_{j\ell}}X_i,
\qquad
\mu_{j\ell}^{\circ}
=
\frac{1}{N_{j\ell}}
\sum_{i\in C_{j\ell}}
\mu_{\ell_0(s_i),0}.
\]
Let
\[
N_{j\ell}^{\mathrm{bad}}
=
N(\mathcal V_{j\ell}\cap W_{0,K}),
\qquad
\Delta_\mu
=
\max_{1\leq\ell,h\leq k_0}
\|\mu_{\ell,0}-\mu_{h,0}\|.
\]

\begin{lemma}[Approximate purity]
\label{lem-app-overlap-purity}
On $\mathcal G_n$, uniformly over
$\pi^\star\in\mathfrak P_K$ and all pairs $(j,\ell)$ satisfying
$b_{j\ell}>0$,
\[
\|\bar X_{j\ell}-\mu_{\ell,0}\|
\leq
a_n+C_{\mathrm{pur}}\frac{K}{b_{j\ell}},
\]
where
\[
C_{\mathrm{pur}}
=
\Delta_\mu
\frac{C_BC_{\mathrm{app}}}{c_B}.
\]
\end{lemma}

\noindent\textit{Proof.}
By Lemma~\ref{lem-app-union-centroid},
\[
\|\bar X_{j\ell}-\mu_{j\ell}^{\circ}\|
\leq
a_n.
\]
Any block in
$\mathcal V_{j\ell}\setminus W_{0,K}$ is wholly contained in
$D_{\ell,0}$. Hence observations outside
$\mathcal V_{j\ell}\cap W_{0,K}$ have true center exactly
$\mu_{\ell,0}$. Consequently,
$\|\mu_{j\ell}^{\circ}-\mu_{\ell,0}\|
\leq
\Delta_\mu
\frac{N_{j\ell}^{\mathrm{bad}}}{N_{j\ell}}.$
On $\mathcal A_n$,
$N_{j\ell}
\geq
c_Bb_{j\ell}L_n,$
while \eqref{eq-app-comparator-error-count} gives
\[
N_{j\ell}^{\mathrm{bad}}
\leq
N(W_{0,K})
\leq
C_BC_{\mathrm{app}}KL_n.
\]
Therefore
\[
\|\mu_{j\ell}^{\circ}-\mu_{\ell,0}\|
\leq
C_{\mathrm{pur}}\frac{K}{b_{j\ell}}.
\]
The result follows by the triangle inequality.
\hfill$\square$

\begin{corollary}[Separation of large distinct-label overlap cells]
\label{cor-app-overlap-separation}
Let $d_n>0$ satisfy
\[
\frac{K}{d_n}\longrightarrow0.
\]
On $\mathcal G_n$, suppose a candidate region $j$ has overlaps with
two distinct comparator labels $\ell\neq h$ satisfying
\[
b_{j\ell}\geq d_n,
\qquad
b_{jh}\geq d_n.
\]
Then, for all sufficiently large $n$,
\[
\|\bar X_{j\ell}-\bar X_{jh}\|
\geq
\frac{\delta_0}{2}.
\]
\end{corollary}

\noindent\textit{Proof.}
By Lemma~\ref{lem-app-overlap-purity},
\[
\|\bar X_{j\ell}-\mu_{\ell,0}\|
\leq
a_n+C_{\mathrm{pur}}\frac{K}{d_n},
\qquad
\|\bar X_{jh}-\mu_{h,0}\|
\leq
a_n+C_{\mathrm{pur}}\frac{K}{d_n}.
\]
Since $\ell\neq h$,
Assumption~\ref{assump-bregman-response} gives
\[
\|\mu_{\ell,0}-\mu_{h,0}\|
\geq\delta_0.
\]
Hence
\[
\|\bar X_{j\ell}-\bar X_{jh}\|
\geq
\delta_0
-
2a_n
-
2C_{\mathrm{pur}}\frac{K}{d_n}.
\]
Because $a_n\to0$ and $K/d_n\to0$, the right-hand side is at least
$\delta_0/2$ for all sufficiently large $n$.
\hfill$\square$

\appsubsection{Common refinement with the comparator}
\label{app-common-refinement}

For a supported candidate $\pi^\star\in\mathfrak P_K$, define the
sampled comparator partition
\[
\mathcal P_{0,K}^\star
=
\{C_{1,0,K}^\star,\ldots,C_{k_0,0,K}^\star\}
\]
and its common refinement with $\mathcal P(\pi^\star)$ by
\begin{equation}
\mathcal P_e(\pi^\star)
=
\left\{
C_j(\pi^\star)\cap C_{\ell,0,K}^\star:
C_j(\pi^\star)\cap C_{\ell,0,K}^\star\neq\varnothing
\right\}.
\label{eq-app-common-refinement}
\end{equation}
By construction,
\[
\mathcal P_e(\pi^\star)
\preceq
\mathcal P(\pi^\star),
\qquad
\mathcal P_e(\pi^\star)
\preceq
\mathcal P_{0,K}^\star.
\]
On $\mathcal G_n$, for all sufficiently large $n$, all empirical
centroids appearing in these three partitions belong to $\mathcal K$.
Therefore the refinement identities and bounds in
Lemma~\ref{lem-app-refinement} may be applied uniformly to every
$\pi^\star\in\mathfrak P_K$.

For later use, define the empirical fit gained by refining the
comparator according to the candidate as
\begin{equation}
S_0(\pi^\star)
=
W(\mathcal P_{0,K}^\star)
-
W\{\mathcal P_e(\pi^\star)\}
\geq0.
\label{eq-app-S0}
\end{equation}

\appsection{Score separation for the bad partition classes}
\label{app-bad-partition-scores}

\appsubsection{Bad partition classes}
\label{app-bad-partition-classes}

For use in the proof, decompose the complement of
$\Pi_{5,K}^\star$ into four classes. Define
\[
\Pi_{1,K}^\star
=
\left\{
\pi^\star\in\mathfrak P_K:
|\pi^\star|<k_0
\right\},
\]
\[
\Pi_{2,K}^\star
=
\left\{
\pi^\star\in\mathfrak P_K:
|\pi^\star|>k_0
\right\},
\]
and
\[
\Pi_{3,K}^\star
=
\left\{
\pi^\star\in\mathfrak P_K:
|\pi^\star|=k_0,\,
\{m_{\pi^\star}(j):1\leq j\leq k_0\}
\neq
\{1,\ldots,k_0\}
\right\}.
\]
Finally, let
\[
\Pi_{4,K}^\star
=
\left\{
\begin{array}{l}
\pi^\star\in\mathfrak P_K:
|\pi^\star|=k_0,\,
\{m_{\pi^\star}(j):1\leq j\leq k_0\}
=
\{1,\ldots,k_0\},\\[2mm]
q_K(\pi^\star)\geq q_n
\end{array}
\right\}.
\]
Thus
\[
\mathfrak P_K
=
\Pi_{1,K}^\star
\,\dot\cup\,
\Pi_{2,K}^\star
\,\dot\cup\,
\Pi_{3,K}^\star
\,\dot\cup\,
\Pi_{4,K}^\star
\,\dot\cup\,
\Pi_{5,K}^\star.
\]

For
$\pi^\star\in\Pi_{4,K}^\star\cup\Pi_{5,K}^\star$,
the best-match map is a bijection. After relabelling the candidate
regions we may therefore take
\[
m_{\pi^\star}(j)=j,
\qquad j=1,\ldots,k_0,
\]
in which case
\begin{equation}
q_K(\pi^\star)
=
\sum_{j=1}^{k_0}
\sum_{\ell\neq j}b_{j\ell}.
\label{eq-app-q-diagonal}
\end{equation}

\appsubsection{Comparator-refinement gain}
\label{app-comparator-refinement-gain}

We first control the quantity $S_0(\pi^\star)$ defined in
\eqref{eq-app-S0}.

\begin{suppproposition}[Global comparator-refinement gain]
\label{prop-app-global-refinement-gain}
On $\mathcal G_n$, for all sufficiently large $n$, uniformly over
$\pi^\star\in\mathfrak P_K$,
\[
0\leq S_0(\pi^\star)
\leq
C_{\mathrm{ref}}
\left\{
n\log^{-\alpha_b}(n)+KL_n
\right\}
\]
for some constant $C_{\mathrm{ref}}<\infty$. Consequently,
\[
\sup_{\pi^\star\in\mathfrak P_K}
\frac{S_0(\pi^\star)}{n}
\longrightarrow0.
\]
\end{suppproposition}

\noindent\textit{Proof.}
The nonnegativity follows because
$\mathcal P_e(\pi^\star)$ refines
$\mathcal P_{0,K}^\star$.

For a comparator label $\ell$, let
$\bar X_{0\ell}
=
\frac{1}{|C_{\ell,0,K}^\star|}
\sum_{i\in C_{\ell,0,K}^\star}X_i.$
By the upper refinement bound in
Lemma~\ref{lem-app-refinement},
$S_0(\pi^\star)
\leq
\frac{\overline h_\Phi}{2}
\sum_{\ell=1}^{k_0}
\sum_{j:b_{j\ell}>0}
N_{j\ell}
\|\bar X_{j\ell}-\bar X_{0\ell}\|^2.$
Since $\bar X_{0\ell}$ is the weighted Euclidean mean of the
$\bar X_{j\ell}$ within comparator region $\ell$,
\[
\sum_{j:b_{j\ell}>0}
N_{j\ell}
\|\bar X_{j\ell}-\bar X_{0\ell}\|^2
\leq
\sum_{j:b_{j\ell}>0}
N_{j\ell}
\|\bar X_{j\ell}-\mu_{\ell,0}\|^2.
\]
Using
$\bar X_{j\ell}-\mu_{\ell,0}
=
(\bar X_{j\ell}-\mu_{j\ell}^{\circ})
+
(\mu_{j\ell}^{\circ}-\mu_{\ell,0}),$
Lemma~\ref{lem-app-union-centroid}, and the contamination argument
from Lemma~\ref{lem-app-overlap-purity},
\[
N_{j\ell}
\|\bar X_{j\ell}-\mu_{\ell,0}\|^2
\leq
2N_{j\ell}a_n^2
+
2\Delta_\mu^2N_{j\ell}^{\mathrm{bad}}.
\]
Summing over all overlap cells gives
\[
S_0(\pi^\star)
\leq
C
\left\{
na_n^2+N(W_{0,K})
\right\}.
\]
By \eqref{eq-app-an} and
\eqref{eq-app-comparator-error-count},
\[
na_n^2
=
\frac{M_F}{c_B}
n\log^{-\alpha_b}(n),
\qquad
N(W_{0,K})
\leq
C_BC_{\mathrm{app}}KL_n.
\]
This proves the first claim. Finally,
$K^2L_n\asymp n$ implies
$KL_n/n\asymp K^{-1}\to0$, while
$\log^{-\alpha_b}(n)\to0$, proving the second claim.
\hfill$\square$

For $\pi^\star\in\Pi_{4,K}^\star$, use the diagonal relabelling and
write
\[
q=q_K(\pi^\star),
\qquad
q_\ell=\sum_{j\neq\ell}b_{j\ell},
\qquad
w_\ell
=
|\mathcal V_{\ell,0,K}^\star\cap W_{0,K}|.
\]
Then
\begin{equation}
\sum_{\ell=1}^{k_0}q_\ell=q,
\qquad
\sum_{\ell=1}^{k_0}w_\ell
\leq
C_{\mathrm{app}}K.
\label{eq-app-q-w-sums}
\end{equation}

\begin{suppproposition}[Mismatch-adaptive refinement gain]
\label{prop-app-adaptive-refinement-gain}
Uniformly over $\pi^\star\in\Pi_{4,K}^\star$, on $\mathcal G_n$,
\[
S_0(\pi^\star)
=
o\{q_K(\pi^\star)L_n\}.
\]
\end{suppproposition}

\noindent\textit{Proof.}
Write $q=q_K(\pi^\star)$. First suppose
\[
q\leq\frac{v_0}{2}K^2.
\]
For each comparator label $\ell$,
\[
b_{\ell\ell}
=
|\mathcal V_{\ell,0,K}^\star|-q_\ell
\geq
\frac{v_0}{2}K^2.
\]
Define the pure diagonal block collection
$\mathcal A_\ell
=
\mathcal V_{\ell\ell}\setminus W_{0,K}.$
Since $|W_{0,K}|=O(K)$, $\mathcal A_\ell$ is nonempty for all
sufficiently large $n$. Let $A_\ell=I(\mathcal A_\ell)$ and let
$\bar X_{A_\ell}$ be its empirical centroid. Every observation in
$A_\ell$ has true center $\mu_{\ell,0}$, and therefore
Lemma~\ref{lem-app-union-centroid} gives
$\|\bar X_{A_\ell}-\mu_{\ell,0}\|
\leq a_n.$

Using $\bar X_{A_\ell}$ as the comparison point in the weighted
variance bound underlying Lemma~\ref{lem-app-refinement}, the
refinement gain within comparator region $\ell$ is bounded by
\[
S_{0,\ell}
\leq
C
\sum_{j:b_{j\ell}>0}
N_{j\ell}
\|\bar X_{j\ell}-\bar X_{A_\ell}\|^2.
\]
For the diagonal overlap cell, only observations lying in the
$w_\ell$ comparator-error blocks can move its centroid away from
$\bar X_{A_\ell}$. Since all relevant empirical centroids lie in
the compact set $\mathcal K$ on $\mathcal G_n$,
\[
N_{\ell\ell}
\|\bar X_{\ell\ell}-\bar X_{A_\ell}\|^2
\leq
Cw_\ell L_n.
\]
For $j\neq\ell$, using
Lemma~\ref{lem-app-union-centroid} and the same contamination bound
as in Lemma~\ref{lem-app-overlap-purity},
\[
N_{j\ell}
\|\bar X_{j\ell}-\bar X_{A_\ell}\|^2
\leq
C N_{j\ell}a_n^2
+
C N_{j\ell}^{\mathrm{bad}}.
\]
Therefore
$S_{0,\ell}
\leq
Cq_\ell L_na_n^2
+
Cw_\ell L_n.$
Since $L_na_n^2=O(\log n)$, summing over $\ell$ and using
\eqref{eq-app-q-w-sums} gives
\begin{equation}
S_0(\pi^\star)
\leq
C\{q\log n+KL_n\}.
\label{eq-app-S0-moderate}
\end{equation}

Because $\pi^\star\in\Pi_{4,K}^\star$,
$q\geq q_n$, and for all sufficiently large $n$,
$q\geq\frac12K\log^{\alpha_0}(n).$
Hence, in the present case,
\[
\frac{S_0(\pi^\star)}{qL_n}
\leq
C\frac{\log n}{L_n}
+
C\frac{K}{q}
\leq
C\log^{-\alpha_b}(n)
+
C\log^{-\alpha_0}(n)
\longrightarrow0.
\]

Now suppose
$q>\frac{v_0}{2}K^2.$
Proposition~\ref{prop-app-global-refinement-gain} gives
$S_0(\pi^\star)
\leq
C\{n\log^{-\alpha_b}(n)+KL_n\}.$
Since $n\asymp K^2L_n$,
$n\log^{-\alpha_b}(n)
\asymp
K^2\log n.$
Therefore
\[
\frac{S_0(\pi^\star)}{qL_n}
\leq
C\frac{K^2}{q}\frac{\log n}{L_n}
+
C\frac{K}{q}
\leq
C\log^{-\alpha_b}(n)
+
\frac{C}{K}
\longrightarrow0.
\]
The two cases prove the result.
\hfill$\square$

\appsubsection{Underfitted and missing-label partitions}
\label{app-underfit-missing-label}

\begin{lemma}[Underfitting forces macroscopic pooling]
\label{lem-app-underfit-pooling}
For every $\pi^\star\in\Pi_{1,K}^\star$, there exist a candidate
index $j^\star$ and distinct comparator labels
$\ell_1\neq\ell_2$ such that
\[
b_{j^\star\ell_1}
\geq
c_\star K^2,
\qquad
b_{j^\star\ell_2}
\geq
c_\star K^2,
\qquad
c_\star=\frac{v_0}{k_{\max}}>0.
\]
\end{lemma}

\noindent\textit{Proof.}
Let $k=|\pi^\star|<k_0$. For every comparator label $\ell$,
$\sum_{j=1}^k b_{j\ell}
=
|\mathcal V_{\ell,0,K}^\star|
\geq
v_0K^2$
by Lemma~\ref{lem-app-macroscopic-comparator}. Hence there exists
a candidate index $J(\ell)$ satisfying
\[
b_{J(\ell),\ell}
\geq
\frac{v_0}{k}K^2
\geq
\frac{v_0}{k_{\max}}K^2.
\]
Since there are $k_0$ comparator labels but only $k<k_0$
candidate indices, two distinct comparator labels have the same
value of $J(\ell)$. This proves the result.
\hfill$\square$

\begin{suppproposition}[Loss separation under underfitting]
\label{prop-app-underfit-loss}
There exists $c_{1,W}>0$ such that, on $\mathcal G_n$, for all
sufficiently large $n$, uniformly over
$\pi^\star\in\Pi_{1,K}^\star$,
\[
W\{\mathcal P(\pi^\star)\}
-
W(\mathcal P_{0,K}^\star)
\geq
c_{1,W}n.
\]
\end{suppproposition}

\noindent\textit{Proof.}
By Lemma~\ref{lem-app-underfit-pooling}, one candidate cluster
contains two overlap cells from distinct comparator labels, each
containing at least $c_\star K^2$ blocks. Taking
$d_n=c_\star K^2$ in
Corollary~\ref{cor-app-overlap-separation} gives
$\|\bar X_{j^\star\ell_1}
-
\bar X_{j^\star\ell_2}\|
\geq
\frac{\delta_0}{2}$
for all sufficiently large $n$.

On $\mathcal A_n$, both overlap cells contain at least
$c_Bc_\star K^2L_n\asymp n$ observations. Hence the lower
refinement bound
\eqref{eq-app-refinement-lower}, retaining only these two cells
inside candidate cluster $j^\star$, gives
\[
W\{\mathcal P(\pi^\star)\}
-
W\{\mathcal P_e(\pi^\star)\}
\geq
c_{\mathrm{pool},1}n
\]
for some $c_{\mathrm{pool},1}>0$ independent of
$\pi^\star$.

Since
\[
W\{\mathcal P(\pi^\star)\}
-
W(\mathcal P_{0,K}^\star)
=
\left[
W\{\mathcal P(\pi^\star)\}
-
W\{\mathcal P_e(\pi^\star)\}
\right]
-
S_0(\pi^\star),
\]
and Proposition~\ref{prop-app-global-refinement-gain} gives
$S_0(\pi^\star)=o(n)$ uniformly, the conclusion follows.
\hfill$\square$

\begin{lemma}[A missing best-match label forces macroscopic pooling]
\label{lem-app-missing-label-pooling}
For every $\pi^\star\in\Pi_{3,K}^\star$, there exist a candidate
index $j^\star$ and distinct comparator labels
$\ell_0\neq\ell_1$ such that
\[
m_{\pi^\star}(j^\star)=\ell_1,
\qquad
b_{j^\star\ell_0}
\geq
\frac{v_0}{k_0}K^2,
\qquad
b_{j^\star\ell_1}
\geq
\frac{v_0}{k_0}K^2.
\]
\end{lemma}

\noindent\textit{Proof.}
Because $\pi^\star\in\Pi_{3,K}^\star$, some comparator label
$\ell_0$ is not selected by the best-match map. Since
$\sum_{j=1}^{k_0}b_{j\ell_0}
=
|\mathcal V_{\ell_0,0,K}^\star|
\geq
v_0K^2,$
there exists $j^\star$ such that
$b_{j^\star\ell_0}
\geq
\frac{v_0}{k_0}K^2.$
Set $\ell_1=m_{\pi^\star}(j^\star)$. Then
$\ell_1\neq\ell_0$, and by the definition of the best match,
\[
b_{j^\star\ell_1}
\geq
b_{j^\star\ell_0}.
\]
This proves the claim.
\hfill$\square$

\begin{suppproposition}[Loss separation for a missing label]
\label{prop-app-missing-label-loss}
There exists $c_{3,W}>0$ such that, on $\mathcal G_n$, for all
sufficiently large $n$, uniformly over
$\pi^\star\in\Pi_{3,K}^\star$,
\[
W\{\mathcal P(\pi^\star)\}
-
W(\mathcal P_{0,K}^\star)
\geq
c_{3,W}n.
\]
\end{suppproposition}

\noindent\textit{Proof.}
Lemma~\ref{lem-app-missing-label-pooling} gives two distinct-label
overlap cells in the same candidate cluster, each containing at
least $(v_0/k_0)K^2$ blocks. Hence
Corollary~\ref{cor-app-overlap-separation} separates their empirical
centroids by at least $\delta_0/2$ for all sufficiently large $n$.

The occupancy event implies that both cells contain order $n$
observations. The lower refinement bound
\eqref{eq-app-refinement-lower} therefore gives
\[
W\{\mathcal P(\pi^\star)\}
-
W\{\mathcal P_e(\pi^\star)\}
\geq
c_{\mathrm{pool},3}n
\]
for some $c_{\mathrm{pool},3}>0$. Subtracting
$S_0(\pi^\star)=o(n)$ uniformly by
Proposition~\ref{prop-app-global-refinement-gain} proves the result.
\hfill$\square$

\appsubsection{Wrong-boundary partitions}
\label{app-wrong-boundary}

\begin{lemma}[A wrong-boundary partition contains two large
distinct-label overlap cells]
\label{lem-app-wrong-boundary-pooling}
Suppose $k_0\geq2$. For every
$\pi^\star\in\Pi_{4,K}^\star$, after the diagonal relabelling,
there exist $j^\star$ and $\ell^\star\neq j^\star$ such that,
writing $q=q_K(\pi^\star)$,
\[
b_{j^\star\ell^\star}
\geq
\frac{q}{k_0(k_0-1)},
\qquad
b_{j^\star j^\star}
\geq
b_{j^\star\ell^\star}.
\]
\end{lemma}

\noindent\textit{Proof.}
After relabelling,
$q
=
\sum_{j=1}^{k_0}\sum_{\ell\neq j}b_{j\ell}.$
There are $k_0(k_0-1)$ off-diagonal terms, so at least one is at
least their average. This gives
\[
b_{j^\star\ell^\star}
\geq
\frac{q}{k_0(k_0-1)}.
\]
Since the diagonal entry in row $j^\star$ is its best match,
$b_{j^\star j^\star}
\geq
b_{j^\star\ell^\star}.$
\hfill$\square$

\begin{suppproposition}[Mismatch-scaled loss separation]
\label{prop-app-wrong-boundary-loss}
There exists $c_{4,W}>0$ such that, on $\mathcal G_n$, for all
sufficiently large $n$, uniformly over
$\pi^\star\in\Pi_{4,K}^\star$,
\[
W\{\mathcal P(\pi^\star)\}
-
W(\mathcal P_{0,K}^\star)
\geq
c_{4,W}q_K(\pi^\star)L_n.
\]
\end{suppproposition}

\noindent\textit{Proof.}
If $k_0=1$, then $\Pi_{4,K}^\star$ is empty, so assume
$k_0\geq2$. Fix
$\pi^\star\in\Pi_{4,K}^\star$ and write
$q=q_K(\pi^\star)$.

By Lemma~\ref{lem-app-wrong-boundary-pooling}, the same candidate
cluster contains two distinct-label overlap cells with at least
\[
d_n
=
\frac{q}{k_0(k_0-1)}
\]
blocks each. Since $q\geq q_n$,
\[
\frac{K}{d_n}
\leq
C\frac{K}{q_n}
\leq
C\log^{-\alpha_0}(n)
\longrightarrow0.
\]
Corollary~\ref{cor-app-overlap-separation} therefore gives
\[
\|\bar X_{j^\star j^\star}
-
\bar X_{j^\star\ell^\star}\|
\geq
\frac{\delta_0}{2}.
\]

On $\mathcal A_n$,
\[
N_{j^\star j^\star}
\geq
c_Bb_{j^\star j^\star}L_n,
\qquad
N_{j^\star\ell^\star}
\geq
c_Bb_{j^\star\ell^\star}L_n.
\]
Moreover, because $b_{j^\star j^\star}$ is the largest overlap
count in row $j^\star$,
\[
|C_{j^\star}(\pi^\star)|
\leq
C_Bk_0b_{j^\star j^\star}L_n.
\]
Consequently,
\[
\frac{
N_{j^\star j^\star}
N_{j^\star\ell^\star}
}{
|C_{j^\star}(\pi^\star)|
}
\geq
c\,b_{j^\star\ell^\star}L_n
\geq
c'qL_n
\]
for fixed positive constants $c,c'$.

Applying the lower refinement bound
\eqref{eq-app-refinement-lower} to these two overlap cells yields
\[
W\{\mathcal P(\pi^\star)\}
-
W\{\mathcal P_e(\pi^\star)\}
\geq
c_{\mathrm{pool},4}qL_n
\]
for some $c_{\mathrm{pool},4}>0$. Finally,
Proposition~\ref{prop-app-adaptive-refinement-gain} gives
\[
S_0(\pi^\star)=o(qL_n)
\]
uniformly over $\Pi_{4,K}^\star$. Hence
\[
W\{\mathcal P(\pi^\star)\}
-
W(\mathcal P_{0,K}^\star)
\geq
c_{4,W}qL_n
\]
for some $c_{4,W}>0$ and all sufficiently large $n$.
\hfill$\square$

\appsubsection{Collapsed-score bounds}
\label{app-collapsed-bad-class-bounds}

Define the auxiliary response-and-cluster-number factor
\begin{equation}
L_{\tau,n}(\pi^\star)
=
\frac{
M_\tau\{\mathcal P(\pi^\star)\}
\lambda_n^{|\pi^\star|}
}{
M_\tau(\mathcal P_{0,K}^\star)
\lambda_n^{k_0}
}.
\label{eq-app-Ltau}
\end{equation}
The remaining spanning-tree, edge-cut, and bounded factorial terms
in the full prior ratio will be handled in
Section~\ref{app-proof-theorem2}.

On $\mathcal G_n$, Proposition~\ref{prop-uniform-integrated-center}
gives, uniformly over $\pi^\star\in\mathfrak P_K$,
\begin{equation}
\log
\frac{
M_\tau\{\mathcal P(\pi^\star)\}
}{
M_\tau(\mathcal P_{0,K}^\star)
}
=
-\tau
\left[
W\{\mathcal P(\pi^\star)\}
-
W(\mathcal P_{0,K}^\star)
\right]
+
R_n(\pi^\star),
\qquad
|R_n(\pi^\star)|
\leq
C_R\log n
\label{eq-app-score-ratio-remainder}
\end{equation}
for some $C_R<\infty$.

\begin{suppproposition}[Uniform score bounds for the four bad classes]
\label{prop-app-classwise-score-bounds}
There exists $c_\star>0$ such that, on $\mathcal G_n$, for all
sufficiently large $n$,
\[
L_{\tau,n}(\pi^\star)
\leq
\begin{cases}
\exp(-c_\star n),
&
\pi^\star\in\Pi_{1,K}^\star,
\\[1mm]
\exp\{-c_\star n\log^{-\alpha_p}(n)\},
&
\pi^\star\in\Pi_{2,K}^\star,
\\[1mm]
\exp(-c_\star n),
&
\pi^\star\in\Pi_{3,K}^\star,
\\[1mm]
\exp\{-c_\star q_K(\pi^\star)L_n\},
&
\pi^\star\in\Pi_{4,K}^\star.
\end{cases}
\]
All bounds are uniform over their indicated classes.
\end{suppproposition}

\noindent\textit{Proof.}

\medskip
\noindent\emph{Class $\Pi_{1,K}^\star$.}
Let $k=|\pi^\star|<k_0$. By
Proposition~\ref{prop-app-underfit-loss},
$W\{\mathcal P(\pi^\star)\}
-
W(\mathcal P_{0,K}^\star)
\geq
c_{1,W}n.$
Using \eqref{eq-app-score-ratio-remainder},
\[
\log L_{\tau,n}(\pi^\star)
\leq
-\tau c_{1,W}n
+
C_R\log n
+
(k-k_0)\log\lambda_n.
\]
Since $\lambda_n\to0$,
\[
(k-k_0)\log\lambda_n
=
(k_0-k)\log(\lambda_n^{-1})
=
O\{n\log^{-\alpha_p}(n)\}
=
o(n),
\]
because $k_0$ is fixed and $\alpha_p>0$. Hence
\[
\log L_{\tau,n}(\pi^\star)
\leq
-c_1n
\]
uniformly over $\Pi_{1,K}^\star$ for some $c_1>0$.

\medskip
\noindent\emph{Class $\Pi_{3,K}^\star$.}
Here $|\pi^\star|=k_0$, so the $\lambda_n$ term cancels.
Proposition~\ref{prop-app-missing-label-loss} and
\eqref{eq-app-score-ratio-remainder} give
\[
\log L_{\tau,n}(\pi^\star)
\leq
-\tau c_{3,W}n+C_R\log n
\leq
-c_3n
\]
for some $c_3>0$ and all sufficiently large $n$.

\medskip
\noindent\emph{Class $\Pi_{4,K}^\star$.}
Again $|\pi^\star|=k_0$. Writing
$q=q_K(\pi^\star)$,
Proposition~\ref{prop-app-wrong-boundary-loss} gives
\[
\log L_{\tau,n}(\pi^\star)
\leq
-\tau c_{4,W}qL_n+C_R\log n.
\]
Since $q\geq q_n$ and, for all sufficiently large $n$,
\[
q_n
\geq
\frac12K\log^{\alpha_0}(n),
\]
we have
\[
\frac{\log n}{qL_n}
\leq
\frac{C}{
K\log^{\alpha_0+\alpha_b}(n)
}
\longrightarrow0.
\]
Therefore
\[
\log L_{\tau,n}(\pi^\star)
\leq
-c_4q_K(\pi^\star)L_n
\]
uniformly over $\Pi_{4,K}^\star$ for some $c_4>0$.

\medskip
\noindent\emph{Class $\Pi_{2,K}^\star$.}
Since $\mathcal P_e(\pi^\star)$ refines
$\mathcal P(\pi^\star)$,
Lemma~\ref{lem-app-refinement} gives
$W\{\mathcal P(\pi^\star)\}
\geq
W\{\mathcal P_e(\pi^\star)\}.$
Hence
$W(\mathcal P_{0,K}^\star)
-
W\{\mathcal P(\pi^\star)\}
\leq
S_0(\pi^\star).$
By Proposition~\ref{prop-app-global-refinement-gain},
$W(\mathcal P_{0,K}^\star)
-
W\{\mathcal P(\pi^\star)\}
\leq
C
\left\{
n\log^{-\alpha_b}(n)+KL_n
\right\}.$
Since $K^2L_n\asymp n$,
$KL_n
=
o\{n\log^{-\alpha_b}(n)\},$
and therefore
\begin{equation}
W(\mathcal P_{0,K}^\star)
-
W\{\mathcal P(\pi^\star)\}
\leq
C_{2,W}n\log^{-\alpha_b}(n)
\label{eq-app-overfit-improvement}
\end{equation}
uniformly over $\Pi_{2,K}^\star$.

Let $k=|\pi^\star|>k_0$. From
\eqref{eq-app-score-ratio-remainder} and
\eqref{eq-app-overfit-improvement},
\[
\log L_{\tau,n}(\pi^\star)
\leq
\tau C_{2,W}n\log^{-\alpha_b}(n)
+
C_R\log n
+
(k-k_0)\log\lambda_n.
\]
Since $k-k_0\geq1$ and $\lambda_n<1$ for all sufficiently large
$n$,
$(k-k_0)\log\lambda_n
\leq
-\log(\lambda_n^{-1}).$
By Assumption~\ref{assump-spatial},
there exists $c_\lambda>0$ such that
$\log(\lambda_n^{-1})
\geq
c_\lambda n\log^{-\alpha_p}(n)$
for all sufficiently large $n$. Because
$0<\alpha_p<\alpha_b$,
\[
n\log^{-\alpha_b}(n)
=
o\{n\log^{-\alpha_p}(n)\},
\qquad
\log n
=
o\{n\log^{-\alpha_p}(n)\}.
\]
Thus
\[
\log L_{\tau,n}(\pi^\star)
\leq
-c_2n\log^{-\alpha_p}(n)
\]
uniformly over $\Pi_{2,K}^\star$ for some $c_2>0$.

Taking
\[
c_\star
=
\min\{c_1,c_2,c_3,c_4\}>0
\]
completes the proof.
\hfill$\square$


\appsection{Proof of Theorem~\ref{thm-partition-contraction}}
\label{app-proof-theorem2}

\appsubsection{Spatial prior and counting bounds}
\label{app-spatial-counting-bounds}

We first record the response-free spatial bounds from
\citet{huang2025consistent} that are used below. These depend only
on the spanning-tree partition construction and not on the response
model.

For a supported partition $\pi^\star\in\mathfrak P_K$, write
$k=|\pi^\star|$ and define
\[
A_K(k)
=
\frac{k_0!}{k!}
\frac{
\binom{K^2-1}{k_0-1}
}{
\binom{K^2-1}{k-1}
}
\]
and
\[
S_K(\pi^\star)
=
\frac{
\displaystyle
\sum_{T\in\mathfrak T_K(\pi^\star)}p_K(T)
}{
\displaystyle
\sum_{T\in\mathfrak T_K(\pi_{0,K}^\star)}p_K(T)
}.
\]
From the partition prior in Section~\ref{sec-methodology} and the
definition \eqref{eq-app-Ltau}, the exact posterior-odds identity is
\begin{equation}
\frac{
\Pi_{\tau,n}(\pi^\star\mid X_{1:n},s_{1:n})
}{
\Pi_{\tau,n}(\pi_{0,K}^\star\mid X_{1:n},s_{1:n})
}
=
L_{\tau,n}(\pi^\star)
A_K(k)
S_K(\pi^\star).
\label{eq-app-posterior-odds}
\end{equation}

The compatible-tree bound of \citet{huang2025consistent} implies
that, under Assumption~\ref{assump-spatial}, there exists
$C_{\mathrm{ST}}<\infty$ such that
\begin{equation}
S_K(\pi^\star)
\leq
\exp\{C_{\mathrm{ST}}K\log K\}
\label{eq-app-tree-ratio}
\end{equation}
uniformly over $\pi^\star\in\mathfrak P_K$ for all sufficiently
large $n$. Their counting argument also gives constants
$C_{\mathrm{all}},C_q<\infty$ such that
\begin{equation}
|\mathfrak P_K|
\leq
\exp(C_{\mathrm{all}}K^2),
\label{eq-app-total-count}
\end{equation}
and, when $k_0\geq2$, if
\[
\Pi_{4,K}^\star(q)
=
\left\{
\pi^\star\in\Pi_{4,K}^\star:
q_K(\pi^\star)=q
\right\},
\]
then
\begin{equation}
|\Pi_{4,K}^\star(q)|
\leq
C\exp(C_q q\log K),
\qquad
1\leq q\leq K^2.
\label{eq-app-q-count}
\end{equation}
The bound \eqref{eq-app-tree-ratio} follows from the compatible-tree
multiplicity argument of \citet{huang2025consistent}; the uniform
tree prior in Assumption~\ref{assump-spatial} satisfies their
tree-prior ratio condition.

Since $k,k_0\leq k_{\max}$ and $k_{\max}$ is fixed, the elementary
binomial-factor ratio satisfies, for some $C_A<\infty$,
\begin{equation}
A_K(k)
\leq
\exp(C_A\log K),
\qquad
1\leq k\leq k_{\max}.
\label{eq-app-AK-bound}
\end{equation}

\appsubsection{Posterior concentration on the good partition class}
\label{app-good-class-concentration}

\begin{suppproposition}[Posterior concentration on
$\Pi_{5,K}^\star$]
\label{prop-app-good-class-concentration}
On $\mathcal G_n$, there exist constants $C,c>0$ such that, for all
sufficiently large $n$,
\[
\Pi_{\tau,n}
\left(
\Pi_{5,K}^\star
\,\middle|\,
X_{1:n},s_{1:n}
\right)
\geq
1-Ce^{-c\zeta_n}.
\]
\end{suppproposition}

\noindent\textit{Proof.}
Because the posterior denominator is at least the contribution of
$\pi_{0,K}^\star$, for any
$\mathcal B\subseteq\mathfrak P_K$,
\begin{equation}
\Pi_{\tau,n}
(\mathcal B\mid X_{1:n},s_{1:n})
\leq
\sum_{\pi^\star\in\mathcal B}
\frac{
\Pi_{\tau,n}(\pi^\star\mid X_{1:n},s_{1:n})
}{
\Pi_{\tau,n}(\pi_{0,K}^\star\mid X_{1:n},s_{1:n})
}.
\label{eq-app-posterior-set-bound}
\end{equation}
We apply this inequality to the four bad classes.

For $\Pi_{1,K}^\star$, Proposition~
\ref{prop-app-classwise-score-bounds},
\eqref{eq-app-tree-ratio},
\eqref{eq-app-total-count}, and
\eqref{eq-app-AK-bound} give
\[
\begin{aligned}
\Pi_{\tau,n}
(\Pi_{1,K}^\star\mid X_{1:n},s_{1:n})
&\leq
\exp\left\{
-c_1n
+
C_{\mathrm{all}}K^2
+
C_{\mathrm{ST}}K\log K
+
C_A\log K
\right\}
\\
&\leq
\exp(-c_1'n)
\end{aligned}
\]
for some $c_1'>0$ and all sufficiently large $n$, because
\[
K^2\asymp\frac{n}{L_n}=o(n),
\qquad
K\log K=o(n).
\]
The same argument gives
\begin{equation}
\Pi_{\tau,n}
(\Pi_{3,K}^\star\mid X_{1:n},s_{1:n})
\leq
\exp(-c_3'n)
\label{eq-app-class3-posterior}
\end{equation}
for some $c_3'>0$.

For the overfitted class,
\[
\begin{aligned}
\Pi_{\tau,n}
(\Pi_{2,K}^\star\mid X_{1:n},s_{1:n})
&\leq
\exp\left\{
-c_2n\log^{-\alpha_p}(n)
+
C_{\mathrm{all}}K^2
+
C_{\mathrm{ST}}K\log K
+
C_A\log K
\right\}.
\end{aligned}
\]
Now
\[
\frac{K^2}
{n\log^{-\alpha_p}(n)}
\asymp
\log^{\alpha_p-1-\alpha_b}(n)
\longrightarrow0,
\]
and $K\log K=o(K^2)$. Hence, for some $c_2'>0$,
\begin{equation}
\Pi_{\tau,n}
(\Pi_{2,K}^\star\mid X_{1:n},s_{1:n})
\leq
\exp\{-c_2'n\log^{-\alpha_p}(n)\}.
\label{eq-app-class2-posterior}
\end{equation}

It remains to treat $\Pi_{4,K}^\star$. If $k_0=1$, this class is
empty. Suppose $k_0\geq2$. Since every partition in
$\Pi_{4,K}^\star$ has $k_0$ regions,
\[
A_K(k_0)=1.
\]
Using \eqref{eq-app-posterior-set-bound},
\eqref{eq-app-tree-ratio},
\eqref{eq-app-q-count}, and
Proposition~\ref{prop-app-classwise-score-bounds},
\[
\begin{aligned}
\Pi_{\tau,n}
(\Pi_{4,K}^\star\mid X_{1:n},s_{1:n})
&\leq
C\exp(C_{\mathrm{ST}}K\log K)
\\
&\quad\times
\sum_{q=q_n}^{K^2}
\exp\left\{
C_q q\log K
-
c_4qL_n
\right\}.
\end{aligned}
\]
Since
\[
\frac{\log K}{L_n}
=
O\{\log^{-\alpha_b}(n)\}
\longrightarrow0,
\]
for all sufficiently large $n$,
\[
C_q\log K
\leq
\frac{c_4}{2}L_n.
\]
Therefore
\[
\Pi_{\tau,n}
(\Pi_{4,K}^\star\mid X_{1:n},s_{1:n})
\leq
C\exp(C_{\mathrm{ST}}K\log K)
\sum_{q=q_n}^{K^2}
\exp\left\{
-\frac{c_4}{2}qL_n
\right\}.
\]
The sum is bounded by a constant multiple of its first term, so
\[
\Pi_{\tau,n}
(\Pi_{4,K}^\star\mid X_{1:n},s_{1:n})
\leq
C
\exp\left\{
C_{\mathrm{ST}}K\log K
-
c_4' q_nL_n
\right\}.
\]
Moreover,
$\frac{K\log K}{q_nL_n}
=
O\left\{
\log^{-(\alpha_0+\alpha_b)}(n)
\right\}
\longrightarrow0.$
Since $\zeta_n=q_nL_n$, it follows that
\begin{equation}
\Pi_{\tau,n}
(\Pi_{4,K}^\star\mid X_{1:n},s_{1:n})
\leq
C\exp(-c_4''\zeta_n)
\label{eq-app-class4-posterior}
\end{equation}
for some $c_4''>0$.

Finally, for every fixed $\alpha_0>0$,
$\zeta_n
\asymp
n^{1/2}
\log^{\alpha_0+(1+\alpha_b)/2}(n),$
and therefore
$\zeta_n
=
o\{n\log^{-\alpha_p}(n)\}.$
Thus the bounds for
$\Pi_{1,K}^\star$, $\Pi_{2,K}^\star$, and
$\Pi_{3,K}^\star$ are also bounded by
$Ce^{-c\zeta_n}$ after adjusting constants. Since
\[
\mathfrak P_K\setminus\Pi_{5,K}^\star
=
\Pi_{1,K}^\star
\cup
\Pi_{2,K}^\star
\cup
\Pi_{3,K}^\star
\cup
\Pi_{4,K}^\star,
\]
the union bound gives
\[
\Pi_{\tau,n}
\left(
\mathfrak P_K\setminus\Pi_{5,K}^\star
\,\middle|\,
X_{1:n},s_{1:n}
\right)
\leq
Ce^{-c\zeta_n}.
\]
This proves the result.
\hfill$\square$

\appsubsection{Geometric consequences of the good partition class}
\label{app-good-class-geometry}

\begin{lemma}[Good-class partition discrepancy]
\label{lem-app-good-class-geometry}
There exists $C_{\mathrm{geo}}<\infty$ such that, for all
sufficiently large $n$,
\[
\sup_{\pi^\star\in\Pi_{5,K}^\star}
\epsilon(\pi^\star,\pi_0)
\leq
C_{\mathrm{geo}}\rho_n.
\]
Moreover, on $\mathcal A_n$,
\[
\sup_{\pi^\star\in\Pi_{5,K}^\star}
\epsilon_n\{
\mathcal P(\pi^\star),\mathcal P_0
\}
\leq
C_{\mathrm{geo}}\rho_n.
\]
\end{lemma}

\noindent\textit{Proof.}
Fix $\pi^\star\in\Pi_{5,K}^\star$ and write
$q=q_K(\pi^\star)<q_n$. Since the best-match map is a bijection,
relabel the candidate regions so that
$m_{\pi^\star}(j)=j$.

First compare $\pi^\star$ with the supported comparator
$\pi_{0,K}^\star$. By the definition of $q_K$,
\[
\epsilon_1(\pi^\star,\pi_{0,K}^\star)
=
\frac{q}{K^2}.
\]
For the reverse directed discrepancy,
\[
\max_j b_{j\ell}\geq b_{\ell\ell}
\]
for every comparator label $\ell$, and hence
\[
\epsilon_1(\pi_{0,K}^\star,\pi^\star)
\leq
1-\frac{1}{K^2}
\sum_{\ell=1}^{k_0}b_{\ell\ell}
=
\frac{q}{K^2}.
\]
Therefore
\begin{equation}
\epsilon(\pi^\star,\pi_{0,K}^\star)
\leq
\frac{2q}{K^2}.
\label{eq-app-candidate-comparator-domain}
\end{equation}

The comparator construction satisfies
$|W_{0,K}|\leq C_{\mathrm{app}}K$. As shown in the proof of
Lemma~\ref{lem-app-macroscopic-comparator}, the total area by which
the labelled comparator regions differ from their corresponding true
regions is $O(K^{-1})$. Hence
\begin{equation}
\epsilon(\pi_{0,K}^\star,\pi_0)
\leq
\frac{C}{K}.
\label{eq-app-comparator-true-domain}
\end{equation}
Using the triangle inequality for the label-invariant partition
metric $\epsilon$, established in \citet{huang2025consistent},
\[
\epsilon(\pi^\star,\pi_0)
\leq
C\left(
\frac{q}{K^2}+\frac{1}{K}
\right).
\]
Since $q<q_n$,
\[
\frac{q}{K^2}
\leq
\frac{q_n}{K^2}
\leq
\frac{\log^{\alpha_0}(n)}{K}
=
\rho_n,
\]
and, for all sufficiently large $n$,
$K^{-1}\leq\rho_n$. This proves the domain-partition bound.

For the sampled-location partitions, on $\mathcal A_n$ the
observations lying in the $q$ off-diagonal overlap blocks number at
most $C_BqL_n$. Consequently,
\[
\epsilon_n\{
\mathcal P(\pi^\star),
\mathcal P(\pi_{0,K}^\star)
\}
\leq
C\frac{qL_n}{n}.
\]
Outside the comparator-error blocks $W_{0,K}$, the comparator label
agrees with the true label. Therefore, by
\eqref{eq-app-comparator-error-count},
\[
\epsilon_n\{
\mathcal P(\pi_{0,K}^\star),
\mathcal P_0
\}
\leq
C\frac{KL_n}{n}.
\]
The triangle inequality for $\epsilon_n$ thus gives
\[
\epsilon_n\{
\mathcal P(\pi^\star),
\mathcal P_0
\}
\leq
C
\left(
\frac{qL_n}{n}
+
\frac{KL_n}{n}
\right).
\]
Since $K^2L_n\asymp n$,
\[
\frac{qL_n}{n}
\leq
C\frac{q}{K^2},
\qquad
\frac{KL_n}{n}
\leq
\frac{C}{K}.
\]
The preceding bounds therefore imply
\[
\epsilon_n\{
\mathcal P(\pi^\star),
\mathcal P_0
\}
\leq
C_{\mathrm{geo}}\rho_n.
\]
The constants are independent of
$\pi^\star\in\Pi_{5,K}^\star$, proving the result.
\hfill$\square$

\appsubsection{Completion of the proof}
\label{app-theorem2-completion}

\noindent\textit{Proof of Theorem~\ref{thm-partition-contraction}.}
By \eqref{eq-app-Gn-probability},
$P_0(\mathcal G_n)\longrightarrow1.$
On $\mathcal G_n$,
Proposition~\ref{prop-app-good-class-concentration} gives constants
$C,c>0$ such that
\[
\Pi_{\tau,n}
\left(
\Pi_{5,K}^\star
\,\middle|\,
X_{1:n},s_{1:n}
\right)
\geq
1-Ce^{-c\zeta_n}.
\]
This proves the first assertion of the theorem.

Every partition in $\Pi_{5,K}^\star$ has exactly $k_0$ regions.
Therefore
$\Pi_{5,K}^\star
\subseteq
\{\pi^\star:|\pi^\star|=k_0\},$
and hence
\[
\Pi_{\tau,n}
\left(
|\pi^\star|=k_0
\,\middle|\,
X_{1:n},s_{1:n}
\right)
\geq
1-Ce^{-c\zeta_n}.
\]

By Lemma~\ref{lem-app-good-class-geometry}, for some
$M<\infty$ and all sufficiently large $n$,
\[
\Pi_{5,K}^\star
\subseteq
\left\{
\pi^\star:
\epsilon(\pi^\star,\pi_0)
\leq
M\rho_n
\right\}.
\]
Thus
\[
\begin{aligned}
\Pi_{\tau,n}
\left\{
\epsilon(\pi^\star,\pi_0)
\leq
M\rho_n
\,\middle|\,
X_{1:n},s_{1:n}
\right\}
&\geq
\Pi_{\tau,n}
\left(
\Pi_{5,K}^\star
\,\middle|\,
X_{1:n},s_{1:n}
\right)
\\
&\geq
1-Ce^{-c\zeta_n}.
\end{aligned}
\]

Finally, since
$\mathcal G_n\subseteq\mathcal A_n$,
Lemma~\ref{lem-app-good-class-geometry} also gives
\[
\Pi_{5,K}^\star
\subseteq
\left\{
\pi^\star:
\epsilon_n\{
\mathcal P(\pi^\star),\mathcal P_0
\}
\leq
M\rho_n
\right\}
\]
on $\mathcal G_n$. Therefore
\[
\begin{aligned}
\Pi_{\tau,n}
\left\{
\epsilon_n\{
\mathcal P(\pi^\star),\mathcal P_0
\}
\leq
M\rho_n
\,\middle|\,
X_{1:n},s_{1:n}
\right\}
&\geq
\Pi_{\tau,n}
\left(
\Pi_{5,K}^\star
\,\middle|\,
X_{1:n},s_{1:n}
\right)
\\
&\geq
1-Ce^{-c\zeta_n}.
\end{aligned}
\]
All four inequalities hold on events $\mathcal G_n$ whose
$P_0$-probability tends to one. This completes the proof.
\hfill$\square$

\appsection{Proof of Theorem~\ref{thm-representative-field-contraction}}
\label{app-proof-theorem3}

\appsubsection{An enlarged high-probability event}
\label{app-theorem3-high-probability-event}

The event $\mathcal G_n$ constructed in
Section~\ref{app-uniform-response-event} controls block occupancies
and blockwise residual sums. For representative contraction we also
require residual control over the macroscopic comparator regions.

Recall
\[
\epsilon_i
=
X_i-\mu_{\ell_0(s_i),0}.
\]
For $\ell=1,\ldots,k_0$, define
\[
R_{\ell,0,K}^\star
=
\sum_{i\in C_{\ell,0,K}^\star}\epsilon_i.
\]
Fix a finite constant $C_H>0$ and define
\[
\mathcal H_n
=
\bigcap_{\ell=1}^{k_0}
\left\{
\|R_{\ell,0,K}^\star\|
\leq
C_Hn\rho_n
\right\},
\]
and
\begin{equation}
\mathcal G_n^{(2)}
=
\mathcal G_n\cap\mathcal H_n.
\label{eq-app-Gn2}
\end{equation}

\begin{lemma}[Comparator-region residual concentration]
\label{lem-app-comparator-residual}
Under the standing assumptions,
\[
P_0(\mathcal H_n)\longrightarrow1,
\qquad
P_0(\mathcal G_n^{(2)})\longrightarrow1.
\]
\end{lemma}

\noindent\textit{Proof.}
If $\sigma_X^2=0$, the argument in
Lemma~\ref{lem-app-block-residual} shows that
$\epsilon_i=0$ almost surely, conditionally on $s_{1:n}$, for every
$i$. Hence $\mathcal H_n$ holds with probability one.

Suppose now that $\sigma_X^2>0$. Conditional on
$s_{1:n}\in\mathcal A_n$, write
$m_\ell
=
|C_{\ell,0,K}^\star|.$
By Lemma~\ref{lem-app-macroscopic-comparator}, for all sufficiently
large $n$,
\[
c_{\mathrm{obs}}n
\leq
m_\ell
\leq
n,
\qquad
\ell=1,\ldots,k_0.
\]
Let $e_1,\ldots,e_r$ be an orthonormal basis of $\mathcal V$ and
put
\[
Z_{\ell,a}
=
\langle e_a,R_{\ell,0,K}^\star\rangle.
\]
Conditional independence and
Assumption~\ref{assump-bregman-response} imply
\[
E_0\{\exp(tZ_{\ell,a})\mid s_{1:n}\}
\leq
\exp\left(
\frac{\sigma_X^2m_\ell t^2}{2}
\right),
\qquad
|t|\leq b_X^{-1}.
\]

Set
$x_n
=
\frac{C_Hn\rho_n}{\sqrt r}.$
The corresponding Chernoff value satisfies
$\frac{x_n}{\sigma_X^2m_\ell}
\leq
\frac{C_H}{\sqrt r\,c_{\mathrm{obs}}\sigma_X^2}
\rho_n
\longrightarrow0.$
Hence the local exponential-moment bound is applicable for all
sufficiently large $n$. Chernoff's inequality gives
\[
P_0\{|Z_{\ell,a}|>x_n\mid s_{1:n}\}
\leq
2\exp\left\{
-\frac{x_n^2}{2\sigma_X^2m_\ell}
\right\}
\leq
2\exp(-c_Hn\rho_n^2)
\]
for some constant $c_H>0$.

If
$\|R_{\ell,0,K}^\star\|>C_Hn\rho_n$, then at least one coordinate
has magnitude greater than $x_n$. Therefore
\[
P_0
\left(
\|R_{\ell,0,K}^\star\|>C_Hn\rho_n
\mid s_{1:n}
\right)
\leq
2r\exp(-c_Hn\rho_n^2).
\]
Since $k_0$ is fixed and
$n\rho_n^2
\asymp
\log^{2\alpha_0+1+\alpha_b}(n)
\longrightarrow\infty,$
a union bound gives
\[
P_0(\mathcal H_n^c\mid s_{1:n})
\longrightarrow0
\]
uniformly over $s_{1:n}\in\mathcal A_n$. Since
$P_0(\mathcal A_n)\to1$, it follows that
$P_0(\mathcal H_n)\to1$.

Finally,
\[
\mathcal G_n^{(2)}
=
\mathcal G_n\cap\mathcal H_n,
\]
and \eqref{eq-app-Gn-probability} gives
$P_0(\mathcal G_n)\to1$. Hence
\[
P_0(\mathcal G_n^{(2)})\longrightarrow1.
\]
\hfill$\square$

\appsubsection{Good-class spatial consequences}
\label{app-theorem3-good-class-spatial}

\begin{lemma}[Good-class clusters have macroscopic sample size]
\label{lem-app-good-class-size}
On $\mathcal A_n$, there exists $c_\star>0$ such that, for all
sufficiently large $n$,
\[
\inf_{\pi^\star\in\Pi_{5,K}^\star}
\min_{1\leq j\leq k_0}
|C_j(\pi^\star)|
\geq
c_\star n.
\]
\end{lemma}

\noindent\textit{Proof.}
Fix $\pi^\star\in\Pi_{5,K}^\star$ and temporarily relabel its
regions so that
\[
m_{\pi^\star}(j)=j,
\qquad
j=1,\ldots,k_0.
\]
Then
$q_K(\pi^\star)
=
\sum_{j=1}^{k_0}\sum_{\ell\neq j}b_{j\ell}
<
q_n.$
For each $j$,
$|\mathcal V_{j,0,K}^\star|
=
b_{jj}
+
\sum_{h\neq j}b_{hj}.$
By Lemma~\ref{lem-app-macroscopic-comparator},
$|\mathcal V_{j,0,K}^\star|
\geq
v_0K^2,$
while
$\sum_{h\neq j}b_{hj}
\leq
q_K(\pi^\star)
<
q_n.$
Hence
$b_{jj}
\geq
v_0K^2-q_n.$
Since
$q_n
=
O\{K\log^{\alpha_0}(n)\}
=
o(K^2),$
for all sufficiently large $n$,
$b_{jj}
\geq
\frac{v_0}{2}K^2.$
Thus
\[
|\mathcal V_j|
\geq
b_{jj}
\geq
\frac{v_0}{2}K^2.
\]
On $\mathcal A_n$,
\[
|C_j(\pi^\star)|
\geq
c_B|\mathcal V_j|L_n
\geq
\frac{c_Bv_0}{2}K^2L_n.
\]
Since $K^2L_n\asymp n$, the claim follows uniformly over the good
class. Undoing the temporary relabelling gives the result.
\hfill$\square$

\begin{lemma}[Good-class block symmetric difference]
\label{lem-app-good-class-symmetric-difference}
For every $\pi^\star\in\Pi_{5,K}^\star$ and every
$j=1,\ldots,k_0$,
\[
\left|
\mathcal V_j
\triangle
\mathcal V_{m_{\pi^\star}(j),0,K}^\star
\right|
\leq
2q_K(\pi^\star)
<
2q_n.
\]
\end{lemma}

\noindent\textit{Proof.}
Again relabel temporarily so that
$m_{\pi^\star}(j)=j$. Then
\[
|\mathcal V_j\setminus
\mathcal V_{j,0,K}^\star|
=
\sum_{\ell\neq j}b_{j\ell},
\]
and
\[
|\mathcal V_{j,0,K}^\star\setminus
\mathcal V_j|
=
\sum_{h\neq j}b_{hj}.
\]
Both sums are bounded by $q_K(\pi^\star)$. Therefore
\[
|\mathcal V_j\triangle
\mathcal V_{j,0,K}^\star|
\leq
2q_K(\pi^\star)
<
2q_n.
\]
Undoing the relabelling proves the claim.
\hfill$\square$

\appsubsection{Uniform empirical-centroid contraction on the good class}
\label{app-theorem3-centroid}

\begin{lemma}[Contamination bias on the good class]
\label{lem-app-good-class-bias}
On $\mathcal A_n$, there exists
$C_{\mathrm{bias}}<\infty$ such that, for all sufficiently large
$n$,
\[
\sup_{\pi^\star\in\Pi_{5,K}^\star}
\max_{1\leq j\leq k_0}
\left\|
\frac{1}{|C_j(\pi^\star)|}
\sum_{i\in C_j(\pi^\star)}
\left\{
\mu_{\ell_0(s_i),0}
-
\mu_{m_{\pi^\star}(j),0}
\right\}
\right\|
\leq
C_{\mathrm{bias}}\rho_n.
\]
\end{lemma}

\noindent\textit{Proof.}
Fix $\pi^\star\in\Pi_{5,K}^\star$ and temporarily relabel so that
$m_{\pi^\star}(j)=j$. Define
\[
\mathcal B_j^{\mathrm{bad}}
=
\{B\in\mathcal V_j:B\not\subseteq D_{j,0}\}.
\]
If
\[
B\in
\mathcal V_j\cap
\mathcal V_{j,0,K}^\star
\]
and $B\notin W_{0,K}$, then
$B\subseteq D_{j,0}$ by the definition of $W_{0,K}$.
Consequently,
\[
|\mathcal B_j^{\mathrm{bad}}|
\leq
|\mathcal V_j\setminus
\mathcal V_{j,0,K}^\star|
+
|W_{0,K}|
\leq
q_n+C_{\mathrm{app}}K.
\]
Since $\alpha_0>0$,
$K=O(q_n)$, and therefore
$|\mathcal B_j^{\mathrm{bad}}|
\leq
Cq_n.$

On $\mathcal A_n$, at most $Cq_nL_n$ observations lie in these
blocks. Let
$\Delta_\mu
=
\max_{\ell,h}
\|\mu_{\ell,0}-\mu_{h,0}\|.$
For observations outside the bad blocks,
$\mu_{\ell_0(s_i),0}-\mu_{j,0}=0,$
whereas for observations inside them the norm of this difference is
at most $\Delta_\mu$. Hence
\[
\left\|
\sum_{i\in C_j(\pi^\star)}
\{\mu_{\ell_0(s_i),0}-\mu_{j,0}\}
\right\|
\leq
C\Delta_\mu q_nL_n.
\]
By Lemma~\ref{lem-app-good-class-size},
$|C_j(\pi^\star)|
\geq
c_\star n.$
Therefore
\[
\left\|
\frac{1}{|C_j(\pi^\star)|}
\sum_{i\in C_j(\pi^\star)}
\{\mu_{\ell_0(s_i),0}-\mu_{j,0}\}
\right\|
\leq
C\frac{q_nL_n}{n}.
\]
Since
\[
\frac{q_nL_n}{n}
\asymp
\frac{K\log^{\alpha_0}(n)L_n}
     {K^2L_n}
=
\frac{\log^{\alpha_0}(n)}{K}
=
\rho_n,
\]
the result follows uniformly in $j$ and $\pi^\star$. Undoing the
temporary relabelling proves the stated bound.
\hfill$\square$

\begin{lemma}[Residual average on the good class]
\label{lem-app-good-class-residual}
On $\mathcal G_n^{(2)}$, there exists
$C_{\mathrm{res}}<\infty$ such that, for all sufficiently large
$n$,
\[
\sup_{\pi^\star\in\Pi_{5,K}^\star}
\max_{1\leq j\leq k_0}
\left\|
\frac{1}{|C_j(\pi^\star)|}
\sum_{i\in C_j(\pi^\star)}
\epsilon_i
\right\|
\leq
C_{\mathrm{res}}\rho_n.
\]
\end{lemma}

\noindent\textit{Proof.}
Fix $\pi^\star\in\Pi_{5,K}^\star$ and temporarily relabel so that
$m_{\pi^\star}(j)=j$. Write
\[
C_j=C_j(\pi^\star),
\qquad
C_j^\star=C_{j,0,K}^\star,
\]
and, for $A\subseteq\{1,\ldots,n\}$, set
$R(A)=\sum_{i\in A}\epsilon_i.$
Then
$R(C_j)
=
R(C_j^\star)
+
R(C_j\setminus C_j^\star)
-
R(C_j^\star\setminus C_j),$
so
\begin{equation}
\|R(C_j)\|
\leq
\|R(C_j^\star)\|
+
\|R(C_j\setminus C_j^\star)\|
+
\|R(C_j^\star\setminus C_j)\|.
\label{eq-app-good-residual-decomp}
\end{equation}

On $\mathcal H_n$,
$\|R(C_j^\star)\|
=
\|R_{j,0,K}^\star\|
\leq
C_Hn\rho_n.$
By Lemma~\ref{lem-app-good-class-symmetric-difference}, each of the
two difference sets in
\eqref{eq-app-good-residual-decomp} is a union of blocks drawn from
a collection containing at most $2q_n$ blocks.

On $\mathcal F_n\cap\mathcal A_n$,
\[
\|R_{K,m}\|
\leq
\{M_FN_{K,m}\log n\}^{1/2}
\leq
C\sqrt{L_n\log n}
\]
for every block. Therefore the residual sum over any collection of
at most $2q_n$ blocks has norm at most
\[
Cq_n\sqrt{L_n\log n}.
\]
Consequently, on $\mathcal G_n^{(2)}$,
\[
\|R(C_j)\|
\leq
C_Hn\rho_n
+
Cq_n\sqrt{L_n\log n}.
\]
Dividing by
$|C_j|\geq c_\star n$ gives
\[
\frac{\|R(C_j)\|}{|C_j|}
\leq
C\rho_n
+
C\frac{q_n\sqrt{L_n\log n}}{n}.
\]
Now
\[
\frac{q_n\sqrt{L_n\log n}}{n}
\asymp
\frac{\log^{\alpha_0}(n)}{K}
\sqrt{\frac{\log n}{L_n}}
=
\rho_n\log^{-\alpha_b/2}(n)
\leq
\rho_n
\]
for all sufficiently large $n$. The result follows uniformly over
the good class.
\hfill$\square$

\begin{suppproposition}[Uniform empirical-centroid contraction]
\label{prop-app-good-class-centroid}
On $\mathcal G_n^{(2)}$, for all sufficiently large $n$,
\[
\sup_{\pi^\star\in\Pi_{5,K}^\star}
\max_{1\leq j\leq k_0}
\left\|
\bar X_{C_j(\pi^\star)}
-
\mu_{m_{\pi^\star}(j),0}
\right\|
\leq
C_{\mathrm{cent}}\rho_n
\]
for some $C_{\mathrm{cent}}<\infty$. Moreover,
\[
\bar X_{C_j(\pi^\star)}\in\mathcal K,
\qquad
\pi^\star\in\Pi_{5,K}^\star,
\quad
j=1,\ldots,k_0.
\]
\end{suppproposition}

\noindent\textit{Proof.}
Because
$\mathcal G_n^{(2)}\subseteq\mathcal G_n$,
Lemma~\ref{lem-app-union-centroid} implies that every candidate
cluster centroid belongs to $\mathcal K$ for all sufficiently large
$n$.

For the contraction bound, write
\[
\begin{aligned}
\bar X_{C_j(\pi^\star)}
-
\mu_{m_{\pi^\star}(j),0}
&=
\frac{1}{|C_j(\pi^\star)|}
\sum_{i\in C_j(\pi^\star)}
\epsilon_i
+
\frac{1}{|C_j(\pi^\star)|}
\sum_{i\in C_j(\pi^\star)}
\left\{
\mu_{\ell_0(s_i),0}
-
\mu_{m_{\pi^\star}(j),0}
\right\}.
\end{aligned}
\]
Lemma~\ref{lem-app-good-class-residual} bounds the first term by
$C_{\mathrm{res}}\rho_n$ uniformly, while
Lemma~\ref{lem-app-good-class-bias} bounds the second by
$C_{\mathrm{bias}}\rho_n$. Thus the assertion holds with
\[
C_{\mathrm{cent}}
=
C_{\mathrm{res}}+C_{\mathrm{bias}}.
\]
\hfill$\square$

\appsubsection{Conditional posterior concentration of the representatives}
\label{app-theorem3-conditional-representative}

\begin{lemma}[Conditional representative-posterior tail]
\label{lem-app-representative-tail}
There exist constants
$C_1,C_2,c_1,c_2\in(0,\infty)$ such that the following holds.
Let $C$ be any nonempty cluster, write
\[
m=|C|,
\qquad
u=\bar X_C,
\]
and suppose $u\in\mathcal K$. Then, for every $0<t\leq\rho$,
\[
\Pi_\tau
\left(
\|\mu-u\|>t
\,\middle|\,
C,X_{1:n}
\right)
\leq
C_1e^{-c_1mt^2}
+
C_2e^{-c_2m}.
\]
Consequently, if $m\geq c_0n$ for some fixed $c_0>0$ and
$t_n\in(0,\rho]$ satisfies
\[
nt_n^2\longrightarrow\infty,
\]
then
\[
\Pi_\tau
\left(
\|\mu-\bar X_C\|>t_n
\,\middle|\,
C,X_{1:n}
\right)
\longrightarrow0.
\]
\end{lemma}

\noindent\textit{Proof.}
By Theorem~\ref{thm-bregman-conjugacy-collapse},
\[
\Pi_\tau(d\mu\mid C,X_{1:n})
=
\frac{
e^{-\tau mD_\Phi(u,\mu)}
\Pi_{\chi_0,\kappa_0}(d\mu)
}{
J_m(u)
}.
\]
For $0<t\leq\rho$, split the numerator of the tail probability into
\[
N_1
=
\int_{\{t<\|\mu-u\|<\rho\}}
e^{-\tau mD_\Phi(u,\mu)}
\Pi_{\chi_0,\kappa_0}(d\mu)
\]
and
\[
N_2
=
\int_{\{\|\mu-u\|\geq\rho\}}
e^{-\tau mD_\Phi(u,\mu)}
\Pi_{\chi_0,\kappa_0}(d\mu).
\]

By Proposition~\ref{prop-uniform-integrated-center},
\[
J_m(u)\geq c_Jm^{-r/2}.
\]
On the local region
$\|\mu-u\|<\rho$, Lemma~\ref{lem-app-local-bregman} gives
\[
D_\Phi(u,\mu)
\geq
\frac{h_\Phi}{2}\|\mu-u\|^2.
\]
Also, by Lemma~\ref{lem-app-prior-density}, the prior density is
bounded above on $\mathcal K_\rho$. Hence
\[
N_1
\leq
C
\int_{\|\mu-u\|>t}
\exp\left\{
-\frac{\tau h_\Phi}{2}
m\|\mu-u\|^2
\right\}
\lambda_{\mathcal V}(d\mu).
\]
With $v=\sqrt m(\mu-u)$,
\[
N_1
\leq
Cm^{-r/2}
\int_{\|v\|>\sqrt m\,t}
e^{-a_\Phi\|v\|^2}
\lambda_{\mathcal V}(dv),
\qquad
a_\Phi=\frac{\tau h_\Phi}{2}.
\]
Since $r$ is fixed, there exist constants $C_G,c_G>0$ such that
\[
\int_{\|v\|>A}
e^{-a_\Phi\|v\|^2}
\lambda_{\mathcal V}(dv)
\leq
C_Ge^{-c_GA^2},
\qquad A\geq0.
\]
Therefore
$\frac{N_1}{J_m(u)}
\leq
C_1e^{-c_1mt^2}.$

For the nonlocal part,
Assumption~\ref{assump-bregman-response} gives
\[
D_\Phi(u,\mu)
\geq
\delta_{\Phi,\rho}
\]
whenever $\|\mu-u\|\geq\rho$. Thus
\[
N_2
\leq
e^{-\tau m\delta_{\Phi,\rho}}.
\]
Using again
$J_m(u)\geq c_Jm^{-r/2}$,
\[
\frac{N_2}{J_m(u)}
\leq
c_J^{-1}m^{r/2}
e^{-\tau m\delta_{\Phi,\rho}}
\leq
C_2e^{-c_2m}
\]
for suitable constants $C_2,c_2>0$. Combining the local and
nonlocal bounds proves the first assertion.

If $m\geq c_0n$ and $nt_n^2\to\infty$, then
\[
mt_n^2\geq c_0nt_n^2\longrightarrow\infty,
\qquad
m\longrightarrow\infty.
\]
Substitution into the preceding tail bound proves the second
assertion.
\hfill$\square$

\begin{suppproposition}[Uniform representative contraction on the good class]
\label{prop-app-good-class-representative}
Let $M_n\to\infty$ be any deterministic sequence. On
$\mathcal G_n^{(2)}$,
\[
\sup_{\pi^\star\in\Pi_{5,K}^\star}
\Pi_\tau
\left(
\max_{1\leq j\leq k_0}
\|\mu_j-\mu_{m_{\pi^\star}(j),0}\|
>
M_n\rho_n
\,\middle|\,
\pi^\star,X_{1:n},s_{1:n}
\right)
\longrightarrow0.
\]
\end{suppproposition}

\noindent\textit{Proof.}
Fix $\pi^\star\in\Pi_{5,K}^\star$. By
Proposition~\ref{prop-app-good-class-centroid},
$\max_{1\leq j\leq k_0}
\left\|
\bar X_{C_j(\pi^\star)}
-
\mu_{m_{\pi^\star}(j),0}
\right\|
\leq
C_{\mathrm{cent}}\rho_n.$
By Lemma~\ref{lem-app-good-class-size},
$|C_j(\pi^\star)|
\geq
c_\star n,
\qquad
j=1,\ldots,k_0,$
and all the empirical centroids belong to $\mathcal K$.

For all sufficiently large $n$,
$M_n>C_{\mathrm{cent}}+1$. Define
\[
t_n
=
\min\left\{
(M_n-C_{\mathrm{cent}})\rho_n,\,
\frac{\rho}{2}
\right\}.
\]
Then $0<t_n\leq\rho$, and
$nt_n^2\longrightarrow\infty.
$ Indeed, if the first term in the minimum is selected along a
subsequence, then
$nt_n^2
=
(M_n-C_{\mathrm{cent}})^2n\rho_n^2
\longrightarrow\infty,$
while if $t_n=\rho/2$, the conclusion is immediate.

If
$\|\mu_j-\mu_{m_{\pi^\star}(j),0}\|
>
M_n\rho_n,$
then the triangle inequality implies
\[
\|\mu_j-\bar X_{C_j(\pi^\star)}\|
>
(M_n-C_{\mathrm{cent}})\rho_n
\geq
t_n.
\]
Therefore, by a union bound,
\[
\begin{aligned}
&\Pi_\tau
\left(
\max_{1\leq j\leq k_0}
\|\mu_j-\mu_{m_{\pi^\star}(j),0}\|
>
M_n\rho_n
\,\middle|\,
\pi^\star,X_{1:n},s_{1:n}
\right)
\\
&\qquad\leq
\sum_{j=1}^{k_0}
\Pi_\tau
\left(
\|\mu_j-\bar X_{C_j(\pi^\star)}\|>t_n
\,\middle|\,
C_j(\pi^\star),X_{1:n}
\right).
\end{aligned}
\]
Lemma~\ref{lem-app-representative-tail} shows that every summand
converges to zero uniformly over the good class. Since $k_0$ is
fixed, the stated supremum converges to zero.
\hfill$\square$

\appsubsection{Representative-contraction conclusion}
\label{app-theorem3-representative-conclusion}

Let $M_n\to\infty$ be deterministic. By the definition of the joint
posterior,
\[
\begin{aligned}
&\Pi_{\tau,n}
\left(
\pi^\star\in\Pi_{5,K}^\star,\,
\max_{1\leq j\leq k_0}
\|\mu_j-\mu_{m_{\pi^\star}(j),0}\|
>
M_n\rho_n
\,\middle|\,
X_{1:n},s_{1:n}
\right)
\\
&\quad=
\int_{\Pi_{5,K}^\star}
\Pi_\tau
\left(
\max_{1\leq j\leq k_0}
\|\mu_j-\mu_{m_{\pi^\star}(j),0}\|
>
M_n\rho_n
\,\middle|\,
\pi^\star,X_{1:n},s_{1:n}
\right)
\Pi_{\tau,n}
(d\pi^\star\mid X_{1:n},s_{1:n}).
\end{aligned}
\]
Hence
\[
\begin{aligned}
&\Pi_{\tau,n}
\left(
\pi^\star\in\Pi_{5,K}^\star,\,
\max_{1\leq j\leq k_0}
\|\mu_j-\mu_{m_{\pi^\star}(j),0}\|
>
M_n\rho_n
\,\middle|\,
X_{1:n},s_{1:n}
\right)
\\
&\qquad\leq
\sup_{\pi^\star\in\Pi_{5,K}^\star}
\Pi_\tau
\left(
\max_{1\leq j\leq k_0}
\|\mu_j-\mu_{m_{\pi^\star}(j),0}\|
>
M_n\rho_n
\,\middle|\,
\pi^\star,X_{1:n},s_{1:n}
\right).
\end{aligned}
\]
By Proposition~\ref{prop-app-good-class-representative}, the
right-hand side converges to zero on $\mathcal G_n^{(2)}$, and
Lemma~\ref{lem-app-comparator-residual} gives
\[
P_0(\mathcal G_n^{(2)})\longrightarrow1.
\]
Therefore, with $P_0$-probability tending to one,
\[
\Pi_{\tau,n}
\left(
\pi^\star\in\Pi_{5,K}^\star,\,
\max_{1\leq j\leq k_0}
\|\mu_j-\mu_{m_{\pi^\star}(j),0}\|
>
M_n\rho_n
\,\middle|\,
X_{1:n},s_{1:n}
\right)
\longrightarrow0.
\]
This proves \eqref{eq-representative-contraction}.

\appsubsection{Good-class spatial mismatch area}
\label{app-theorem3-mismatch-area}

For $\pi^\star\in\Pi_{5,K}^\star$, define the labelled mismatch set
\begin{equation}
\mathcal M(\pi^\star)
=
\bigcup_{j=1}^{k_0}
\left(
D_j^\star
\setminus
D_{m_{\pi^\star}(j),0}
\right).
\label{eq-app-labelled-mismatch-set}
\end{equation}

\begin{lemma}[Good-class mismatch area]
\label{lem-app-good-class-mismatch-area}
There exists $C_{\mathrm{mis}}<\infty$ such that, for all sufficiently
large $n$,
\[
\sup_{\pi^\star\in\Pi_{5,K}^\star}
|\mathcal M(\pi^\star)|
\leq
C_{\mathrm{mis}}\rho_n.
\]
\end{lemma}

\noindent\textit{Proof.}
Fix $\pi^\star\in\Pi_{5,K}^\star$ and temporarily relabel the
candidate regions so that
\[
m_{\pi^\star}(j)=j,
\qquad
j=1,\ldots,k_0.
\]
Then
$\mathcal M(\pi^\star)
=
\bigcup_{j=1}^{k_0}
(D_j^\star\setminus D_{j,0}).$

Consider a mesh block $B\in\mathcal V_j$. If $B$ contributes to
$D_j^\star\setminus D_{j,0}$, then either $B$ belongs to a
comparator region with label $\ell\neq j$, or
$B\in
\mathcal V_j\cap\mathcal V_{j,0,K}^\star$
and $B$ is one of the comparator-error blocks in $W_{0,K}$.
Indeed, if
\[
B\in
\mathcal V_j\cap\mathcal V_{j,0,K}^\star
\quad\text{and}\quad
B\notin W_{0,K},
\]
then the definition of $W_{0,K}$ gives
$B\subseteq D_{j,0}$.

Consequently, up to block boundaries of Lebesgue measure zero,
$\mathcal M(\pi^\star)$ is covered by at most
\[
\sum_{j=1}^{k_0}\sum_{\ell\neq j}b_{j\ell}
+
|W_{0,K}|
=
q_K(\pi^\star)+|W_{0,K}|
\]
mesh blocks. Since $\pi^\star\in\Pi_{5,K}^\star$,
$q_K(\pi^\star)<q_n,$
while \eqref{eq-app-comparator-error} gives
\[
|W_{0,K}|
\leq
C_{\mathrm{app}}K.
\]
Each mesh block has area $K^{-2}$. Hence
\[
|\mathcal M(\pi^\star)|
\leq
\frac{q_n+C_{\mathrm{app}}K}{K^2}.
\]
Because
\[
q_n
=
\lfloor K\log^{\alpha_0}(n)\rfloor
\]
and $\alpha_0>0$, we have $K=O(q_n)$ and therefore
\[
|\mathcal M(\pi^\star)|
\leq
C\frac{\log^{\alpha_0}(n)}{K}
=
C\rho_n.
\]
The constant does not depend on $\pi^\star$. Undoing the temporary
relabelling proves the result.
\hfill$\square$

\appsubsection{A high-posterior-probability representative event}
\label{app-theorem3-representative-event}

Define the deterministic sequence
\begin{equation}
b_n
=
\rho_n^{1/2}.
\label{eq-app-bn}
\end{equation}
Since $\rho_n\to0$,
\[
b_n\longrightarrow0,
\qquad
\frac{b_n}{\rho_n}
=
\rho_n^{-1/2}
\longrightarrow\infty.
\]
For a joint posterior draw with
$\pi^\star\in\Pi_{5,K}^\star$, define
\begin{equation}
\mathcal R_n
=
\left\{
\max_{1\leq j\leq k_0}
\|\mu_j-\mu_{m_{\pi^\star}(j),0}\|
\leq
b_n
\right\}.
\label{eq-app-Rn}
\end{equation}

Applying
Proposition~\ref{prop-app-good-class-representative} with
\[
M_n
=
\frac{b_n}{\rho_n}
=
\rho_n^{-1/2}
\longrightarrow\infty
\]
gives, on $\mathcal G_n^{(2)}$,
\begin{equation}
\sup_{\pi^\star\in\Pi_{5,K}^\star}
\Pi_\tau
\left(
\mathcal R_n^c
\,\middle|\,
\pi^\star,X_{1:n},s_{1:n}
\right)
\longrightarrow0.
\label{eq-app-Rn-conditional}
\end{equation}
Hence, by integrating with respect to the marginal partition
posterior,
\begin{equation}
\Pi_{\tau,n}
\left(
\pi^\star\in\Pi_{5,K}^\star,\,
\mathcal R_n^c
\,\middle|\,
X_{1:n},s_{1:n}
\right)
\longrightarrow0.
\label{eq-app-Rn-joint}
\end{equation}

For all sufficiently large $n$, $b_n\leq\rho$. On
$\mathcal R_n$, each true matched center belongs to
$\mathcal H_0\subseteq\mathcal K$, and therefore
$\operatorname{dist}(\mu_j,\mathcal K)
\leq
\|\mu_j-\mu_{m_{\pi^\star}(j),0}\|
\leq
b_n
\leq
\rho.$
Consequently,
\begin{equation}
\mu_j\in\mathcal K_\rho,
\qquad
j=1,\ldots,k_0,
\label{eq-app-representatives-Krho}
\end{equation}
on $\mathcal R_n$ for all sufficiently large $n$.

Let
$p_{\max}
=
\sup_{s\in\mathcal D}p_D(s)<\infty.$
Since $\mathcal K_\rho$ is compact, define
\begin{equation}
B_\mu
=
\sup_{u,v\in\mathcal K_\rho}
\|u-v\|^2
<
\infty.
\label{eq-app-Bmu}
\end{equation}
Moreover, because
$\Phi\in C^2(\mathcal U)$ and
$\mathcal K_\rho\subset\mathcal U$ is compact,
$D_\Phi$ is continuous on
$\mathcal K_\rho\times\mathcal K_\rho$. Hence
\begin{equation}
B_\Phi
=
\sup_{u,v\in\mathcal K_\rho}
D_\Phi(u,v)
<
\infty.
\label{eq-app-Bphi}
\end{equation}

\appsubsection{Integrated squared center-field contraction}
\label{app-theorem3-l2-field}

Let $M_n'\to\infty$ be arbitrary and deterministic. We prove
\eqref{eq-field-l2-contraction}.

Work on $\mathcal G_n^{(2)}$. By
Proposition~\ref{prop-app-good-class-concentration},
\begin{equation}
\Pi_{\tau,n}
\left(
\Pi_{5,K}^{\star\,c}
\,\middle|\,
X_{1:n},s_{1:n}
\right)
\leq
Ce^{-c\zeta_n}
\longrightarrow0.
\label{eq-app-good-class-complement-theorem3}
\end{equation}
Together with \eqref{eq-app-Rn-joint}, it is therefore enough to
bound the field loss uniformly on
$\{\pi^\star\in\Pi_{5,K}^\star\}\cap\mathcal R_n.$

Fix $\pi^\star\in\Pi_{5,K}^\star$ and work on $\mathcal R_n$.
First consider
$s\in\mathcal M(\pi^\star)^c.$
If $s\in D_j^\star$, then, apart from null boundary sets,
$s\in D_{m_{\pi^\star}(j),0}.$
Hence
\[
\mu_{\pi^\star}(s)=\mu_j,
\qquad
\mu_0(s)=\mu_{m_{\pi^\star}(j),0}.
\]
By the definition of $\mathcal R_n$,
$\|\mu_{\pi^\star}(s)-\mu_0(s)\|^2
\leq
b_n^2
=
\rho_n.$
Therefore, since $|\mathcal D|=1$,
\begin{equation}
\int_{\mathcal M(\pi^\star)^c}
\|\mu_{\pi^\star}(s)-\mu_0(s)\|^2
p_D(s)\,ds
\leq
p_{\max}\rho_n.
\label{eq-app-l2-matched}
\end{equation}

Now consider
$s\in\mathcal M(\pi^\star)$. On $\mathcal R_n$,
\eqref{eq-app-representatives-Krho} gives
$\mu_{\pi^\star}(s)\in\mathcal K_\rho.
$ Also,
\[
\mu_0(s)\in
\{\mu_{1,0},\ldots,\mu_{k_0,0}\}
\subseteq
\mathcal H_0
\subseteq
\mathcal K
\subseteq
\mathcal K_\rho.
\]
Thus, by \eqref{eq-app-Bmu},
$\|\mu_{\pi^\star}(s)-\mu_0(s)\|^2
\leq
B_\mu.$
Using Lemma~\ref{lem-app-good-class-mismatch-area},
\begin{equation}
\begin{aligned}
\int_{\mathcal M(\pi^\star)}
\|\mu_{\pi^\star}(s)-\mu_0(s)\|^2
p_D(s)\,ds
&\leq
p_{\max}B_\mu
|\mathcal M(\pi^\star)|
\\
&\leq
p_{\max}B_\mu C_{\mathrm{mis}}\rho_n.
\end{aligned}
\label{eq-app-l2-mismatch}
\end{equation}

Combining
\eqref{eq-app-l2-matched} and
\eqref{eq-app-l2-mismatch}, there exists
$C_{L^2}<\infty$ such that, uniformly over
$\pi^\star\in\Pi_{5,K}^\star$, on $\mathcal R_n$,
\begin{equation}
\int_{\mathcal D}
\|\mu_{\pi^\star}(s)-\mu_0(s)\|^2
p_D(s)\,ds
\leq
C_{L^2}\rho_n.
\label{eq-app-l2-uniform}
\end{equation}

Since $M_n'\to\infty$, for all sufficiently large $n$,
\[
M_n'>C_{L^2}.
\]
Thus the event
\[
\left\{
\int_{\mathcal D}
\|\mu_{\pi^\star}(s)-\mu_0(s)\|^2
p_D(s)\,ds
>
M_n'\rho_n
\right\}
\]
cannot occur on
$\{\pi^\star\in\Pi_{5,K}^\star\}\cap\mathcal R_n$.
Consequently,
\[
\begin{aligned}
&\Pi_{\tau,n}
\left(
\int_{\mathcal D}
\|\mu_{\pi^\star}(s)-\mu_0(s)\|^2
p_D(s)\,ds
>
M_n'\rho_n
\,\middle|\,
X_{1:n},s_{1:n}
\right)
\\
&\qquad\leq
\Pi_{\tau,n}
\left(
\Pi_{5,K}^{\star\,c}
\,\middle|\,
X_{1:n},s_{1:n}
\right)
\\
&\qquad\quad+
\Pi_{\tau,n}
\left(
\pi^\star\in\Pi_{5,K}^\star,\,
\mathcal R_n^c
\,\middle|\,
X_{1:n},s_{1:n}
\right).
\end{aligned}
\]
The first term tends to zero by
\eqref{eq-app-good-class-complement-theorem3}, and the second tends
to zero by \eqref{eq-app-Rn-joint}. Therefore, with
$P_0$-probability tending to one,
\[
\Pi_{\tau,n}
\left(
\int_{\mathcal D}
\|\mu_{\pi^\star}(s)-\mu_0(s)\|^2
p_D(s)\,ds
>
M_n'\rho_n
\,\middle|\,
X_{1:n},s_{1:n}
\right)
\longrightarrow0.
\]
This proves \eqref{eq-field-l2-contraction}.

\appsubsection{Integrated Bregman center-field contraction}
\label{app-theorem3-bregman-field}

Let $M_n''\to\infty$ be arbitrary and deterministic. We now prove
\eqref{eq-field-bregman-contraction}.

Again work on
\[
\{\pi^\star\in\Pi_{5,K}^\star\}\cap\mathcal R_n.
\]
For
$s\in\mathcal M(\pi^\star)^c$, if $s\in D_j^\star$, then
\[
\mu_0(s)
=
\mu_{m_{\pi^\star}(j),0}
\in\mathcal K,
\qquad
\mu_{\pi^\star}(s)=\mu_j.
\]
Moreover,
$\|\mu_j-\mu_{m_{\pi^\star}(j),0}\|
\leq
b_n
\leq
\rho$
for all sufficiently large $n$. Hence
Lemma~\ref{lem-app-local-bregman} gives
\[
\begin{aligned}
D_\Phi\{
\mu_0(s),\mu_{\pi^\star}(s)
\}
&=
D_\Phi\{
\mu_{m_{\pi^\star}(j),0},\mu_j
\}
\\
&\leq
\frac{\overline h_\Phi}{2}
\|\mu_j-\mu_{m_{\pi^\star}(j),0}\|^2
\\
&\leq
\frac{\overline h_\Phi}{2}b_n^2
=
\frac{\overline h_\Phi}{2}\rho_n.
\end{aligned}
\]
Therefore
\begin{equation}
\int_{\mathcal M(\pi^\star)^c}
D_\Phi\{
\mu_0(s),\mu_{\pi^\star}(s)
\}
p_D(s)\,ds
\leq
p_{\max}
\frac{\overline h_\Phi}{2}
\rho_n.
\label{eq-app-bregman-matched}
\end{equation}

On the mismatch set
$\mathcal M(\pi^\star)$, both
$\mu_0(s)$ and $\mu_{\pi^\star}(s)$ belong to
$\mathcal K_\rho$ on $\mathcal R_n$. Hence
\[
D_\Phi\{
\mu_0(s),\mu_{\pi^\star}(s)
\}
\leq
B_\Phi.
\]
Lemma~\ref{lem-app-good-class-mismatch-area} therefore gives
\begin{equation}
\begin{aligned}
\int_{\mathcal M(\pi^\star)}
D_\Phi\{
\mu_0(s),\mu_{\pi^\star}(s)
\}
p_D(s)\,ds
&\leq
p_{\max}B_\Phi
|\mathcal M(\pi^\star)|
\\
&\leq
p_{\max}B_\Phi
C_{\mathrm{mis}}\rho_n.
\end{aligned}
\label{eq-app-bregman-mismatch}
\end{equation}

Combining
\eqref{eq-app-bregman-matched} and
\eqref{eq-app-bregman-mismatch}, there exists
$C_\Phi<\infty$ such that, uniformly over
$\pi^\star\in\Pi_{5,K}^\star$, on $\mathcal R_n$,
\begin{equation}
\int_{\mathcal D}
D_\Phi\{
\mu_0(s),\mu_{\pi^\star}(s)
\}
p_D(s)\,ds
\leq
C_\Phi\rho_n.
\label{eq-app-bregman-uniform}
\end{equation}

Since $M_n''\to\infty$, for all sufficiently large $n$,
$M_n''>C_\Phi.$
Therefore
\[
\begin{aligned}
&\Pi_{\tau,n}
\left(
\int_{\mathcal D}
D_\Phi\{
\mu_0(s),\mu_{\pi^\star}(s)
\}
p_D(s)\,ds
>
M_n''\rho_n
\,\middle|\,
X_{1:n},s_{1:n}
\right)
\\
&\qquad\leq
\Pi_{\tau,n}
\left(
\Pi_{5,K}^{\star\,c}
\,\middle|\,
X_{1:n},s_{1:n}
\right)
\\
&\qquad\quad+
\Pi_{\tau,n}
\left(
\pi^\star\in\Pi_{5,K}^\star,\,
\mathcal R_n^c
\,\middle|\,
X_{1:n},s_{1:n}
\right).
\end{aligned}
\]
Both terms converge to zero on $\mathcal G_n^{(2)}$ by
\eqref{eq-app-good-class-complement-theorem3} and
\eqref{eq-app-Rn-joint}. Since
\[
P_0(\mathcal G_n^{(2)})\longrightarrow1,
\]
we conclude that, with $P_0$-probability tending to one,
\[
\Pi_{\tau,n}
\left(
\int_{\mathcal D}
D_\Phi\{
\mu_0(s),\mu_{\pi^\star}(s)
\}
p_D(s)\,ds
>
M_n''\rho_n
\,\middle|\,
X_{1:n},s_{1:n}
\right)
\longrightarrow0.
\]
This proves \eqref{eq-field-bregman-contraction} and completes the
proof of Theorem~\ref{thm-representative-field-contraction}.
\hfill$\square$

\appsection{Proof of Theorem~\ref{thm:tau-consistency}}
\label{app-proof-tau-consistency}

\noindent\textit{Proof.}
Write
\[
H_i=\nabla^2\Phi\{\mu_0(s_i)\},\qquad
\Omega_i=\Omega_0(s_i),\qquad
g_i=\operatorname{tr}(H_i\Omega_i),\qquad
a_n=\frac1n\sum_{i=1}^n g_i,
\]
and let
\[
N_n=\sum_{B:m_B\geq2}(m_B-1),\qquad
T_n=\sum_{B:m_B\geq2}\sum_{i\in I_B}
(X_i-\bar X_B)^\top
\nabla^2\Phi(\bar X_B^{\mathrm{int}})
(X_i-\bar X_B).
\]
Thus $\widehat a_n=T_n/N_n$ when $N_n>0$, with the convention $\widehat a_n=0$ when $N_n=0$. On the occupancy event $\mathcal A_n$ from Section~\ref{app-spatial-comparator}, whose $P_0$-probability tends to one, there are constants $0<c_B<C_B<\infty$ such that
$c_BL_n\leq m_B\leq C_BL_n
$ simultaneously for all $B\in\mathcal B_K$. Since $L_n\to\infty$, all blocks have $m_B\geq2$ for sufficiently large $n$ on $\mathcal A_n$, and therefore
\[
N_n=n-K^2,\qquad
\frac{N_n}{n}\longrightarrow1,
\]
because $K^2L_n\asymp n$.

Let $\mathcal I_n$ be the blocks contained in a single true region and let $\mathcal C_n=\mathcal B_K\setminus\mathcal I_n$. Assumption~\ref{assump-spatial} implies
\[
|\mathcal C_n|=O(K),\qquad
\sum_{B\in\mathcal C_n}m_B=O(KL_n)=o(n)
\]
on $\mathcal A_n$. For $B\in\mathcal I_n$, let $\ell(B)$ denote its true-region index and define
\[
H_B^0=\nabla^2\Phi(\mu_{\ell(B),0}),\qquad
Q_B^0=\sum_{i\in I_B}(X_i-\bar X_B)^\top H_B^0(X_i-\bar X_B).
\]
Conditional independence and the common conditional mean within an interior block give
$E_0(Q_B^0\mid s_{1:n})
=
\frac{m_B-1}{m_B}\sum_{i\in I_B}g_i.$
The local exponential-moment condition implies uniformly bounded conditional fourth moments, and hence
$\operatorname{Var}_0(Q_B^0\mid s_{1:n})\leq Cm_B.$

Since different blocks involve disjoint collections of conditionally independent responses,
$\frac1n\sum_{B\in\mathcal I_n}
\{Q_B^0-E_0(Q_B^0\mid s_{1:n})\}
\xrightarrow{P_0}0.$
Moreover, the Hessian bound and the conditional second-moment bound imply $0\leq g_i\leq C_g$ uniformly, so
$0\leq
\sum_{B\in\mathcal I_n}\frac1{m_B}\sum_{i\in I_B}g_i
\leq C_gK^2=o(n).$
Consequently,
$\frac1n
\left\{
\sum_{B\in\mathcal I_n}Q_B^0
-
\sum_{B\in\mathcal I_n}\sum_{i\in I_B}g_i
\right\}
\xrightarrow{P_0}0.$

For the interior blocks, the exponential-moment condition and $m_B\gtrsim L_n$, together with a union bound over at most $K^2$ blocks, give
$\max_{B\in\mathcal I_n}
\|\bar X_B-\mu_{\ell(B),0}\|
\xrightarrow{P_0}0.$
Since $\nabla^2\Phi$ is uniformly continuous on the compact set $\mathcal K$,
\[
\max_{B\in\mathcal I_n}
\left\|
\nabla^2\Phi(\bar X_B)-H_B^0
\right\|_{\mathrm{op}}
\xrightarrow{P_0}0.
\]
Also,
\[
\sum_{B\in\mathcal I_n}\sum_{i\in I_B}
\|X_i-\bar X_B\|^2
\leq
\sum_{i=1}^n\|\epsilon_i\|^2
=
O_{P_0}(n),
\qquad
\epsilon_i=X_i-\mu_0(s_i),
\]
so replacing $H_B^0$ by $\nabla^2\Phi(\bar X_B)$ changes the preceding interior sum by $o_{P_0}(n)$.

For every nonempty block let
$\bar\mu_B=\frac1{m_B}\sum_{i\in I_B}\mu_0(s_i)\in\mathcal H_0.$
On $\mathcal A_n$, the same concentration argument gives
$\max_{B\in\mathcal B_K}
\|\bar X_B-\bar\mu_B\|
\xrightarrow{P_0}0.$
Hence, with $P_0$-probability tending to one,
\[
\bar X_B\in\mathcal K\subset\mathcal U
\qquad\text{for every }B\in\mathcal B_K.
\]
On this event $\bar X_B^{\mathrm{int}}=\bar X_B$ for every block, so the regularized statistic $T_n$ coincides with the corresponding unregularized statistic. Furthermore,
$\|\nabla^2\Phi(\bar X_B)\|_{\mathrm{op}}
\leq \overline h_\Phi$
uniformly. Since the true regional centers form a finite set and
$\sum_{B\in\mathcal C_n}m_B=O(KL_n)$,
\[
\sum_{B\in\mathcal C_n}\sum_{i\in I_B}
(X_i-\bar X_B)^\top
\nabla^2\Phi(\bar X_B)
(X_i-\bar X_B)
=
O_{P_0}(KL_n)
=
o_{P_0}(n),
\]
and similarly
$\sum_{B\in\mathcal C_n}\sum_{i\in I_B}g_i
=
O(KL_n)
=
o(n).$
It follows that
\[
\frac{T_n}{n}-a_n\xrightarrow{P_0}0.
\]
Together with $N_n/n\xrightarrow{P_0}1$, this yields
$\widehat a_n-a_n\xrightarrow{P_0}0.$

Finally, define
$g_0(s)=
\operatorname{tr}
\left[
\nabla^2\Phi\{\mu_0(s)\}\Omega_0(s)
\right].$
The assumptions imply that $g_0$ is bounded and integrable. Since
$s_1,\ldots,s_n$ are i.i.d.,
\[
a_n=\frac1n\sum_{i=1}^n g_0(s_i)
\xrightarrow{P_0}
E_S\{g_0(S)\}
=
a(P_0).
\]
Therefore
\[
\widehat a_n\xrightarrow{P_0}a(P_0).
\]
Because $0<a(P_0)<\infty$, we have
$P_0(\widehat a_n>0)\to1$. On this event
$\widehat\tau_n=r/\widehat a_n$, while its prescribed value on
$\{\widehat a_n=0\}$ is asymptotically irrelevant. Hence, by the continuous mapping theorem,
\[
\widehat\tau_n
\xrightarrow{P_0}
\frac{r}{a(P_0)}
=
\tau(P_0).
\]
This completes the proof.
\hfill$\square$

\appsection{Additional Bregman response geometries}
\label{app-additional-geometries}
\medskip
\noindent\emph{Nonnegative count responses.}
Let $\mathcal M=\mathbb N_0^d$, let $\mathcal V=\mathbb R^d$, let $\mathcal U=(0,\infty)^d$, and take $z$ to be the identity map. Take $\Phi$ to be the generalized entropy generator on $\mathcal V$, with effective domain $\operatorname{dom}(\Phi)=[0,\infty)^d$ and $\Phi(x)=\sum_{h=1}^d\{x_h\log x_h-x_h\}$ for $x\in\operatorname{dom}(\Phi)$, with $0\log0=0$. For $\mu\in\mathcal U$, we have $\nabla_h\Phi(\mu)=\log\mu_h$ and $A_\Phi(\mu)=\sum_{h=1}^d\mu_h$. Hence $D_\Phi(x,\mu)=\sum_{h=1}^d\left\{x_h\log\left(\frac{x_h}{\mu_h}\right)-x_h+\mu_h\right\}$, the generalized Kullback--Leibler divergence. Under the parameterization $\chi_h=a_h-1$ and $\kappa=b$, the prior family in \eqref{eq-bregman-prior-family} is a product of independent Gamma distributions with shape $a_h$ and common rate $b$. Specifically, if $a_h>0$ for all $h$ and $b>0$, then $\Pi_{\chi,\kappa}=\prod_{h=1}^d\operatorname{Gamma}(a_h,b)$ and $Z_\Phi(\chi,\kappa)=\prod_{h=1}^d\Gamma(a_h)b^{-a_h}$. Therefore, under independent priors $\mu_h\sim\operatorname{Gamma}(a_{0h},b_0)$, the generalized posterior is $\mu_h\mid X_C\sim\operatorname{Gamma}\left(a_{0h}+\tau\sum_{i\in C}X_{ih},b_0+\tau n_C\right)$, independently for $h=1,\ldots,d$. The collapsed cluster score is $M_\tau(C)=\exp\left[-\tau\sum_{i\in C}\sum_{h=1}^d\{X_{ih}\log X_{ih}-X_{ih}\}\right]\prod_{h=1}^d\frac{\Gamma\left(a_{0h}+\tau\sum_{i\in C}X_{ih}\right)b_0^{a_{0h}}}{\Gamma(a_{0h})(b_0+\tau n_C)^{a_{0h}+\tau\sum_{i\in C}X_{ih}}}$.

\medskip
\noindent\emph{Network-valued responses.}
Consider simple undirected binary networks on a common set of $q$ labeled vertices, and let $m=\frac{q(q-1)}{2}$. Let $\mathcal M$ denote the corresponding network space, let $\mathcal V=\mathbb R^m$, let $\mathcal U=(0,1)^m$, and let $z$ map a network to the vector of its upper-triangular adjacency indicators. Take $\Phi$ to be the negative Bernoulli entropy generator on $\mathcal V$, with effective domain $\operatorname{dom}(\Phi)=[0,1]^m$ and $\Phi(x)=\sum_{h=1}^m\{x_h\log x_h+(1-x_h)\log(1-x_h)\}$ for $x\in\operatorname{dom}(\Phi)$, with $0\log0=0$. For $\mu\in\mathcal U$, we have $\nabla_h\Phi(\mu)=\log\left\{\frac{\mu_h}{1-\mu_h}\right\}$ and $A_\Phi(\mu)=-\sum_{h=1}^m\log(1-\mu_h)$. Hence $D_\Phi(x,\mu)=\sum_{h=1}^m\left[x_h\log\left\{\frac{x_h}{\mu_h}\right\}+(1-x_h)\log\left\{\frac{1-x_h}{1-\mu_h}\right\}\right]$. Under the parameterization $\chi_h=a_h-1$ and $\kappa=a_h+b_h-2$, when $a_h+b_h$ is common across $h$, the prior family in \eqref{eq-bregman-prior-family} is a product of independent Beta distributions. In particular, under independent $\operatorname{Beta}(a_0,b_0)$ priors for the edge probabilities, the generalized posterior is $\mu_h\mid X_C\sim\operatorname{Beta}\left(a_0+\tau\sum_{i\in C}X_{ih},b_0+\tau\left\{n_C-\sum_{i\in C}X_{ih}\right\}\right)$. Since $\Phi(X_i)=0$ for binary adjacency vectors, the collapsed cluster score is $M_\tau(C)=\prod_{h=1}^m\frac{B\left(a_0+\tau\sum_{i\in C}X_{ih},b_0+\tau\left\{n_C-\sum_{i\in C}X_{ih}\right\}\right)}{B(a_0,b_0)}$.
\appsection{Posterior computation details}
\label{app-posterior-computation}

Let $(T,F)$ denote the current tree--cut state, with
$k=|F|+1$, and let $r_b(k)$, $r_d(k)$, $r_c(k)$, and $r_h(k)$ denote the
probabilities of proposing a birth, death, change, and tree-hyper move,
respectively. We take $r_b(k_{\max})=0$ and $r_d(1)=r_c(1)=0$.

\medskip
\noindent\emph{Birth move.}
If $k<k_{\max}$, select uniformly one of the $K^2-k$ uncut edges in
$E(T)\setminus F$ and delete it. Thus $F'=F\cup\{e\}$ and the number of
clusters increases from $k$ to $k+1$.

\medskip
\noindent\emph{Death move.}
If $k>1$, select uniformly one of the $k-1$ deleted edges in $F$ and
restore it. Thus $F'=F\setminus\{e\}$ and the number of clusters decreases
from $k$ to $k-1$.

\medskip
\noindent\emph{Change move.}
Select uniformly one deleted edge $e_1\in F$ and restore it, giving an
intermediate state with $k-1$ clusters. Then select uniformly one of the
$K^2-k+1$ uncut edges of the intermediate tree and delete it. The resulting
state has $k$ clusters, so the change move modifies the partition without
changing its number of clusters.

\medskip
\noindent\emph{tree-hyper move.}
Conditional on the current partition $\pi=\{D_1,\ldots,D_k\}$, the
spanning tree is resampled exactly while preserving $\pi$, following
\citet{bhattacharyya2026bayesian}. Let $G_{K,j}$ be the subgraph of $G_K$
induced by the mesh blocks in $D_j$, and let $G_K/\pi$ be the quotient
multigraph obtained by contracting each $D_j$ to one vertex while retaining
the multiplicity of edges between regions. Independently sample a uniform
spanning tree within each $G_{K,j}$ and a uniform spanning tree of
$G_K/\pi$. Their union gives a new spanning tree $T'$ compatible with
$\pi$, and $F'$ is given by the $k-1$ edges of $T'$ connecting distinct
regions. This is a draw from the exact conditional distribution of $T$
given $\pi$ under the uniform spanning-tree prior.

\medskip
\noindent\emph{Metropolis--Hastings ratios.}
For the birth, death, and change moves, write the acceptance probability as
\[
\min\{1,R_{\mathrm{prior}}R_{\mathrm{prop}}R_M\},
\]
where $R_{\mathrm{prior}}$ is the prior ratio,
$R_{\mathrm{prop}}$ is the proposal ratio, and $R_M$
is the ratio of collapsed Bregman scores.

For a birth move from $k$ to $k+1$ clusters,
\[
R_{\mathrm{prior}}
=
\frac{\lambda_n}{k+1}
\frac{\binom{K^2-1}{k-1}}{\binom{K^2-1}{k}}
=
\frac{\lambda_n}{k+1}\frac{k}{K^2-k},
\]
and
\[
R_{\mathrm{prop}}
=
\frac{r_d(k+1)/k}{r_b(k)/(K^2-k)}
=
\frac{r_d(k+1)}{r_b(k)}
\frac{K^2-k}{k}.
\]
If the deleted edge splits an observation cluster $C$ into
$C_1$ and $C_2$, then
\[
R_M
=
\frac{M_\tau(C_1)M_\tau(C_2)}{M_\tau(C)}.
\]
Hence the birth acceptance probability is
\[
\alpha_b
=
\min\left\{
1,
\frac{\lambda_n}{k+1}
\frac{r_d(k+1)}{r_b(k)}
\frac{M_\tau(C_1)M_\tau(C_2)}{M_\tau(C)}
\right\}.
\]

For a death move from $k$ to $k-1$ clusters,
\[
R_{\mathrm{prior}}
=
\frac{k}{\lambda_n}
\frac{\binom{K^2-1}{k-1}}{\binom{K^2-1}{k-2}}
=
\frac{k}{\lambda_n}
\frac{K^2-k+1}{k-1},
\]
and
\[
R_{\mathrm{prop}}
=
\frac{r_b(k-1)/(K^2-k+1)}{r_d(k)/(k-1)}
=
\frac{r_b(k-1)}{r_d(k)}
\frac{k-1}{K^2-k+1}.
\]
If the restored edge merges observation clusters $C_1$ and $C_2$ into
$C=C_1\cup C_2$, then
\[
R_M
=
\frac{M_\tau(C)}{M_\tau(C_1)M_\tau(C_2)}.
\]
Hence the death acceptance probability is
\[
\alpha_d
=
\min\left\{
1,
\frac{k}{\lambda_n}
\frac{r_b(k-1)}{r_d(k)}
\frac{M_\tau(C)}{M_\tau(C_1)M_\tau(C_2)}
\right\}.
\]

For a change move,
\[
R_{\mathrm{prior}}=1,
\qquad
R_{\mathrm{prop}}=1.
\]
Let $\mathcal A_{\mathrm{old}}$ and $\mathcal A_{\mathrm{new}}$ denote the
observation clusters modified by the move before and after the proposal.
All unchanged cluster scores cancel, giving
\[
R_M
=
\frac{
\displaystyle\prod_{C\in\mathcal A_{\mathrm{new}}}M_\tau(C)
}{
\displaystyle\prod_{C\in\mathcal A_{\mathrm{old}}}M_\tau(C)
},
\qquad
\alpha_c=\min\{1,R_M\}.
\]

For the tree-hyper move, the partition and all collapsed scores remain
unchanged. Since the new tree is sampled from its exact conditional
distribution given the current partition, the move is a Gibbs update and
is accepted with probability one.
\end{document}